\documentclass[11pt]{article}
\renewcommand{\baselinestretch}{1.4}
\usepackage{amsmath, amssymb, amsthm, latexsym,color, natbib}
\usepackage{graphicx, rotating}
\usepackage{fancyhdr}
\usepackage{algpseudocode}
\usepackage{pdflscape,url}
\makeatletter
\usepackage{xr} 
 \@addtoreset{equation}{section}
 \renewcommand{\theequation}{\thesection.\arabic{equation}}
\newtheorem{definition}{Definition}[section]
\newtheorem{theorem}{Theorem}
\newtheorem{lemma}{Lemma} 

\newtheorem{proposition}{Proposition}

\renewcommand{\hat}{\widehat}
\def\singlespace{\def\baselinestretch{1}\@normalsize}

\def\wh{\widehat}
\def\wt{\widetilde}
\def\askip{\vspace{0.1in}}

\newcommand{\cor}{{\rm Corr}}
\newcommand{\cov}{{\rm Cov}}
\newcommand{\diag}{{\rm diag}}

\newcommand{\var}{{\rm Var}}

\def\la{\lambda}

\newcommand{\ve}{{\varepsilon}}

\newcommand{\bA}{{\mathbf A}}
\newcommand{\bB}{{\mathbf B}}

\newcommand{\bL}{{\mathbf L}}
\newcommand{\bM}{{\mathbf M}}
\newcommand{\bQ}{{\mathbf Q}}

\newcommand{\bU}{{\mathbf U}}

\newcommand{\bW}{{\mathbf W}}
\newcommand{\bX}{{\mathbf X}}
\newcommand{\bY}{{\mathbf Y}}
\newcommand{\bZ}{{\mathbf Z}}

\newcommand{\bc}{{\mathbf c}}

\newcommand{\bz}{{\mathbf z}}

\newcommand{\bOmega}{\boldsymbol{\Omega}}

\newcommand{\bSigma}{\boldsymbol{\Sigma}}

\newcommand{\bgamma}{\boldsymbol{\gamma}}

\newcommand{\bTheta} {\boldsymbol{\Theta}}
\newcommand{\bPhi} {\boldsymbol{\Phi}}
\newcommand{\bPsi} {\boldsymbol{\Psi}}

\newcommand{\bGamma} {\boldsymbol{\Gamma}}
\newcommand{\bLambda} {\boldsymbol{\Lambda}}
\newcommand{\bC}{{\mathbf C}}
\newcommand{\bD}{{\mathbf D}}

\newcommand{\calJ}{{\mathcal J}}
\newcommand{\calL}{{\mathcal L}}

\newcommand{\Z}{\EuScript Z}

\def\6bullets{\bullet\bullet\bullet\bullet\bullet\bullet}

\DeclareMathAlphabet\EuScriptBF{U}{eus}{b}{n}

\newcommand{\eulbE}{\EuScriptBF E}

\newcommand{\Date}[1]{\def\@Date{#1}}
\def\today{\number\day~\ifcase\month\or
 January\or February\or March\or April\or May\or June\or
 July\or August\or September\or October\or November\or December\fi~\number\year}

\def \b1{{\bf 1}}
\def \bY{{\bf Y}}

\def\mM{\mathbb{M}}

\def\b{{\bf b}}

\def\c{{\bf c}}

\def\D{{\bf D}}
\def\L{{\bf L}}
\def\N{{\bf N}}
\def\W{{\bf W}}

\def\bO{{\bf O}}

\def\Y{{\bf Y}}
\def\Z{{\bf Z}}
\def\z{{\bf z}}
\DeclareMathOperator*{\argmax}{arg\,max} 

\def\wh{\widehat}

\newcommand{\ignore}[1]{}{}

\begin{document}
\title{\bf Autoregressive Networks
}
\author{ 
\normalsize Binyan Jiang\\
\normalsize Department of Applied Mathematics, Hong Kong Polytechnic University\\
\normalsize Hong Kong, by.jiang@polyu.edu.hk
\and 
\normalsize Jialiang Li\\
\normalsize Department of Statistics and Applied Probability, National University of Singapore\\
\normalsize Singapore,
stalj@nus.edu.sg
\and
\normalsize Qiwei Yao\\
\normalsize Department of Statistics,
London School of Economics,
London, WC2A 2AE\\
\normalsize  United Kingdom,
q.yao@lse.ac.uk
}
\maketitle

\begin{abstract}
	We propose a first-order autoregressive (i.e. AR(1)) model for dynamic 
network processes in which edges change over time while nodes remain
unchanged. The model depicts
the dynamic changes explicitly. It also facilitates simple and
efficient statistical inference methods including a permutation test
for diagnostic checking for the fitted network models.
The proposed model can be applied to the network processes with
various underlying structures but with independent edges. As an
illustration, an AR(1)
stochastic block model has been investigated in depth, which characterizes
the latent communities by the transition probabilities over time.
This leads to a new and more effective spectral clustering algorithm for identifying
the latent communities.
We have derived a finite sample condition under which the
perfect recovery of the community structure can be
achieved by the newly defined spectral clustering algorithm.
Furthermore the inference for a change point is incorporated into the
AR(1) stochastic block model to cater for possible structure changes. 
We have derived the
explicit error rates for the maximum likelihood
estimator of the change-point.
Application with three real data sets illustrates both relevance and
usefulness of the proposed AR(1) models and the associate inference methods.
\end{abstract}


\section{Introduction}

	Understanding and being able to model the network changes over time are of
immense importance for, e.g., monitoring anomalies in internet traffic
networks, predicting demand and setting prices in electricity supply
networks, managing natural resources in environmental
readings in sensor networks, and understanding how news and opinion
propagates in online social networks.
In spite of the existence of a large body
of literature on dynamic networks,
the development of the foundation for dynamic network models is
still in its infancy
\citep{ek17}.
As for dealing with dynamic changes of networks,
early attempts
are based on the evolution analysis of network snapshots over
time 
\citep{as14, dh18}.
Although this reflects the fact that most networks
change slowly over time,
it provides little insight on the dynamics
underlying the changes and is almost powerless for future prediction.
The popular approaches for modelling dynamic changes include,
among others, Markov process models \citep{snijder2005, 
	len18}, the exponential random graph models \citep{hfx10, kh14}, and
latent process based models \citep{friel2016,dd16,mm17}. The estimation 
for those models is compute-intensive, relying on various MCMC or 
EM algorithms.

In this paper we propose a simple first-order autoregressive (i.e. AR(1)) model
for dynamic network processes of which the edges changes over time while
the nodes are unchanged. 
Though our setting is a special case of the Markov chain network models
(see \cite{ybm15}, and also \cite{snijder2005} and \cite{len18}), a
simple AR(1) structure makes it possible
to measure explicitly the underlying dynamic properties
such as autocorrelation coefficients, and the Hamming distance.
It facilitates the maximum likelihood estimation (MLE) 
in a simple and direct manner with uniform error rates.
Furthermore diagnostic checking for the fitted network models can
be performed in terms of
an easy-to-use permutation test, which is impossible under a merely
Markovian structure.

Our setting can be applied to any network processes with
various underlying structures as long as the edges are independent with 
each other, which we illustrate through an AR(1) stochastic
block model. 
The latent communities in our setting are characterized by the transition
probabilities over time, instead of the (static) connection probabilities -- the
approach often adopted from static stochastic block models;
see \cite{p19} and the references therein. This new structure 
also paves the way for a new spectral clustering algorithm which identifies
the latent communities more effectively -- a phenomenon corroborated
by both the asymptotic theory and the simulation results. To cater for possible
structure changes of underlying processes, we incorporate a change point detection
mechanism in the AR(1) stochastic block modeling. Again the change point is estimated
by the maximum likelihood method. 
The AR(1) continuous time stochastic
block model of \cite{len18} is based on a sophisticated construction. Its
estimation is based on a reversible
jump MCMC, though
a discrete-time version of their model admits the same Markov Chain
representation (\ref{b3}) below.

Theoretical developments for dynamic stochastic block models in the literature
were typically based on the assumption that networks observed at different times
are independent; see \cite{p19,bbm18}
and references therein. The autoregressive structure considered
in this paper brings the extra complexity due to serial dependence.
By establishing the $\alpha$-mixing property with exponentially decaying
coefficients for the AR(1) network processes,
we are able to show that
the proposed spectral clustering algorithm leads to a consistent
recovery of the latent community structure.
On the other hand,   
an extra challenge in detecting a change point in the
dynamic stochastic block
network process is that the estimation for latent community structures
before and after a possible change point is typically not consistent
during the search for the change point. 
To overcome this obstacle, we introduce a truncation technique which breaks the 
searching interval into two parts such that the error bounds for 
the estimated change point can be  established. 

The proposed methods in this paper only apply to the dynamic networks observed on
discrete times. Even so the relevant
literature is large, across mathematics,
computer science, engineer, statistics, biology, genetics and social
sciences.
We can only list a small selection of more statistics-oriented papers in
addition to the aforementioned references.
\cite{fsx09} proposed a state space mixed membership stochastic block
model (with a logistic normal prior).
\cite{crane16} studied the limit properties of
Markovian, exchangeable and c\`adl\`ag (i.e. every edge remains in each
state which it visits for a positive amount of time) dynamic network. 
\cite{p19} studied the theoretical properties (such as the minimax lower
bounds for the risk) of a dynamic stochastic block model, assuming `smooth'
connectivity probabilities.
The literature on change point detection in dynamic networks includes
\cite{yczgj11, wyr18, wsw19, zcl19, bbm18, zlycc19}.
\cite{knn16, zpllw17, zwwh19, cfz20, zhpw20} adopted autoregressive models
for modelling continuous responses observed from the nodes of a network process.
\cite{kgk17} used dynamic network as a tool to model non-stationary
vector autoregressive processes.
For the development on continuous-time dynamic networks, we refer readers to
\cite{snijder2005},
\cite{matias2018}, \cite{len18} and \cite{corn2018}. 

The new contributions of this paper include: (i) We propose
a new and simple AR(1) model for edge dymanics (see (\ref{b1}) below),
which facilitates the easy-to-use inference
methods including a permutation test for model diagnostic checking.
(ii) The AR(1) setting  can be applied to various network processes
with specific underlying structures such as dynamic stochastic block models,
as illustrated in Section \ref{sec3} below, and also dynamic dot product
model, dynamic graphon model, etc.
(iii) The AR(1) structure also makes it possible to develop the theoretical guarantees for the serial
dependent network processes. For example, based on a concentration
inequality, we have derived a finite sample condition, under which the
perfect recovery of the community structure can be achieved by the newly
defined spectral clustering algorithm  (Theorems \ref{DKthm} 
and \ref{clusteringCons}  in Section \ref{sec321} below).    Furthermore, we have
shown that the MLE for the change-point in the
AR(1) stochastic block process is consistent with explicit error rates
(Theorem \ref{Thm_CP} in Section \ref{sec33} below).
Those results are based on some rigorous technical development for
the dependent network processes.
Note that both \cite{p19} and \cite{bbm18} assume that networks observed at
different times are independent with each other in their asymptotic theories
for dynamic stochastic block models.
Illustration with the three real network data sets indicates convincingly
that the proposed AR(1) model and the associated inference methods are
practically relevant and fruitful.

The rest of the paper is organized as follows. 
A general framework of AR(1)
network processes, the probabilistic properties, and the MLE are presented 
in Section 2.
It also contains a new and easy-to-use permutation test for
the diagnostic checking for the fitted network models. Section 3 deals
with AR(1) stochastic block models. The asymptotic
theory is developed for the new spectral clustering algorithm based on the transition
probabilities. Further extension of both the inference method and the asymptotic theory
to the setting with a change point is established. Simulation results are
reported in Section 4,
and the illustration with three real dynamic network data sets is
presented in Section 5. All technical proofs are  relegated to the Appendix.

\section{Autoregressive network models}
	\label{sec2}

\subsection{AR(1) models}
Let $\{ \bX_t, t=0, 1, 2, \cdots\}$ be a dynamic network process defined on the
$p$ fixed nodes, denoted by $\{ 1, \cdots, p\}$, 
where $\bX_t \equiv (X_{i,j}^t) $ denotes the $p\times p$ adjacency matrix
at time $t$. 
We also assume that all networks are Erd\"os-Renyi in the
sense that $X_{i,j}^t$, $(i,j) \in \calJ$, are independent and
take values either 1 or 0, where
$\calJ = \{ (i,j): 1\le i \le j \le p\}$ for undirected networks, $\calJ
= \{ (i,j): 1\le i < j \le p\}$ for undirected networks without selfloops,
$\calJ = \{ (i,j): 1\le i, j \le p\}$ for directed networks, and $\calJ 
= \{ (i,j): 1\le i \ne j \le p\}$ for directed networks without selfloops.
Note that an edge from node $i$ to $j$ is indicated by $X_{i,j}=1$, and no edge
is denoted by $X_{i,j}=0$. 
For undirected networks, $X_{i,j}^t = X_{j,i}^t$.

	\askip

\noindent
\begin{definition} \label{def1}
	An AR(1) network process is defined as
	\begin{equation} \label{b1}
	X^t_{i,j} \;=\; X^{t-1}_{i,j}\, I( \ve_{i,j}^t=0) \;+\; I (\ve_{i,j}^t=1), \quad t\ge 1,
	\end{equation}
	where $I(\cdot)$ denotes  the indicator function,   the innovations
	$\ve_{i,j}^t$, $(i,j) \in \calJ$, are independent, and
	\begin{equation} \label{b2}
	P(\ve_{i,j}^t =1) = \alpha_{i,j}^t, \quad
	P(\ve_{i,j}^t = -1) = \beta_{i,j}^t, \quad
	P(\ve_{i,j}^t =0) = 1 - \alpha_{i,j}^t - \beta_{i,j}^t.
	\end{equation}
	In the above expression, $\alpha_{i,j}^t, \beta_{i,j}^t$ are non-negative
	constants, and $\alpha_{i,j}^t + \beta_{i,j}^t \le 1$.
\end{definition}

\askip

Equation (\ref{b1}) is an analogue of the noisy network
model of \cite{cky18}. The innovation (or noise)
$\ve_{i,j}^t$ is `added' via the two indicator functions 
to ensure that $X_{i,j}^t$ is still binary.
Obviously, $\{ \bX_t, t=0, 1, 2, \cdots\} $ is a Markov chain, and
\begin{equation} \label{b3}
P(X_{i,j}^t =1 | X_{i,j}^{t-1}=0) = \alpha_{i,j}^t, \quad
P(X_{i,j}^t =0 | X_{i,j}^{t-1}=1) = \beta_{i,j}^t,
\end{equation}
or collectively,
\begin{align} \label{b4}
& P(\bX_t | \bX_{t-1}, \cdots, \bX_0)  \;=\; P(\bX_t | \bX_{t-1})
\;=\; \prod_{(i,j)\in \calJ} P(X_{i,j}^t |X_{i,j}^{t-1})\\
= \; & \prod_{(i,j)\in \calJ} (\alpha_{i,j}^t)^{X_{i,j}^t(1-X_{i,j}^{t-1})}
(1-\alpha_{i,j}^t)^{(1-X_{i,j}^t)(1-X_{i,j}^{t-1})}
(\beta_{i,j}^t)^{(1-X_{i,j}^t)X_{i,j}^{t-1}}
(1- \beta_{i,j}^t)^{X_{i,j}^t X_{i,j}^{t-1}}.
\nonumber
\end{align}
It is clear that the smaller $\alpha_{i,j}^t$ is, the more likely  the no-edge
status at time $t-1$ (i.e. $X_{i,j}^{t-1}=0$) 
will be retained at time $t$ (i.e. $X_{i,j}^{t}=0$);
and the smaller $\beta_{i,j}^t$ is, the more likely an edge at time $t-1$
(i.e. $X_{i,j}^{t-1}=1$) will be retained at time $t$ (i.e. $X_{i,j}^{t}=1$).
For most slowly changing networks (such as social networks), we expect 
$\alpha_{i,j}^t$ and $\beta_{i,j}^t$ to be small.

It is natural to model dynamic networks by a Markov chain. See, e.g.
\cite{hfx10, kh14, ybm15, friel2016, crane16, mm17, rastelli2017, len18}. For
example, the Markovian transition probabilities under
a discrete version of
the stationary independent arcs network model of Snijders (2005, Section 5)
can be written equivalently as
(2.3) with $\alpha^t_{i,j} \equiv \alpha$ and $\beta^t_{i,j}\equiv\beta$.
In this paper we build the Markovian structure based on
the explicit AR(1) model (\ref{b1}), which enables us to study the
theoretical properties
of the network processes, and to develop simple and efficient inference
methods with appropriate theoretical guarantee.

	\subsection{Stationarity}
Note that $\{\bX_t\}$ is a homogeneous Markov chain if 
\begin{equation} \label{b5}
\alpha_{i,j}^t \equiv \alpha_{i,j} \quad {\rm and} \quad
\beta_{i,j}^t \equiv \beta_{i,j} \;\;\; {\rm for \; all\;} t\ge 1 \; \; {\rm and} \;\; (i,j)\in \calJ.
\end{equation}
Specify the distribution of the initial network $\bX_0 = (X_{i,j}^0)$ as 
follows:
\begin{equation} \label{b6}
P(X_{i,j}^0 = 1) = \pi_{i,j} = 1 - P(X_{i,j}^0=0),
\end{equation}
where $\pi_{i,j} \in (0, 1)$, $(i,j) \in \calJ$, are constants. 

\askip

\begin{proposition} \label{thm1} 
	Let the homogeneity condition (\ref{b5}) hold with
	$\alpha_{i,j} + \beta_{i,j}\in (0, 1]$, and 
	\begin{equation} \label{b7}
	\pi_{i,j} = \alpha_{i,j}/(\alpha_{i,j}+ \beta_{i,j}), \;\; (i,j) \in \calJ.
	\end{equation}
	Then $\{ \bX_t, t=0, 1, 2, \cdots \}$ is a strictly stationary process.
	Furthermore for any $(i,j), (\ell, m) \in \calJ$ and $t, s \ge 0$,
	\begin{equation} \label{b8}
	E(X_{i,j}^t)= {\alpha_{i,j} \over \alpha_{i,j} + \beta_{i,j}}, \qquad
	\var(X_{i,j}^t) = {\alpha_{i,j} \beta_{i,j} \over (\alpha_{i,j} +
		\beta_{i,j})^2}, 
	\end{equation}
	\begin{equation} \label{b9}
	\rho_{i,j}(|t-s|) \equiv
	\cor(X_{i,j}^t, X_{\ell m}^{s}) =
	\Big\{
	\begin{array}{ll}
	(1 - \alpha_{i,j} - \beta_{i,j})^{|t-s|} \quad & {\rm if} \; (i,j)=(\ell, m),\\
	0 & {\rm otherwise}.
	\end{array}
	\end{equation}
\end{proposition}

\askip

The Hamming distance
counts the number of different edges in the two networks, and is 
a measure the closeness of two networks 
\citep{dh18}.

\askip
\begin{definition} \label{def2}
	For any two matrices $\bA = (A_{i,j})$ and $\bB=(B_{i,j})$ of the same size,
	the Hamming distance is defined as
	$
	D_H(\bA, \bB) \; =\; \sum_{i,j} I(A_{i,j} \ne B_{i,j}).
	$
\end{definition}

\begin{proposition} \label{thm2}
	Let $\{\bX_t, t=0, 1, \cdots\}$ be a stationary network process
	satisfying the condition of Proposition \ref{thm1}. Let
	$d_H(|t-s|) = E\{ D_H(\bX_t , \bX_s) \}$ for any $t, s \ge 0$.
	Then $d_H(0) = 0$, and it holds  for any $k\ge 1$  that
	\begin{align} \label{b10}
	d_H(k) & \;=\;
	d_H(k-1) +
	\sum_{(i,j)\in \calJ} {2\alpha_{i,j} \beta_{i,j} \over 
		\alpha_{i,j} + \beta_{i,j}} (1-\alpha_{i,j} - \beta_{i,j})^{k-1}\\ &
	\; = \sum_{(i,j)\in \calJ} {2\alpha_{i,j} \beta_{i,j} \over
		(\alpha_{i,j} + \beta_{i,j})^2} \{ 1 - (1 - \alpha_{i,j} - \beta_{i,j})^k \}.
	\label{b11}
	\end{align}
\end{proposition}

\askip

Proposition \ref{thm2} indicates that the expected Hamming distance
$d_H(d)= E\{D_H(\bX_t, \bX_{t+k})\}$ increases strictly, as $k$
increases, initially from $d_H(1) =\sum {2 \alpha_{i,j}
	\beta_{i,j} \over \alpha_{i,j}+
	\beta_{i,j} }$ towards the limit $d_H(\infty) = \sum {2 \alpha_{i,j}
	\beta_{i,j} \over (\alpha_{i,j}+
	\beta_{i,j})^2 }$ which
is also the expected Hamming distance of the two independent
networks sharing  the same marginal distribution
of $\bX_t$.

Proposition \ref{amixing} below shows that   
$\{\bX_t, \, t=0, 1, \cdots \}$ is $\alpha$-mixing with  exponentially
decaying coefficients.
Note that the conventional
mixing results for ARMA processes do not apply here, as they typically require that
the innovation distribution is continuous; see, e.g., Section 2.6.1 of \cite{fy03}. 
Let ${\mathcal F}_a^b$ be the $\sigma$-algebra generated by $\{X_{i,j}^k, a\leq k \leq b\}$.
The $\alpha$-mixing coefficient of process $\{ X_{i,j}^t, t=0,1, \cdots\}$ is defined as
\[
\alpha^{i,j}(\tau)=\sup_{k\in {\mathbb N}} \sup_{A\in {\mathcal F}_0^k, B\in {\mathcal F}_{k+\tau}^\infty, }|P(A\cap B)-P(A)P(B)|.
\]

\begin{proposition}\label{amixing}
	Let condition (\ref{b5}) hold, $\alpha_{i,j}, \beta_{i,j} >0$, and
	$\alpha_{i,j} + \beta_{i,j} \le 1$. Then
	$\alpha^{i,j}(\tau)\leq \rho_{i,j}(\tau)= (1-\alpha_{i,j}-\beta_{i,j})^\tau$
	for any $\tau \ge 1$.
\end{proposition}

\subsection{Estimation}

	To simplify the notation, we assume the availability of the
observations $\bX_0, \bX_1, \cdots, \bX_n$ from a stationary
network process which satisfies  the condition of Proposition \ref{thm1}.
Without imposing any further structure on the model, the 
parameters $(\alpha_{i,j}, \beta_{i,j})$, for different $(i,j)$,
can be estimated separately.
Condition  on $\bX_0$, the
maximum likelihood estimators are
\begin{equation} \label{b12}
\wh \alpha_{i,j} =
{ \sum_{t=1}^n X^t_{i,j}(1- X_{i,j}^{t-1}) \over
	\sum_{t=1}^n (1- X_{i,j}^{t-1})}, \qquad
\wh \beta_{i,j} =
{ \sum_{t=1}^n (1- X_{i,j}^t)X_{i,j}^{t-1} 
	\over \sum_{t=1}^n X_{i,j}^{t-1}}.
\end{equation}
See (\ref{b4}). For definiteness we shall set $0/0=1$.
To state the asymptotic properties, we list some regularity conditions first.
\begin{itemize}
	\item[C1.]	There exists a constant $l>0$ such that $\alpha_{i,j},\beta_{i,j} \ge
	l$ and $ \alpha_{i,j}+\beta_{i,j}\leq 1$ for all $(i,j)\in {\mathcal J}$.
	
	\item[C2.] $n, p\rightarrow \infty$, and $(\log n)(\log\log n)\sqrt{\frac{\log p}{n}}\rightarrow 0$.
\end{itemize}	
Condition C1 defines the parameter space, and condition C2  indicates
that the number of nodes is allowed to diverge in a smaller order than
$\exp\big\{\frac{n}{(\log n)^2(\log\log n)^2 } \big\}$.

\begin{proposition}\label{uniformCon}
	Let conditions (2.5), C1 and C2 hold.
	Then it holds that
	\[
	\max_{(i,j)\in {\mathcal J}}|\wh{\alpha}_{i,j}-\alpha_{i,j}|
	=O_p\left(\sqrt{\frac{\log p}{n}}\right) ~~ {\rm and}~~  \max_{(i,j)\in
		{\mathcal J}}|\wh{\beta}_{i,j}-\beta_{i,j}| =O_p\left(\sqrt{\frac{\log
			p}{n}}\right).
	\] 
\end{proposition}
Proposition \ref{uniformCon} provides a uniform convergence rate for the MLEs
in \eqref{b12}. 
To state the joint asymptotic normality.
Let ${ J}_1=\{ (i_1, j_1),\ldots, (i_{m_1}, j_{m_1})\}$, ${ J}_2=\{
(k_1, \ell_1),\ldots, (k_{m_2}, \ell_{m_2})\}$ be two arbitrary  subsets
of ${\mathcal J}$ with  $m_1, m_2 \ge 1$ fixed. 
Denote 
$\bTheta_{J_1,J_2}=( \alpha_{i_1, j_1},\ldots,
\alpha_{i_{m_1},j_{m_1}}, $ $\beta_{k_1, \ell_1},\ldots,  \beta_{k_{m_2},
	\ell_{m_2}})^\top,$ and  correspondingly denote the MLEs as 
$\wh{\bTheta}_{J_1,J_2}=( \wh\alpha_{i_1, j_1},\ldots,
\wh\alpha_{i_{m_1},j_{m_1}},$ $ \wh\beta_{k_1, \ell_1},\ldots,  \wh\beta_{k_{m_2}, \ell_{m_2}})^\top$.  

\begin{proposition}\label{CLT1}
	Let  conditions (2.5), C1 and C2 hold. 
	Then
	$
	\sqrt{n}(\wh{\bTheta}_{J_1, J_2}-\bTheta_{J_1,J_2}) \rightarrow N({\bf 0}, \bSigma_{J_1, J_2}),
	$
	where $\bSigma_{J_1,J_2}={\rm diag}(\sigma_{11}, \ldots,\sigma_{m_1+m_2, m_1+m_2} )$ is a diagonal matrix with 
	\begin{eqnarray*}
		&&	\sigma_{rr} = \frac{\alpha_{i_r,j_r}  (1-\alpha_{i_r,j_r}) (\alpha_{i_r,j_r}+\beta_{i_r,j_r}) }{ \beta_{i_r,j_r}},~~1\leq r \leq m_1, \\
		&&	\sigma_{rr} = \frac{\beta_{k_r,\ell_r}  (1-\beta_{k_r,\ell_r}) (\alpha_{k_r,\ell_r}+\beta_{k_r,\ell_r}) }{ \alpha_{k_r,\ell_r}},~~m_1+1\leq r \leq m_1+m_2. \\
	\end{eqnarray*}
\end{proposition}

\subsection{Model diagnostic check}
Based on estimators in (\ref{b12}), we define  `residual' $\wh\ve_{i,j}^t$, resulted
from fitting model (\ref{b1}) to the data, as the estimated value of
$E(\ve_{i,j}^t|X^t_{i,j}, X^{t-1}_{i,j})$,
i.e.
\begin{align*}
\wh\ve_{i,j}^t    = &
{\wh \alpha_{i,j} \over 1 - \wh \beta_{i,j}}
I(X^t_{i,j}=1, X^{t-1}_{i,j}=1) - {\wh \beta_{i,j} \over 1 - \wh \alpha_{i,j}}
I(X^t_{i,j}=0, X^{t-1}_{i,j}=0)\\
+ &
I(X^t_{i,j}=1, X^{t-1}_{i,j}=0) - I(X^t_{i,j}=0, X^{t-1}_{i,j}=1),
\quad (i,j) \in \calJ, \; t=1, \cdots, n.
\end{align*}
One way to check the adequacy of the model is to test for
the independence of $\wh\eulbE_t \equiv (\wh\ve_{i,j}^t)$ for $t=1, \cdots, n$.
Since $\wh\ve_{i,j}^t$, $t=1, \cdots, n$, only take 4 different values for each 
$(i,j) \in \calJ$, we adopt the two-way, or three-way
contingency table to test the independence of $\wh\eulbE_t$ and $\wh\eulbE_{t-1}$,
or $\wh\eulbE_t, \wh\eulbE_{t-1}$ and $\wh\eulbE_{t-2}$. For example the test statistic
for the two-way contingency table is
\begin{equation} \label{b13}
T = {1 \over n \, |\calJ|} \sum_{(i,j) \in \calJ} \sum_{k, \ell=1}^4 \{ n_{i,j}(k, \ell) - n_{i,j}(k,
\cdot) 
n_{i,j}(\cdot, \ell)/(n-1) \}^2 \big/ \{ n_{i,j}(k, \cdot)
n_{i,j}(\cdot, \ell)/(n-1)\},
\end{equation}
where $|\calJ|$ denotes the cardinality of $\calJ$, and
for $1 \le k, \ell  \le 4,$
\begin{align*}
n_{i,j}(k, \ell) & =\sum_{t=2}^n I\{ \wh\ve_{i,j}^t = u_{i,j}(k),\;
\wh\ve_{i,j}^{t-1} =u_{i,j}( \ell)\},\\
n_{i,j}(k, \cdot) & = \sum_{t=2}^n I\{ \wh\ve_{i,j}^t = u_{i,j}(k)\}, \quad
n_{i,j}(\cdot, \ell) = \sum_{t=2}^n I\{ \wh\ve_{i,j}^{t-1}=u_{i,j}(\ell)\}.
\end{align*}
In the above expressions,
$u_{i,j}(1)=-1, u_{i,j}(2)= - {\wh \beta_{i,j} \over 1 - \wh \alpha_{i,j}}, u_{i,j}(3) = {\wh \alpha_{i,j} \over 1 - \wh \beta_{i,j}}$ and $
u_{i,j}(4) = 1$.    
We calculate the $P$-values of the
test $T$ based on the following permutation algorithm:
\begin{enumerate}
	\item
	Permute $\wh\eulbE_1, \cdots, \wh\eulbE_n$ to obtain a 
	new sequence $\eulbE_1^\star, \cdots, \eulbE_n^\star$. Calculate
	the test statistic $T^\star$ in the same manner as $T$ with $\{ \wh\eulbE_t \}$
	replaced by $\{ \eulbE_t^\star\}$. 
	\item
	Repeat 1 above $M$ times, obtaining permutation test statistics $T^\star_j, \;
	j=1, \cdots, M$, where $M>0$ is a large integer. The $P$-value of the test 
	(for rejecting the stationary AR(1) model) is then
	\[
	{1 \over M} \sum_{j=1}^M I(T < T^\star_j).
	\]
\end{enumerate}

	\section{Autoregressive stochastic block models} \label{sec3}

The general setting in Section \ref{sec2} may apply to various
network processes with some specific underlying structures as long as the
edges are independent with each other.
In this section we illustrate the idea with a new dynamic stochastic
block (DSB) model.

\subsection{Models}

The DSB networks are undirected (i.e. $X_{i,j}^t \equiv X_{j,i}^t$) with no self-loops
(i.e. $X_{i,i}^t\equiv 0$).
Most available DSB models assume  that the networks observed
at different times are independent \citep{p19, bbm18} or conditionally independent
\citep{xh14, dd16, mm17} as connection probabilities and node memberships
evolve over time. We take a radically different approach as we impose autoregressive
structure (\ref{b1}) in the network process. Furthermore,
instead of assuming that the members in the same communities share the same
(unconditional) connection probabilities,
we entertain the idea that the transition probabilities (\ref{b3}) for the members
in the same communities are the same. This reflects more directly the dynamic behavior
of the process, and implies the former assumption on the unconditional
connection probabilities under the stationarity. See (\ref{b7}).
Furthermore, since the information on both $\alpha_{i,j}$ and $ \beta_{i,j}$,
instead of that on $ \pi_{i,j}= \alpha_{i,j}/(\alpha_{i,j} + \beta_{i,j})$ only,
will be used in estimation, we expect that the new approach leads to
more efficient estimation. This is confirmed by both the theory (Theorem
\ref{DKthm} and also Remark 3 below) and the
numerical experiments (Section \ref{sec42} below).

Let $\nu_t$ be the membership function at time $t$, i.e. for any $1\le i \le p$,
$\nu_t(i)$ takes an integer value between 1 and $q$ $(\le p)$; indicating that
node $i$ belongs to the $\nu_t(i)$-th community at time $t$, where $q$ is a
fixed integer. This effectively assumes that the $p$ nodes are divided
into the $q$ communities.
We assume that $q$ is fixed though some communities may
contain no nodes at some times.  

\begin{definition} \label{def3}
	An AR(1) stochastic block network process $\{\bX_t= (X_{i,j}^t), \,
	t=0,1,2, \cdots \}$ is defined
	by (\ref{b1}), where for $1\le i < j \le p $,
	\begin{align} \label{c1}
	P(\ve_{i,j}^t =1) & = \alpha_{i,j}^t = \theta^t_{\nu_t(i), \nu_t(j)}, \quad
	P(\ve_{i,j}^t =-1) =\beta_{i,j}^t= \eta^t_{\nu_t(i), \nu_t(j)}, \\ \nonumber
	P(\ve_{i,j}^t =0)& =1- \alpha_{i,j}^t - \beta_{i,j}^t = 
	1- \theta^t_{\nu_t(i), \nu_t(j)} - \eta^t_{\nu_t(i), \nu_t(j)}.
	\end{align}
	In the above expressions, $\theta^t_{k, \ell}, \eta^t_{k, \ell}$ are non-negative 
	constants, and $\theta^t_{k, \ell}+\eta^t_{k, \ell}\le 1$ for all $1\le k
	\le \ell \le q$.
\end{definition}

The evolution of membership process $\nu_t$ and/or the
connection probabilities was often assumed to be driven by some latent (Markov)
processes. The statistical inference for those models
is carried out using computational Bayesian
methods such as MCMC or EM. See, for example, \cite{yczgj11, xh14, dd16,
	mm17, rastelli2017}.
\cite{bbm18} adopted a change point approach: assuming both the membership and   
the connection probabilities remain constants either before or after a change point.
See also \cite{len18, wsw19}. This reflects the fact that many networks (e.g. social
networks) hardly change, and a sudden change is typically triggered by some
external events.

We adopt a change point approach in this paper. Section \ref{sec32}  
considers the estimation for both the community
membership and transition probabilities when there are no change points
in the process.
This will serve as a building block for the inference 
with a change point in Section \ref{sec33}.
Note that detecting change points in dynamic networks is a surging research area. 
In addition to the aforementioned references,
more recent developments include 
\cite{wyr18, zlycc19}. Also note that the method of
\cite{zcl19} can be applied to detect multiple change points for any dynamic
networks.

\subsection{Estimation without change points} \label{sec32}

	We first consider a simple scenario of no change points in the observed period, i.e.
\begin{equation} \label{c2}
\nu_t(\cdot) \equiv \nu(\cdot) \quad {\rm and} \quad
(\theta^t_{k, \ell}, \eta^t_{k, \ell}) \equiv (\theta_{k, \ell}, \eta_{k, \ell}),
\quad t=1, \cdots, n, \;\; 1\le k \le \ell \le q.
\end{equation}
Then fitting the DSB model consists of two steps: (i) estimating
$\nu(\cdot)$ to cluster the $p$ nodes into
$q$ communities, and (ii) estimating transition probabilities
$\theta_{k, \ell}$ and $\eta_{k, \ell}$ for $1\le k \le \ell \le q$.
To simplify the presentation, $q$ is assumed to be known, which is the assumption
taken by most papers on change point detection for DSB networks. 
In practice, one can determine $q$ by, for example, the jittering method of
\cite{cky20}, or a Bayesian information criterion; see an example in
Section \ref{sec52} below.

\subsubsection{Why does it work?} \label{sec320}
We first provide a theoretical underpinning (Proposition \ref{OracleSC}
below) on identifying the latent communities
based on $\alpha_{i,j}$ and $\beta_{i,j}$.
The stochastic block model with $p$ nodes and $q$ communities can be
parameterized by a pair of matrices $(\bZ, \bOmega)$,
where $\bZ = (z_{i,j})\in \{0,1\}^{p\times q}$ is the membership matrix
such that it has exactly one 1 in each row and at least one 1 in each
column, and $\bOmega=(\omega_{k,\ell})_{q\times q}\in [0,1]^{q\times q}$
is a symmetric and full rank connectivity matrix, with
$\omega_{k,\ell}=\frac{\theta_{k ,\ell}}{\theta_{k, \ell}+\eta_{k,
		\ell}}$.  Then 
$z_{i,j}=1$ if and only if the $i$-th node belongs to the $j$-th community.
On the other hand, $\omega_{k,\ell}$ 
is the connection probability between the nodes in community $k$ and the
nodes in community $\ell$, and 
$s_k\equiv \sum_{i=1}^pz_{i,k}$ is the size of community $k\in\{1,\ldots, q\}$. 
Clearly, matrix $\bZ$ and function $\nu(\cdot)$ are the two equivalent
representations for the community membership of the network nodes. 

Let $\W=\bZ \bOmega \bZ^\top$. Under  model \eqref{c2}, the marginal edge formation probability is given as $E(\bX_t)=\W-\diag(\W)$.
Define 
\begin{align*}\label{Omega}
&\bOmega_1=(\theta_{k,\ell})_{q\times q},\qquad \bOmega_2=(1-\eta_{k,\ell})_{q\times q}, \\
&\W_{1}=\bZ \bOmega_1 \bZ^\top=(\alpha_{i,j})_{p\times p},\qquad \W_{2}=\bZ \bOmega_2 \bZ^\top=(1-\beta_{i,j})_{p\times p}, 
\end{align*}
where $\alpha_{i,j} = \theta_{\nu(i), \nu(j)}, \;
\beta_{i,j} = \eta_{\nu(i), \nu(j) }$.
Then $\W_{1}-\diag(\W_{1})$ can be viewed as the edge formation
probability matrix of the latent noise process
$\ve_{i,j}^{1,t}\equiv I(\ve_{i,j}^t=1)$.
Furthermore, under 
model \eqref{c2}, the latent network process $\{
(\ve_{i,j}^{1,t})_{1\leq i,j\leq p}, t=0,1,2,\ldots\}$ has the
same membership structures as $\{\bX_t= (X_{i,j}^t), \,
t=0,1,2, \cdots \}$.  
Since $X_{i,j}^t=1$ is implied by $\ve_{i,j}^t=1$
under model \eqref{b1},
the elements in   
$\W_{1}-\diag(\W_{1})$ are thus positively correlated with the
elements in $\W-\diag(\W)$.
Similarly 
$\W_{2}-\diag(\W_{2})$ can be viewed as the edge formation probability
matrix of the latent noise process $\ve_{i,j}^{2,t}\equiv I(\ve_{i,j}^t\neq
-1)$, and $\{ (\ve_{i,j}^{2,t})_{1\leq i,j\leq p}, t=0,1,2,\ldots\}$ has
the same membership structure as $\{\bX_t, \,
t=0,1,2, \cdots \}$.  Since $\ve_{i,j}^t=-1$ implies $X_{i,j}^t=0$, the elements in   
$\W_{2}-\diag(\W_{2})$  are  also positively correlated to those in $\W-\diag(\W)$.  
Let ${\D}_1$ and ${\D}_2$ be two $p\times p$ diagonal matrices with, respectively, $d_{i,1}, d_{i,2}
$ as their $(i,i)$-th elements,
where 
\[
{d}_{i,1} = \sum_{j=1}^p   \alpha_{i,j}, \qquad  
d_{i,2} = \sum_{j=1}^p (1-  \beta_{i,j}).
\]
The normalized Laplacian matrices based on $\W_1$ and $\W_2$ are then defined as: 
\begin{equation}\label{c0}  
\L_1 = \D^{-1/2}_1 \W_1 \D^{-1/2}_1, \qquad
\L_2 = \D^{-1/2}_2 \W_2 \D^{-1/2}_2, \qquad
\L=\L_1+\L_2. 
\end{equation} 
Correspondingly, let ${\bf s}=  ( {s_1},\ldots,  {s_q})^\top$. We denote the degree corrected connectivity matrices as
\begin{equation*}  
\wt\bOmega_1 = 
\wt\D^{-1/2}_1 \bOmega_1 \wt\D^{-1/2}_1, \qquad
\wt\bOmega_2 = \wt\D^{-1/2}_2 \bOmega_2 \wt\D^{-1/2}_2, \qquad
\wt\bOmega=\wt\bOmega_1+\wt\bOmega_2, 
\end{equation*} 
where $\wt\D_1=\diag\{\bOmega_1 {\bf s}\}$ and $\wt\D_2=\diag\{\bOmega_2 {\bf s}\}$ are diagonal matrices with, respectively, the degrees of the nodes in the $i$-th community corresponding to $\bOmega_1$ and $\bOmega_2$ as their $(i,i)$-th elements.
The following lemma shows that  the  block structure in the membership matrix $\bZ$ can be recovered by the leading eigenvectors of $\L$. 

\begin{proposition}\label{OracleSC}
	Suppose $\wt\bOmega$ is full rank, we have, $rank(\L)=q$. 
	Let $\bGamma_q\bLambda\bGamma_q^\top$ be the eigen-decomposition of $\L$,
	where $\bLambda 
	=\diag\{\lambda_1,\ldots,
	\lambda_q\}$ is the diagonal matrix consisting of the nonzero
	eigenvalues of $\L$ arranged in the order $|\lambda_1|\ge \cdots \ge |\lambda_q|>0$.
	There exists a matrix $\bU\in {\cal R}^{q\times q}$ such that 
	$\bGamma_q=\Z\bU$. Furthermore, for any $1\leq i,j\leq p$,
	$\bz_{i}\bU=\bz_{j}\bU$ if and only if $\bz_{i}=\bz_{j}$, where $\bz_{i}$
	denotes the $i$-th row of $\Z$.

\end{proposition}

\noindent
{\bf Remark 1}.
The $q$ columns of $\bGamma_q$ are the orthonormal eigenvectors  of $\bL$
corresponding to the $q$ non-zero eigenvalues.
Proposition \ref{OracleSC} implies that there are only $q$ distinct rows
in the $p\times q$ matrix $\bGamma_q$, and two nodes belong to a same
community if and only if
the corresponding rows in $\bGamma_q$ are the same. 
Intuitively the discriminant power of $\bGamma_q$
can be understood as follows.
For any unit vector $\bgamma = (\gamma_1, \cdots, \gamma_p)^\top$,
\begin{align} \label{c5}
\bgamma^\top \bL \bgamma = 2  -  \sum_{1\le i <j \le p}
\alpha_{i,j} \Big({\gamma_i\over \sqrt{ d_{i,1}}} - {\gamma_j \over
	\sqrt{ d_{j,1}}} \Big)^2
-
\sum_{1\le i <j \le p} (1 -
\beta_{i,j}) \Big({\gamma_i\over \sqrt{ d_{i,2}}} - {\gamma_j \over \sqrt{d_{j,2}}} \Big)^2.
\end{align}
For $\bgamma$ being an eigenvector corresponding to the positive eigenvalue of $\bL$,
the sum of the 2nd and the 3rd terms on
the RHS (\ref{c5}) is minimized.
Thus $|\gamma_i - \gamma_j|$ is small when $\alpha_{i,j}$ and/or $(1- \beta_{i,j})$
are large; noting that $ d_{i,k}= d_{j, k}$ for $k=1, 2$ when nodes $i$
and $j$ belong to the same community.
The communities in a network are 
often formed in the way that the members within the same community
are more likely to be connected with each other, and the members
belong to different communities are unlikely or less likely to be connected.
Hence when nodes $i$ and $j$ belong to the same community,
$\alpha_{i,j}$ tends to be large and $\beta_{i,j}$ tends to be small
(see (\ref{b3})). The converse 
is true when the two nodes belong to two different communities.
The eigenvectors corresponding to negative eigenvalues are capable to
identify the so-called heterophilic communities, see pp.1892-3 of \cite{rohe2011spectral}.

\subsubsection{Estimating membership $\nu(\cdot)$}
\label{sec321}

	It follows from Proposition \ref{thm1}, (\ref{c1}) and (\ref{c2}) that
\[
P(X_{ij}^t =1) = \theta_{\nu(i), \nu(j)} / (\theta_{\nu(i), \nu(j)}
+ \eta_{\nu(i), \nu(j)}) \equiv \omega_{\nu(i), \nu(j)}, \quad 1 \le i<j \le p,
\]
provided that $X_{ij}^0$ is initiated with the same marginal distribution.
A simple approach adopted in literature is to apply a
community detection method for static stochastic block models using 
the averaged data $\bar \bX = \sum_{1\le t \le n} \bX_t/n$ to detect
the latent communities characterized by the connection probabilities
$\{ \omega_{k, \ell}, 1\le k \le \ell \le q\}$. We take a different approach
based on estimators $\{ (\wh \alpha_{i,j}, \wh \beta_{i,j}), 1 \le i
<j \le p\}$ defined in (\ref{b12})
to identify the clusters determined by the transition probabilities
$\{ (\theta_{k, \ell}, \eta_{k, \ell}), 1\le k \le \ell \le q\}$  instead.
More precisely, we propose a new
spectral clustering algorithm  to estimate
$\bGamma_q$ specified in Proposition \ref{OracleSC} above.

Let $\wh\bW_1, \wh\bW_2$ be two $p\times p$ matrices with, respectively, $\wh\alpha_{i,j},
(1- \wh\beta_{i,j})$ as their $(i,j)$-th elements for $i\ne j$, and 0 on the main
diagonals. Let $\wh\bD_1, \wh\bD_2$ be two $p\times p$
diagonal matrices with, respectively, $\wh d_{i,1},\wh d_{i,2}
$ as their $(i,i)$-th elements,
where 
\[
\wh d_{i,1} = \sum_{j=1}^p \wh \alpha_{i,j}, \qquad  
\wh d_{i,2} = \sum_{j=1}^p (1- \wh \beta_{i,j}).
\]
Define two (normalized) Laplacian matrices
\begin{equation} \label{c3}
\wh \bL_1 =\wh  \bD^{-1/2}_1\wh  \bW_1\wh  \bD^{-1/2}_1, \qquad
\wh \bL_2 =\wh  \bD^{-1/2}_2\wh  \bW_2\wh  \bD^{-1/2}_2.
\end{equation}
Perform the eigen-decomposition for the sum of $\bL_1 $ and $\bL_2$:
\begin{equation} \label{c4}
\wh\bL \equiv \wh \bL_1 +\wh \bL_2 = \wh\bGamma\, \diag(\wh\la_1, \cdots, \wh\la_p) \wh\bGamma^\top,
\end{equation}
where the eigenvalues are arranged in the order
$\wh\lambda_1^2\geq \ldots \geq \wh\lambda_p^2$, and the columns of the
$p\times p$ orthogonal matrix $\wh \bGamma$ are the corresponding eigenvectors.
We call $\wh\la_1, \ldots, \wh\la_q$ the $q$ leading eigenvalues of $\wh\bL$. 
Denote by $\wh\bGamma_q$ the $p\times q$ matrix consisting of the first $q$ columns of
$\wh\bGamma$, which are called the leading eigenvectors of $\wh\bL$.
The spectral clustering applies the $k$-means clustering algorithm to
the $p$ rows of $\wh\bGamma_q$ to obtain the community assignments for the $p$
nodes $\wh \nu(i) \in \{ 1, \cdots, q\}$ for $i=1, \cdots, p$.

\noindent
{\bf Remark 2}.
Proposition \ref{OracleSC} implies that the true memberships can be recovered
by the $q$ distinct rows of $\bGamma_q$.
Note that 
\[\wh \L=\wh \bL_1+\wh \bL_2 \approx \L_1-\diag(\L_1)+\L_2-\diag(\L_2)=\L-\diag(\L).
\]
We shall see that the effect of the term $\diag(\L)$ on the eigenvectors $\bGamma_q$ is negligible when $p$ is large (see for example \eqref{De_diag} in the proof of Lemma \ref{Frob} in
Appendix A), and hence the rows of $\wh \bGamma_q$ should be slightly perturbed versions of the $q$ distinct rows in $\bGamma_q$.

The following theorem justified the validity of using $\wh\L$ for spectral clustering.
Note that $\|\cdot\|_2$ and $\|\cdot\|_F$ denote, respectively, the $L_2$
and the Frobenius norm of matrices.

\begin{theorem} \label{DKthm}
	Let conditions (2.5), C1 and C2
	hold, and  $\lambda_q^{-2} \left(\sqrt{\frac{\log (pn)}{np} } +\frac{1}{n}+\frac{1}{p}\right) \rightarrow 0$, as  $n, p\rightarrow \infty$.  
	Then it holds that
	\begin{eqnarray}\label{Con_EVal}
	\max_{i=1,\ldots, p} |\lambda_i^2-\wh\lambda_i^2|\leq  \|\wh\L\wh\L-\L\L\|_2\leq \|\wh\L\wh\L-\L\L\|_F=O_p\left(\sqrt{\frac{\log (pn)}{np} } +\frac{1}{n}+\frac{1}{p}\right) . 
	\end{eqnarray}
	Moreover, for any constant $B>0$, there exists a constant $C>0$ such that
	the inequality
	\begin{eqnarray}\label{Con_EVec}
	\|\wh\bGamma_q -\bGamma_q\bO_q\|_F\leq
	4\lambda_{q}^{-2}C\left(\sqrt{\frac{\log (pn)}{np} }
	+\frac{1}{n}+\frac{1}{p}\right)
	\end{eqnarray}
	holds with probability greater than
	$1-16p\left[ (pn)^{-(1+B)}+  \exp\{-B\sqrt{p}\}\right]$, where $\bO_q$
	is a $q\times q$ orthogonal matrix.
\end{theorem}

It follows from \eqref{Con_EVal} that the leading eigenvalues of $\L$
can be consistently recovered by  the
leading eigenvalues of $\wh\L$.   By \eqref{Con_EVec}, the
leading eigenvectors of $\L$ can also be consistently estimated,
subject to a rotation (due to the possible multiplicity of
some leading eigenvalues $\L$).
Proposition \ref{OracleSC} indicates that there are only $q$ distinct
rows in $\bGamma_q$, and, therefore, also $q$ distinct
rows in  $\bGamma_q\bO_q$, corresponding to the $q$ latent communities for
the $p$ nodes. This paves the way for the $k$-means algorithm stated below.
Put
\[
{\cal M}_{p,q}=\{ \bM\in {\cal R}^{p\times q}: \bM  ~{ {\rm has}~ q~ {\rm distinct ~rows}} \}.
\]

\noindent
{\bf
	The $k$-means clustering algorithm}: Let
\[
(\wh \bc_1, \cdots, \wh \bc_p)^\top = \arg \min_{\bM \in  {\cal M}_{p,q}}\|\wh\bGamma_q- \bM\|_F^2.
\]
There are only $q$ distinct vectors among $\wh \bc_1, \cdots, \wh \bc_p$, forming
the $q$ communities. Theorem  \ref{clusteringCons} below shows that they are identical to the
latent communities of the $p$ nodes under (\ref{Con_EVec}) and (\ref{KM_SN}).
The latter holds if
$
\sqrt{ \ {s_{\max} }} \lambda_{q}^{-2}C\left(\sqrt{\frac{\log (pn)}{np} }
+\frac{1}{n}+\frac{1}{p}\right)\rightarrow 0 
$, where $s_{\max}=\max\{s_1, \ldots, s_q\}$ is the size of the largest community.

\begin{theorem}\label{clusteringCons} 
	Let \eqref{Con_EVec} hold and 
	\begin{equation}\label{KM_SN}
	\sqrt{\frac{1}{s_{\max}}} >  4\sqrt{2} \lambda_{q}^{-2}C\left(\sqrt{\frac{\log (pn)}{np} } +\frac{1}{n}+\frac{1}{p}\right)
	. 	\end{equation}
	Then
	$\wh \c_i=\wh \c_j$ if and only if $\nu(i) = \nu(j)$, $1\leq i,j\leq p$. 
\end{theorem}

\noindent
{\bf Remark 3}. By Lemma A.1 of 
\cite{rohe2011spectral}, the error bound for the standard 
spectral clustering algorithm  (with $n=1$)
is $O_p\left({\frac{\log p}{\sqrt{p}} }+\frac{1}{p}
\right)$, where the term $\frac{1}{p}$ reflects the bias caused by the
inconsistent estimation of diagonal terms (see equation (A.5) and
subsequent derivations in \cite{rohe2011spectral}). This bias 
comes directly from the removal of the diagonal elements of $\L$, as
pointed out in Remark 2 above.
Although the algorithm was designed for static networks, it has often
been applied to dynamic networks using
${1\over n} \sum_t \bX_t$ in the place of
a single observed network; see, e.g. \cite{bbm18}.
With some simple modification to the proof
of  Lemma A.1 of \cite{rohe2011spectral}, it can be shown that
the error bound  is then reduced to 
\begin{equation} \label{err1}
O_p\left({\frac{\log (pn)}{\sqrt{np}} }+\frac{1}{p} \right),
\end{equation}
provided that the observed networks are i.i.d. The error would
only increase when the observations are not independent.
On the other hand, our proposed spectral clustering algorithm for
(dependent) dynamic networks
entails the error rate specified in \eqref{Con_EVal} and
\eqref{Con_EVec} which is smaller than (\ref{err1})
as long as $n$ is sufficiently large (i.e. $(p/n)^{1\over 2}/\log(np) \to 0$).
Note that we need $n$ to be large enough in relation to $p$ in
order to capture the dynamic dependence of the networks.

\subsubsection{Estimation for $\theta_{k, \ell}$ and $\eta_{k, \ell}$}
\label{se322}

	For any $1\leq k\leq \ell \leq q$, we define 
\begin{equation} \label{Skl}
S_{k,l}=
\Big\{
\begin{array}{ll}
\{(i,j) : 1\leq i\neq j\leq p, \nu(i)=k, \nu(j)=\ell\}   \quad & {\rm if} \; k\neq l,\\
\{(i,j) : 1\leq i< j\leq p, \nu(i)=k=\nu(j)=\ell\}   \quad & {\rm if} \; k= l,
\end{array}
\end{equation} 
Clearly the cardinality of $S_{k,\ell}$ is $n_{k,\ell}=s_ks_\ell$ when $k\neq \ell$ and $n_{k,\ell}=s_k(s_k-1)/2$ when $k=\ell$.

Based on the procedure presented in Section \ref{sec321},
we obtain an estimated membership function $\wh \nu(\cdot)$.
Consequently, the MLEs for $(\theta_{k, \ell}, \eta_{k, \ell})$, $1\le k \le \ell \le q$, admit the form
\begin{align} \label{c6}
\wh \theta_{k, \ell}& ={
	\sum_{(i,j)\in \wh S_{k,\ell} }
	\sum_{t=1}^n X_{i,j}^t(1 - X_{i,j}^{t-1}) \Big/ \hspace{-3mm}
	\sum_{(i,j)\in \wh S_{k,\ell} }
	\sum_{t=1}^n (1 - X_{i,j}^{t-1})},\\[1ex] \label{c7}
\wh \eta_{k, \ell}& =
\sum_{(i,j)\in \wh S_{k,\ell} }
\sum_{t=1}^n (1- X_{i,j}^t)X_{i,j}^{t-1} \Big/ \hspace{-3mm}
\sum_{(i,j)\in \wh S_{k,\ell} }
\sum_{t=1}^n X_{i,j}^{t-1},
\end{align}
where 
\begin{equation*} 
\wh S_{k,\ell}=
\Big\{
\begin{array}{ll}
\{(i,j) : 1\leq i\neq j\leq p,\; \wh\nu(i)=k,\; \wh\nu(j)=\ell\}   \quad & {\rm if} \; k\neq \ell,\\
\{(i,j) : 1\leq i< j\leq p,\; \wh\nu(i)=\wh\nu(j)=k\}   \quad & {\rm if} \; k= \ell.
\end{array}
\end{equation*} 
See (\ref{b12}) and also  (\ref{c1}).

Theorem \ref{clusteringCons} implies that the memberships of the nodes can be
consistently recovered. Consequently, the consistency and the asymptotic normality
of the MLEs $\wh \theta_{k, \ell}$ and $\wh \eta_{k, \ell}$ can be established
in the same manner as for Propositions
\ref{uniformCon} and \ref{CLT1}. We state the results below.

Let ${\cal K}_1=\{ (i_1, j_1),\ldots, (i_{m_1}, j_{m_1})\}$ and ${ \cal K}_2=\{
(k_1, \ell_1),\ldots, (k_{m_2}, \ell_{m_2})\}$ be two arbitrary subsets
of $\{{(k,\ell): \linebreak 1\leq k\leq  \ell\leq q}\}$ with 
$m_1, m_2 \ge 1$ fixed. Let $$\bPsi_{{\cal K}_1,{\cal K}_2}=( \theta_{i_1,
	j_1},\ldots, \theta_{i_{m_1},j_{m_1}}, \eta_{k_1, \ell_1},\ldots, 
\eta_{k_{m_2}, \ell_{m_2}})',$$ and let  
$\wh{\bPsi}_{{\cal K}_1,{\cal K}_2}$ denote its MLE.
Put
$\N_{{\cal K}_1,{\cal K}_2}=\diag(n_{i_1,j_1},\ldots, n_{i_{m_1},j_{m_1}},
n_{k_1,\ell_1},\ldots, n_{k_{m_2},\ell_{m_2}}  )$ where $n_{k,\ell}$ is the cardinality of $S_{k,\ell}$ defined as in \eqref{Skl}.

\begin{theorem}\label{uniformCon2}
	Let conditions (2.5), C1 and C2 hold,
	and  
	$
	\frac{\sqrt{ \ {s_{\max} }} }{\lambda_{q}^2}\left(\sqrt{\frac{\log (pn)}{np} } +\frac{1}{n}+\frac{1}{p}\right)\rightarrow 0. 
	$
	Then it holds that
	\[
	\max_{1\leq k,\ell\leq q}|\wh{\theta}_{k,\ell}-\theta_{k,\ell}| =O_p\left(\sqrt{\frac{\log q}{ns_{\min}^2}}\right) ~~ {\rm and}~~  \max_{1\leq k,\ell\leq q}|\wh{\eta}_{k,\ell}-\eta_{k,\ell}| =O_p\left(\sqrt{\frac{\log q}{ns_{\min}^2}}\right) ,
	\] 
	where $s_{\min}=\min\{s_1,\ldots, s_q\}$.
\end{theorem}

\begin{theorem}\label{CLT2}
	Let the condition of Theorem \ref{uniformCon2} hold. Then
	\[
	\sqrt{n}	\N_{{\cal K}_1,{\cal K}_2}^{\frac{1}{2}}(\wh{\bPsi}_{{\cal K}_1, {\cal K}_2}-\bPsi_{{\cal K}_1,{\cal K}_2}) \rightarrow N({\bf 0}, \wt\bSigma_{{\cal K}_1, {\cal K}_2}),
	\]
	where   $\wt\bSigma_{{\cal K}_1,{\cal K}_2}={\rm diag}(\wt\sigma_{11},
	\ldots,\wt\sigma_{m_1+m_2, m_1+m_2} )$ with 
	\begin{eqnarray*}
		&&	\wt\sigma_{rr} = \frac{\theta_{i_r,j_r}  (1-\theta_{i_r,j_r}) (\theta_{i_r,j_r}+\eta_{i_r,j_r}) }{ \eta_{i_r,j_r}},~~1\leq r \leq m_1, \\
		&&	\wt\sigma_{rr} = \frac{\eta_{k_r,\ell_r}  (1-\eta_{k_r,\ell_r}) (\theta_{k_r,\ell_r}+\eta_{k_r,\ell_r}) }{ \theta_{k_r,\ell_r}},~~m_1+1\leq r \leq m_1+m_2. \\
	\end{eqnarray*}
\end{theorem}

Finally to prepare for the inference in Section \ref{sec33} below, we
introduce some notations.
First we denote $\wh \nu$ by $\wh \nu^{1,n}$, to reflect
the fact that the community clustering was carried out using the data
$\bX_1, \cdots, \bX_n$ (conditionally on $\bX_0$). See Section \ref{sec321} above.
Further we denote the maximum log likelihood by
\begin{equation} \label{c8}
\wh l (1, n; \; \wh \nu^{1,n}) =
l( \{ \wh\theta_{k, \ell}, \wh\eta_{k, \ell}\};\; \wh \nu^{1,n})
\end{equation}
to highlight the fact that both the node clustering and the estimation for
transition probabilities are based on the data $\bX_1, \cdots, \bX_n$.

\subsection{Inference with a change point} \label{sec33}

Now we assume that there is a change point $\tau_0$ at which both
the membership of nodes and the transition probabilities 
$\{\theta_{k, \ell}, \eta_{k, \ell}\}$ change. It is necessary to
assume $n_0 \le \tau_0 \le n - n_0$, where $n_0$ is an integer
and $n_0/n \equiv c_0 >0 $ is a small constant, as we need enough
information before and after the change in order to detect 
$\tau_0$.   
We assume that within the time period $[0,\tau_0]$, the network follows a
stationary model \eqref{c2} with parameters $\{(\theta_{1,k,\ell},
\eta_{1,k,\ell}): 1\leq k, l\leq q\}$ and  a membership map
$\nu^{1,\tau_0}(\cdot)$. 
Within the time period $[\tau_0+1, n]$ the network follows a stationary
model \eqref{c2}  with parameters $\{(\theta_{2,k,\ell},
\eta_{2,k,\ell}): 1\leq k, l\leq q\}$ and a membership map
$\nu^{\tau_0+1,n}(\cdot)$. 
Though we assume that the number of communities is unchanged 
after the change point, our results can be easily extended to the case that the number of communities also changes. 

We estimate the change point $\tau_0$ by the maximum likelihood method:
\begin{equation} \label{c9}
\wh \tau = \arg \max_{n_0 \le \tau \le n - n_0} \{\, \wh l (1, \tau; \; \wh \nu^{1,\tau})
+ \wh l (\tau+1, n; \; \wh \nu^{\tau+1,n})\},
\end{equation}
where $\wh l(\cdot)$ is given in (\ref{c8}).

To measure the difference between the two sets of transition probabilities
before and after the change, we put
\[
\Delta_F^2=
{1 \over p^{2}}\big( \| \bW_{1,1}-\bW_{2,1}\|_F^2+\|
\bW_{1,2}-\bW_{2,2}\|_F^2\big),
\]
where the four $p\times p$ matrices are defined as 
\[
\bW_{1,1}=(\theta_{1,\nu^{1,\tau_0}(i),\nu^{1,\tau_0}(j)}),\quad\;
\bW_{1,2}=(1-\eta_{1,\nu^{1,\tau_0}(i),\nu^{1,\tau_0}(j)}),
\]
\[
\bW_{2,1}=(\theta_{2,\nu^{\tau_0+1,n}(i),\nu^{\tau_0}(j)+1,n}),\quad\;
\bW_{2,2}=(1-\eta_{1,\nu^{\tau_0+1,n}(i),\nu^{\tau_0+1,n}(j)}).
\]
Note that $\Delta_F$ can be viewed as the signal strength for detecting
the change point $\tau_0$. Let $s_{\max}, \; s_{\min}$ denote, respectively,
the largest, and the smallest community size among all
the communities before and after the change.
Similar to \eqref{c0}, we denote the normalized Laplacian matrices corresponding to    $\bW_{i,j}$ as $\bL_{i,j}$ for $i,j=1,2$. Let $|\lambda_{i,1}|\geq |\lambda_{i,2}|\ge \ldots\ge |\lambda_{i,q}|$
be the absolute nonzero eigenvalues of $\bL_{i,1}+\bL_{i,2}$ for $i=1,2$,
and we denote  $\lambda_{\min}=\min\{ |\lambda_{1,q}|, |\lambda_{2,q}|\}$. 
Now some regularity conditions are in order.

\begin{itemize}
	\item[C3.]	For some constant $l>0$, $ \theta_{i,k,\ell},\eta_{i,k,\ell}> l$,
	and $ \theta_{i,k,\ell}+\eta_{i,k,\ell}\leq 1$ for all $i=1,2$ and $1\leq k\leq \ell\leq q$.

	\item[C4.]  $\log(np)/\sqrt{p}\rightarrow 0$, and  $\sqrt{ \ {s_{\max} }}
	\lambda_{\min}^{-2} \big(\sqrt{{\log (pn)}/{np} } +\frac{1}{n}+\frac{1}{p}+
	\frac{     {{  \log (np )}/{n  } } +\sqrt{{  \log (np )}/{(np^2) } }   
	}{\Delta_F^2 }
	\big)\rightarrow 0$.    
	
	\item[C5.]   $\frac{\Delta_F^2 }{     
		{{  \log (np )}/{n  } } +\sqrt{{  \log (np )}/{(np^2) } }   
	}\rightarrow \infty$.

\end{itemize}	
Condition C3 is similar to C1.  
The condition ${\log(np)}/{\sqrt{p}}\rightarrow 0$ in C4 
controls the  misclassification rate of the k-means algorithm. Recall that
there is a bias term $O(p^{-1})$ in spectral clustering caused by  the
removal of the diagonal of the Laplacian matrix (see Remark 2 above). 
Intuitively,  as $p$ increases, 
the effect of this bias term on the misclassification rate of the
k-means algorithm becomes negligible. On the other hand, note that the
length of the time interval for searching for the change point in \eqref{c9}
is of order $O(n)$; the $\log(n)$ term here in some sense reflects the
effect of the difficulty in detecting the true change point when the
searching interval is extended as $n$ increases. The second condition in
C4 is similar to \eqref{KM_SN}, which ensures that the true communities
can be recovered. 
Condition C5 requires that the average signal strength $\Delta_F^2=
p^{-2}\big[ \| \bW_{1,1}-\bW_{2,1}\|_F^2+\|
\bW_{1,2}-\bW_{2,2}\|_F^2\big] $ is of higher order than
$	{\frac{  \log (np )}{n  } } +\sqrt{\frac{  \log (np )}{np^2 } } 
$ for change point detection. 

\cite{ybm15} deals with
the MLE for a change-point in a network Markov chain 
but without a latent community structure. Hence it does not have the complication
to estimate the community memberships in addition.
Allowing the membership change in our setting leads to an extra challenge: in the
process of searching for the location of the change-point, the estimation for the
latent communities before or after a specified location may not be consistent.
To overcome this obstacle, we introduce a truncation which breaks the searching
interval into two parts such that the error in the estimated change-point
can be bounded. 
\cite{bbm18} also allows the membership change. But it assumes that the
networks observed at different times are independent with each other.

\begin{theorem}\label{Thm_CP}
	Let conditions C2-C5 hold. Then the following assertions hold.
	\begin{itemize}
		\item[(i)]  When $\nu^{1,\tau_0}\equiv \nu^{\tau_0+1,n}$, 
		\begin{eqnarray*}
			\frac{|\tau_0-\wh\tau|}{n} =
			O_p\left(  \frac{    {\frac{  \log (np )}{n  } } +\sqrt{\frac{  \log (np )}{np^2 } }  }{\Delta_F^2}    \times  \min \Bigg\{1, \frac{\min\Big\{  1,  (n^{-1}p^2\log(np))^{\frac{1}{4}}\Big\} }{\Delta_F s_{\min}} \Bigg\}\right) .
		\end{eqnarray*}
		\item[(ii)] When $\nu^{1,\tau_0}\neq \nu^{\tau_0+1,n}$,  
		\begin{eqnarray*}
			\frac{|\tau_0-\wh\tau|}{n} =O_p\Bigg( \frac{    {\frac{  \log (np )}{n  } } +\sqrt{\frac{  \log (np )}{np^2 } }  }{\Delta_F^2}    \times
			\min \left\{1, \frac{\min\Big\{  1,  (n^{-1}p^2\log(np))^{\frac{1}{4}}\Big\}}{  \Delta_Fs_{\min} }   +\frac{1}{\Delta_F^2}
			\Bigg\} \right)  .
		\end{eqnarray*}
	\end{itemize}
\end{theorem}

Notice that for $\tau<\tau_0$, the observations in the time
interval $[\tau+1,n]$ are a mixture of the two different 
network processes if $\nu^{1,\tau_0}\neq \nu^{\tau_0+1,n}$.
In the worst case scenario then, all $q$ communities can be changed
after the change point $\tau_0$. This causes the extra estimation error term
$\frac{1}{\Delta_F^2}$ in Theorem \ref{Thm_CP}(ii).

\section{Simulations}

	\subsection{Parameter estimation}
We generate data according to model (\ref{b1}) in which
the parameters $\alpha_{ij}$ and $\beta_{ij}$ are drawn independently
from $U[0.1, 0.5]$, $1\le i , j \le p$.
The initial value $\bX_0$  was simulated according to (\ref{b6})
with $\pi_{ij}=0.5$.  
We calculate the estimates according to  (\ref{b12}).
For each setting (with different $p$ and $n$), we replicate the experiment
500 times.
Furthermore we also calculate the 95\% confidence intervals 
for $\alpha_{ij}$ and $\beta_{ij}$ based on the asymptotically normal
distributions specified in Proposition \ref{CLT1}, and report the relative
frequencies of the intervals covering the true values of the parameters.
The results are summarized in Table \ref{table:sim1}.

 	\begin{singlespace}
\begin{table}[htbp]
	\small
	\caption{The mean squared errors (MSE) of the estimated parameters in
		AR(1) network model (\ref{b1}) and
		the relative frequencies (coverage rates) of the event that the
		asymptotic 95\% confidence intervals cover the true values in a simulation
		with 500 replications.}
	\label{table:sim1}
	\begin{center}
		\begin{tabular}{|c|c|c|c|c|c|c|c|} 
			\hline
			&& \multicolumn{2}{c|}{$\hat{\alpha}_{i,j}$ } & \multicolumn{2}{c|}{$\hat\beta_{i,j}$ }\\ 
			\hline
			n&p& MSE & Coverage (\%) &  MSE  & Coverage (\%) \\ 
			\hline
			5&100&.130 & 39.2& .131& 39.3 \\ 
			5 & 200 & .131 & 39.3 & .131 & 39.4\\  
			20 & 100 & .038 & 86.1& .037& 86.0\\  
			20 & 200& .037 & 86.1 & .037 & 86.0\\ 
			50 & 100 & .012& 92.3& .012& 92.2\\ 
			50 & 200 & .011& 92..2& .012 & 92.2\\  
			100& 100 & .005& 93.7 & .005 & 93.8\\ 
			100& 200 & .005 & 93.8 & .005 & 93.9\\  
			200& 100& .002 & 94.5& .002 & 94.5\\  
			200& 200& .002 & 94.6& .002 & 94.5\\ 
			\hline			
		\end{tabular}
	\end{center}
\end{table}
 	\end{singlespace}

The MSE decreases as $n$ increases, showing 
steadily improvement in performance.
The coverage rates of the asymptotic
confidence intervals are very close to the nominal level when $n\ge 50$.
The results hardly change between $p=100$ and 200.

\subsection{Community Detection} \label{sec42}

We now consider model (\ref{c1}) with $q =2$ or 3 clusters, in which 
$\theta_{i,i}=\eta_{i,i}=0.4$ for $i=1,\cdots,q$, and $\theta_{i,j}$
and $\eta_{i,j}$, for $1\le i , j \le q$, are drawn independently from 
$U[0.05, 0.25]$. For each setting, we replicate the experiment 500 times.

We identify the $q$ latent communities using the newly
proposed spectral clustering algorithm
based on matrix $\wh \bL = \wh \bL_1 + \wh \bL_2 $ defined in (\ref{c4}).
For the comparison purpose, we also implement 
the standard spectral clustering method
for static networks (cf. \cite{rohe2011spectral}) but
using the average 
\begin{equation} \label{spOld}
\bar \bX = {1 \over n} \sum_{ t=1}^n \bX_t
\end{equation}
in place of
the single observed adjacency matrix. This idea has been frequently
used in spectral clustering for dynamic networks; see, for example,
\cite{wsw19, zcl19, bbm18}. 
We report the normalized mutual information (NMI) and the adjusted Rand
index (ARI): Both metrics take values between 0 and 1,
and both measure the closeness between the true 
communities and the estimated communities in the sense that
the larger the values of NMI and ARI are, the
closer the two sets of communities are; see \cite{vinh2010}.
The results are summarized in Table \ref{table:cluster}.
The newly proposed algorithm based on $\wh \bL$ always outperforms the
algorithm based on $\bar \bX$, even when $n$ is as small as 5.
The differences between the two methods are substantial in terms of the
scores of both NMI and ARI. For example when $q=2, p=100$ and $n=5$,
NMI and ARI are, respectively, 0.621 and 0.666 for the new method, and they
are merely 0.148 and 0.158 for the standard method based on $\bar \bX$.
This is due to the fact that the new method identifies the latent
communities using the information on both $\alpha_{i,j}$ and $\beta_{i,j}$
while the standard method uses the information on $\pi_{i,j} = {\alpha_{i,j}
	\over \alpha_{i,j}+\beta_{i,j}}$ only.

After the communities were identified, we 
estimate $\theta_{i,j}$ and $\eta_{i,j}$ by
(\ref{c6}) and (\ref{c7}), respectively. The mean
squared errors (MSE) are evaluated for all the parameters. The results
are summarized in Table \ref{table:mse}. 
For the comparison purpose, we also report the estimates 
based on the identified communities by the $\bar \bX$-based clustering.
The MSE values of the estimates based on the communities
identified by the new clustering method are always smaller
than those of based on $\bar \bX$.
Noticeably now the estimates with small $n$ such as $n=5$ are already reasonably
accurate, as  the information from all the nodes within the
same community is pulled together.

 \begin{singlespace}
\begin{table}[htbp]
	\centering
	\caption{Normalized mutual information (NMI) and adjusted Rand
		index (ARI) of the true communities and the estimated communities
		in the simulation with 500 replications.
		The communities are estimated by the spectral clustering algorithm (SCA) based on
		either matrix $\wh \bL$ in (\ref{c4}) or matrix $\bar \bX$ in (\ref{spOld}).}
	\label{table:cluster}
	\begin{tabular}{|c|c|c|c|c|c|c|}
		\hline
		&& & \multicolumn{2}{c|}{SCA based on $\wh \bL$}&
		\multicolumn{2}{c|}{SCA based on $\bar \bX$}\\
		\hline
		q&p&n& NMI & ARI &NMI & ARI\\
		\hline
		2& 100& 5& .621& .666& .148& .158\\
		&&20& .733 & .755 & .395 & .402\\
		&&50& .932 & .938 & .572 & .584\\
		&& 100 & .994 & .995 & .692 & .696\\
		2 & 200&5 & .808& .839& .375& .406\\
		&& 20& .850 & .857 & .569 & .589\\
		&& 50 & .949 & .953 & .712 & .722\\
		&&100& .994&.995&.790&.796\\
		\hline
		3 & 100 & 5& .542& .536& .078 & .057\\
		&&20& .686 & .678 & .351 & .325\\
		&&50 & .931 & .929 & .581 & .562\\
		&& 100 & .988 & .987 & .696 & .670 \\
		3 & 200 & 5& .729& .731& .195& .175\\
		&&20 & .779& .763 & .550 & .542\\
		&&50 & .954 & .952 & .726 & .711\\
		&&100&.994& .994& .822& .802\\
		\hline
	\end{tabular}
\end{table}
 \end{singlespace}

 	 \begin{singlespace}
\begin{table}[htbp]
	\centering
	\caption{The mean squared errors (MSE) of the estimated parameters in 
		AR(1) network models with $q$ communities. The communities are estimated
		by the spectral clustering algorithm (SCA) based on either matrix $\wh
		\bL$ in (\ref{c4}) or matrix $\bar \bX$ in (\ref{spOld}).}
	\label{table:mse}
	\begin{tabular}{|c|c|c|c|c|c|c|}
		\hline
		&& & \multicolumn{2}{c|}{SCA based on $\wh \bL$}&\multicolumn{2}{c|}{
			SCA based on $\bar \bX$}\\
		\hline
		q&p&n& $\hat{\theta}_{i,j}$& $\hat{\eta}_{i,j}$ & $\hat{\theta}_{i,j}$& $\hat{\eta}_{i,j}$ \\
		\hline
		2& 100& 5& .0149&.0170&.0298&.0312\\
		&&20& .0120 &.0141 & .0229  & .0233\\
		&&50& .0075 & .0083 & .0178 & .0177\\
		&&100& .0058& .0061& .0147& .0148\\
		2 & 200& 5& .0099&.0116&.0223&.0248\\
		&&20& .0093 & .0111 & .0219 & .0248 \\
		&& 50 & .0068 & .0073 & .0140 & .0145\\
		&&100&.0061& .0062& .0117& .0118\\
		\hline
		3 & 100 & 5& .0194& .0211& .0318& .0325\\
		&&20& .0156 & .0181 & .0251 & .0255\\
		&&50 & .0093 & .0104 & .0193 & .0193\\
		&&100& .0081& .0085& .0163& .0162\\
		3 & 200 & 5& .0143& .0162& .0287& .0301\\
		&&20 & .0134 & .0156 & .0200 & .0205\\
		&&50 & .0090 & .0093 & .0156 & .0153\\
		&&100& .0079& .0083& .0130& .0131\\
		\hline
	\end{tabular}
\end{table}
 	 \end{singlespace}

\section{Illustration with real data}
We illustrate the proposed methodology through three real data examples in this section.
More real data analysis can be found in Appendix B. 
\subsection{RFID sensors data}
Contacts between patients, patients and health care workers (HCW) and
among HCW represent one of the important routes of transmission of
hospital-acquired infections. \cite{van2013} collected records of
contacts among patients and various types of HCW in the
geriatric unit of a hospital in Lyon, France, between 1pm on
Monday 6 December and 2pm on Friday 10 December 2010. Each of the $p=75$
individuals in this study consented to wear Radio-Frequency IDentification
(RFID) sensors on small identification badges during this period, which
made it possible to record when any two of them were in face-to-face
contact with each other (i.e. within 1-1.5 meters) in every
20-second interval during the period. This data set is now available in R
packages {\tt igraphdata} and {\tt sand}.

Following \cite{van2013}, we combine together the recorded information in each
24 hours  to form 5 daily
networks ($n=5$), i.e. an edge between two individuals is equal to 1
if they made at least one contact during the 24 hours, and 0 otherwise.
Those 5 networks  are plotted in
Figure \ref{f1}.
We fit the data with stationary AR(1) model (\ref{b1}) and (\ref{b5}).
Some summary statistics of the estimated parameters, according to the 4
different roles of the individuals,  are presented
in Table \ref{rfidcoef}, together with the direct relatively frequency estimates
$\wt \pi_{i,j} = \bar X_{i,j}= \sum_{t=1}^5 X_{i,j}^t /5$.
We apply the permutation test (\ref{b13}) (with 500 permutations) to the residuals
resulted from the fitted AR(1) model. The $P$-value is $0.45$, indicating no significant
evidence against the stationarity assumption.

Since the original data were recorded for each 20 seconds, they can also be
combined into half-day series with $n=10$. Figure \ref{f10} presents
the 10 half-day networks.
We repeat the above exercise for this new sequence.
Now the $P$-value of the permutation test is $0.008$, indicating the stationary
AR(1) model should be rejected for this sequence of 10 networks. This is intuitively
understandable, as people behave differently at the different times during a day (such as
daytime or night). Those within-day nonstationary behaviour shows up
in the data accumulation over every 12 hours, and it disappears
in the accumulation over 24 hour periods.
Also overall the adjacent two networks in Figure \ref{f10} look more
different from each other than the adjacent pairs in Figure \ref{f1}.

There is no evidence  of the existence of any communities among
the 75 individuals in this data set.  Our analysis confirms this too. For
example the results of the spectral clustering algorithm based on,
respectively, $\wh \bL$ and $\bar \bX$ do not corroborate with each other
at all as the NMI is smaller than 0.1.

 \begin{singlespace}
\begin{table}[htbp]
	\centering
	\caption{Mean estimated coefficients (standard errors) for the four types of individuals in RFID data. Status codes: administrative staff (ADM), medical doctor (MED), paramedical staff, such as nurses or nurses' aides (NUR), and patients (PAT).}
	\label{rfidcoef}
	\begin{tabular}{|c|c|c|c|c|}
		\hline
		& \multicolumn{4}{c|}{$\hat{\alpha}_{ij}$}\\
		\hline
		Status & ADM& NUR& MED & PAT\\
		\hline
		ADM & .1249 (.2212) & .1739 (.2521) & .1666 (.2641) & .1113 (.2021)\\
		NUR & & .2347 (.2927)& .2398 (.3022) & .1922 (.2513)\\
		MED & & & .3594 (.3883)& .1264 (.2175)\\
		PAT &  &  &  & .0089 (.0552)\\
		\hline
		& \multicolumn{4}{c|}{$\hat{\beta}_{ij}$}\\
		\hline
		Status & ADM& NUR& MED & PAT\\
		\hline
		ADM & .1666 (.3660)& .2326 (.3883) & .2925 (.4235) & .2061 (.3798)\\
		NUR & & .3714 (.4470) & .3001 (.4167) & .3656 (.4498)\\
		MED & & & .4187 (.3973)& .2311 (.4066)\\
		PAT &  &  &  & .0198 (.1331)\\
		\hline
		& \multicolumn{4}{c|}{$\hat{\pi}_{ij}=\hat{\alpha}_{ij}/(\hat{\alpha}_{ij}+\hat{\beta}_{ij})$}\\
		\hline
		Status & ADM& NUR& MED & PAT\\
		\hline
		ADM & .2265 (.3900)& .2478 (.3672) & .1893 (.3119) & .1239 (.2490)\\
		NUR & & .2488 (.3244)& .2729 (.3491)& .2088 (.3016)\\
		MED & & & .3310 (.3674)& .1398 (.2660)\\
		PAT &  &  &  & .0124 (.0928)\\
		\hline
		& \multicolumn{4}{c|}{$\wt \pi_{i,j} =\bar{X}_{ij}$}\\
		\hline
		Status & ADM& NUR& MED & PAT\\
		\hline
		ADM & .1250 (.3312)& .1583 (.3652) & .1704 (.3764) & .0887 (.2845)\\
		NUR & & .1854 (.3887)& .1730 (.3784) & .1542 (.3612)\\
		MED & & & .3901 (.4881)& .0927 (.2902)\\
		PAT &  &  &  & .0090 (.0946)\\
		\hline
	\end{tabular}
\end{table}
 \end{singlespace}

\begin{sidewaysfigure}[htbp]
	\centering
	\includegraphics[scale=0.8]{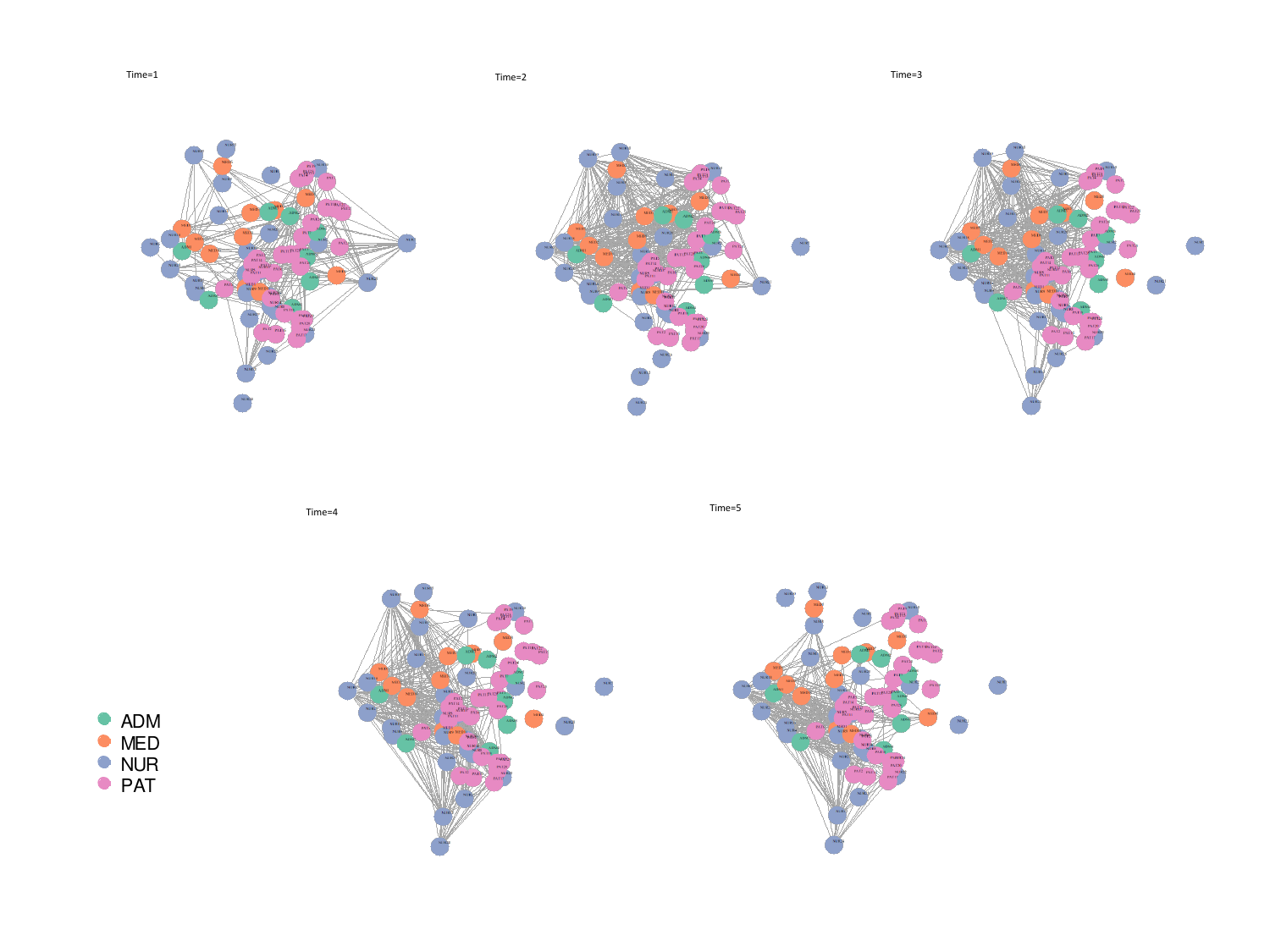}
	\caption{The RFID sensors data: the 5 networks obtained by	combining together the information		within each of the five 24-hou{\tiny {\tiny }}r periods. The four different identities of the individuals are marked in four different colours. }
	\label{f1}
\end{sidewaysfigure}


\begin{sidewaysfigure}[htbp]
	\centering
	\includegraphics[scale=0.8]{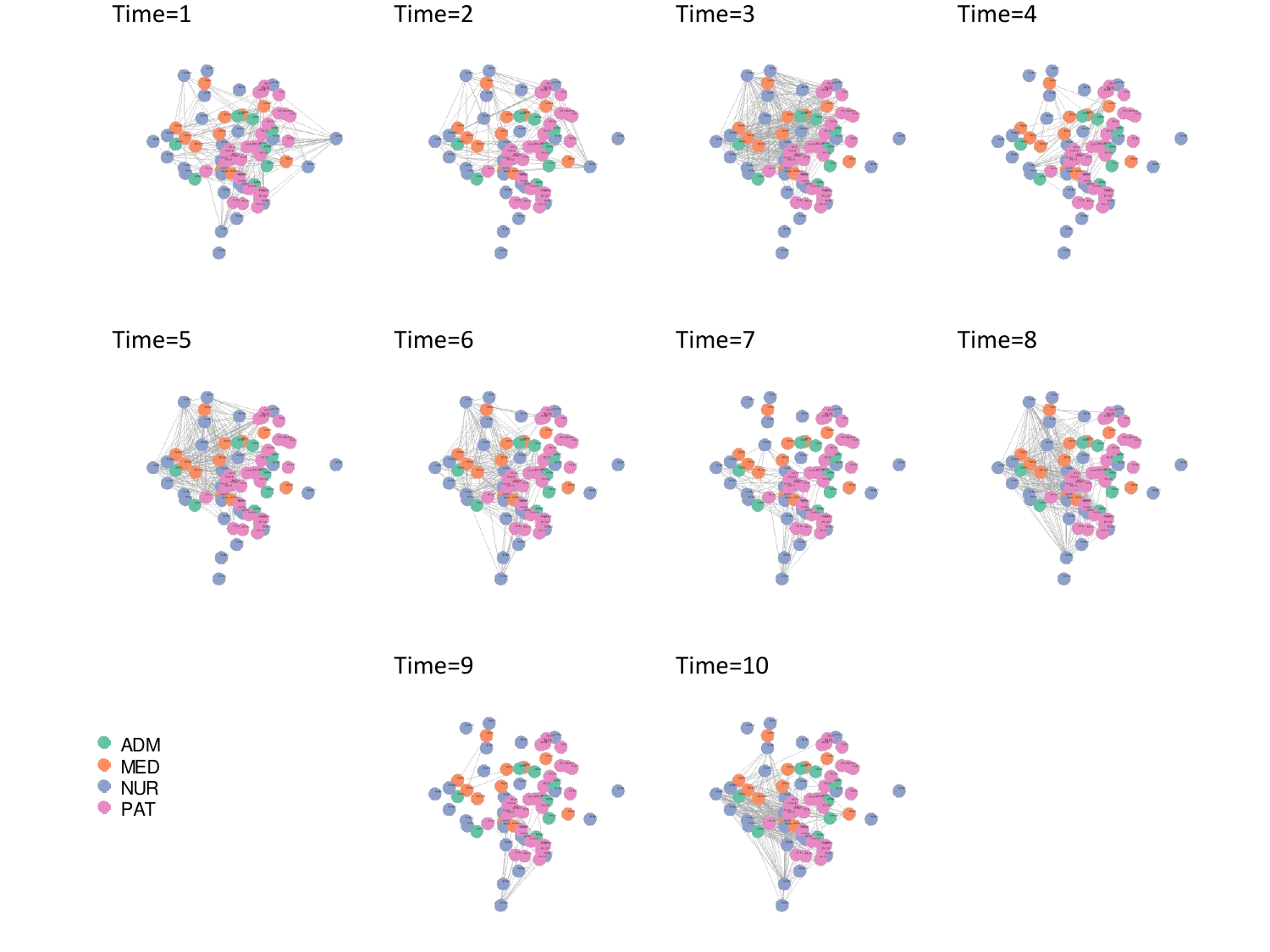}
	\caption{The RFID sensors data: the 10 networks obtained by 
		combining together the information
		within each of the ten 12-hour periods. The four different identities of the individuals 
		are marked in four different colours. }
	\label{f10}
\end{sidewaysfigure}

\subsection{French high school contact data} \label{sec52}

Now we consider a contact network data collected in a high school  in
Marseilles, France \citep{mastrand2015}. The data are the recorded 
face-to-face contacts among the students from 9 classes during $n=5$
days in December 2013, measured by the SocioPatterns infrastructure.
Those are students in the so-called {\sl classes preparatoires} -- a part of the
French post-secondary education system. We label the 3 classes
majored in mathematics and physics as MP1, MP2 and MP3, the 3 classes
majored in biology as BIO1, BIO2 and BIO3, the 2 classes majored in
physics and chemistry as PC1 and PC2, and the class majored in
engineering as EGI. 
The data are available at
\url{www.sociopatterns.org/datasets/high-school-contact-and-friendship-networks/}.
We \linebreak have removed the individuals with missing values, and include the remaining 
$p=327$ students in our clustering analysis based on the AR(1) stochastic
block network model (see Definition~\ref{def3}).

We start the analysis with $q=2$. The detected 2 clusters by the spectral
clustering algorithm (SCA) based on either $\wh \bL$ in (\ref{c4}) or $\bar \bX$ are
reported in Table~\ref{table:classtype}. The two methods lead to almost
identical results: 3 classes majored in biology are in one cluster and the other
6 classes are in the other cluster. The number of `misplaced' students is 2 and 1,
respectively, by the SCA based on $\wh \bL$ and $\bar \bX$.
Figure \ref{french} shows that the identified two clusters are clearly separated
from each other across all the 5 days.
The permutation test (\ref{b13}) on the residuals indicates that the
stationary AR(1) stochastic block network model seems to be appropriate for this
data set, as the $P$-value is 0.676.
We repeat the analysis for $q=3$,
leading to equally plausible results: 3 biology classes are
in one cluster, 3 mathematics and physics classes are in another cluster,
and the 3 remaining classes form the 3rd cluster. See also Figure \ref{frenchq3}
for the graphical illustration with the 3 clusters.

\begin{sidewaysfigure}[htbp]
	\centering
	\includegraphics[scale=0.8]{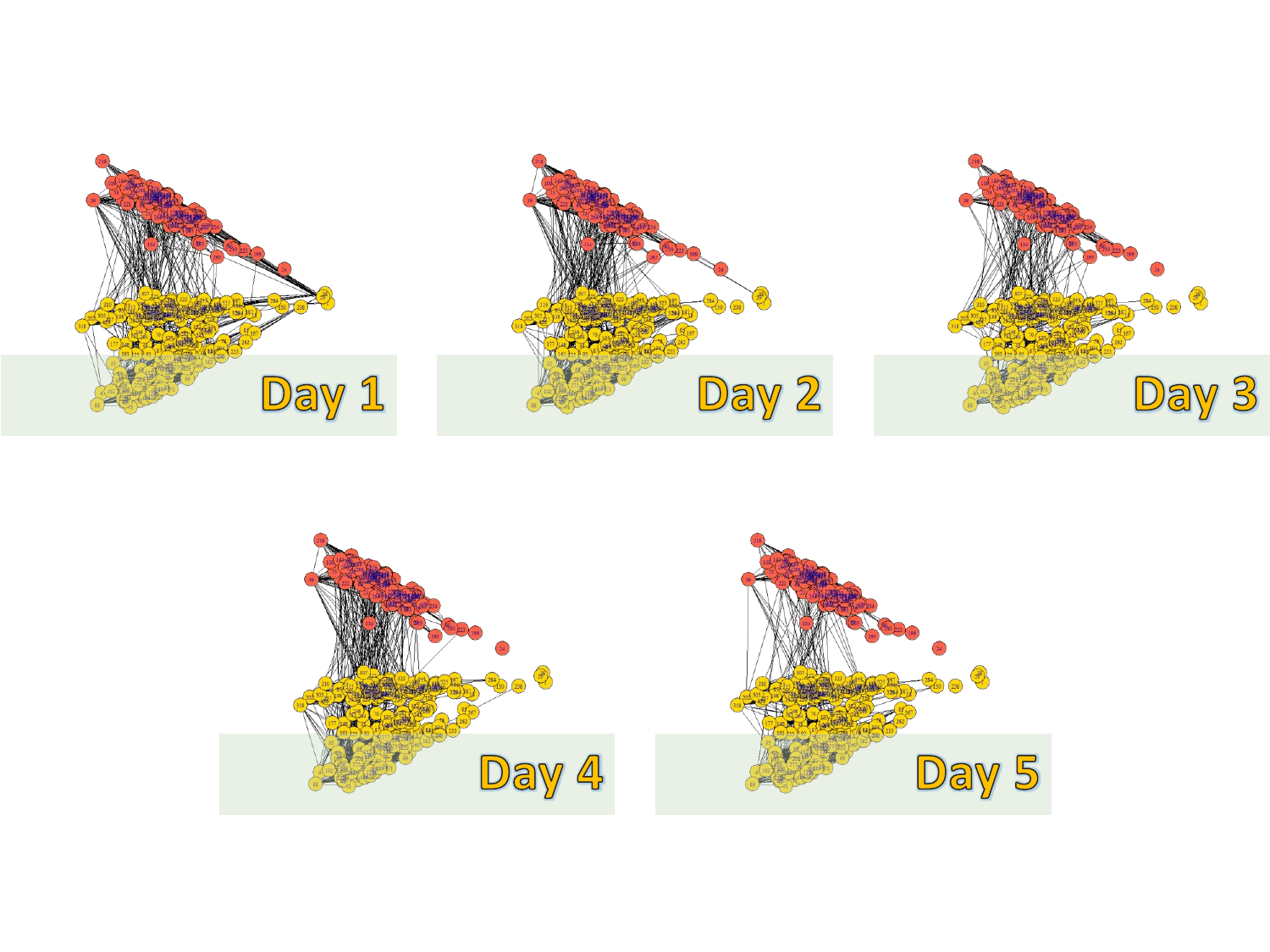}
	\caption{French high school contact networks over 5 days: 
		the nodes marked in two colours represent the $q=2$ clusters determined
		by SCA based on $\wh \bL$ in (\ref{c5}).}
	\label{french}
\end{sidewaysfigure}

 \begin{singlespace}
\begin{table}[htbp]
	\centering
	\caption{French high school contact network data: the detected
		clusters by spectral clustering algorithm (SCA) based on either $\wh \bL$ in (\ref{c5})
		or $\bar \bX$. The number of clusters is set at $q=2$.}
	\label{table:classtype}
	\begin{tabular}{|c|c|c|c|c|}
		\hline
		& \multicolumn{2}{c|}{SCA based on $\wh \bL$}&\multicolumn{2}{c|}{SCA based on $\bar X$}\\
		\hline
		Class & Cluster 1 & Cluster 2 & Cluster 1 & Cluster 2\\
		\hline
		BIO1 & 0&  37&1& 36\\
		BIO2  &1& 32&0& 33\\
		BIO3 & 1 & 39 &0& 40\\
		MP1  &  33 &  0 & 33& 0\\
		MP2 & 29 & 0 &29& 0\\
		MP3  & 38 & 0 &38& 0\\
		PC1 &  44&  0 &44&0\\
		PC2  & 39&   0&39&0\\
		EGI  &34&  0&	34&0\\
		\hline
	\end{tabular}
\end{table}
 \end{singlespace}

To choose the number of clusters $q$ objectively, we define the Bayesian information
criteria (BIC) as follows:
\begin{align*}
{\rm
	BIC}(q)=- 2\max \, \log({\rm likelihood}) +\log\{n(p/q)^2\}q(q+1).
\end{align*}
For each fixed $q$, we effectively build $q(q+1)/2$ models independently and each
model has 2 parameters $\theta_{k,\ell}$ and $\eta_{k,\ell}$, $1\le k \le \ell \le q$.
The number of the available observations for each model is approximately $n(p/q)^2$, 
assuming that the numbers of nodes in  all the $q$ clusters are about the same,
which is then  $p/q$.
Thus the penalty term in the BIC above is $\sum_{1\le k \le \ell \le q} 2
\log\{ n(p/q)^2\} = \log\{n(p/q)^2\}q(q+1)$.

Table \ref{table:aic} lists the values of BIC($q$) for different $q$.
The minimum is obtained at $q=9$,
exactly the number of original classes in the school.
Performing the SCA based on $\wh \bL$ with $q=9$, we obtain
almost perfect classification: all the 9 original classes are identified
as the 9 clusters with only in total 4 students being placed outside their
own classes. Figure \ref{frenchq7} plots the networks with the identified
9 clusters in 9 different colours.
The estimated $\theta_{i,j}$ and $\eta_{i,j}$, together
with their standard errors calculated based on the asymptotic
normality presented in Theorem~\ref{CLT2}, are reported in Table~\ref{table:french7}.
As $\wh \theta_{i,j}$ for $i\ne j$ are very small (i.e. $\le 0.027$),
the students from different classes who have not contacted with each other
are unlikely to contact next day. See (\ref{c1}) and (\ref{b3}).
On the other hand, as $\wh \eta_{i,j}$ for $i\ne j$ are large (i.e. $\ge 0.761$),
the students from different classes who have contacted with each other
are likely to lose the contacts next day. Note that $\wh \theta_{i,i}$ are 
greater than $\wh \theta_{i,j}$ for $i\ne j$ substantially, and $\wh \eta_{i,i}$ are
smaller than $\wh \eta_{i,j}$ for $i\ne j$ substantially. This implies that
the students in the same class are more likely to contact with each other
than those across the different classes.

 \begin{singlespace}
\begin{table}[htbp]
	\centering
	\caption{Fitting AR(1) stochastic block models to
		the French high school data: BIC values for different cluster numbers $q$.}
	\label{table:aic}
	\begin{tabular}{|c|c|c|c|c|c|c|c|c|}
		\hline
		$q$ & 2 & 3 & 5 & 7 & 8&9 & 10 & 11\\
		\hline
		BIC($q$) & 43624 & 40586 & 37726 & 36112& 35224 & 34943 & 35002 & 35120\\
		\hline
	\end{tabular}
\end{table}
 \end{singlespace}

To apply the variational EM algorithm of \cite{mm17} to analyze this data set,
we use the {\tt R} package {\tt dynsbm}. The algorithm is designed to identify
time-varying dynamic stochastic block structure in the sense 
that both the membership of nodes and the transition probabilities may vary
with time. Furthermore it also identifies the nodes not belonging to any clusters.
The number of the clusters selected by the so-called integrated
classification likelihood criterion is also 9. The identified 9 clusters
are always dominated by the 9 original classes in the school, though they
vary from day to day.
The number of the identified students not belonging to any of the 9 clusters was
15, 17, 24, 32 and 28, respectively, in those 5 days.
Furthermore the number of the students who were not put in their own classes was
14, 9, 12, 10 and 12, respectively. 
The more detailed
results are reported in Appendix B. 
Those findings are less clear-cut than those obtained from
our method above. This is hardly surprising as \cite{mm17} adopts
a general setting without imposing  stationarity.

\begin{sidewaysfigure}[htbp]
	\centering
	\includegraphics[scale=0.8]{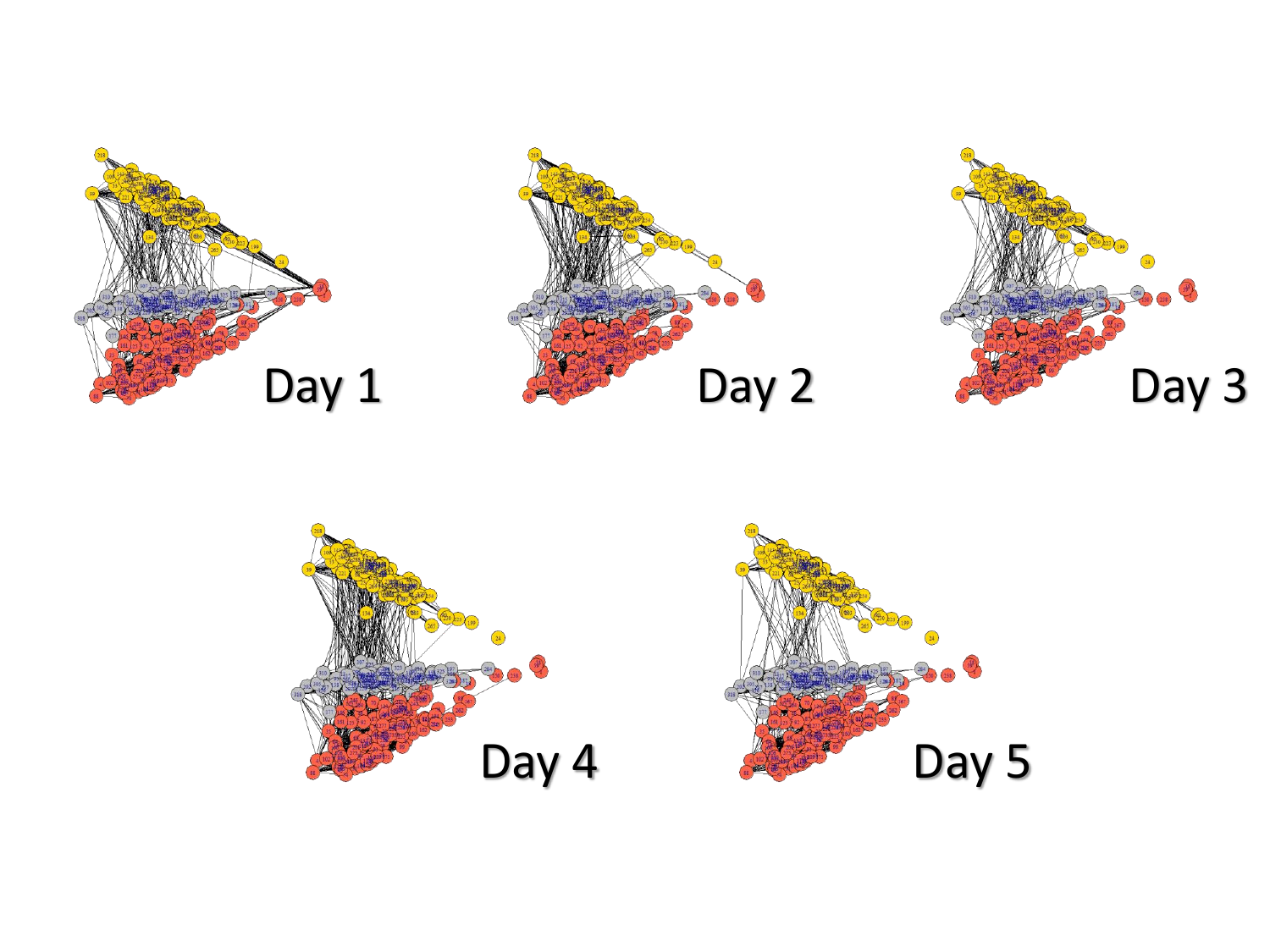}
	\caption{French high school contact networks over 5 days:
		the nodes marked in three colours represent the $q=3$ clusters determined
		by SCA based on $\wh \bL$ in (\ref{c5}).}
	\label{frenchq3}
\end{sidewaysfigure}
\begin{sidewaysfigure}[htbp]
	\centering
	\includegraphics[scale=0.8]{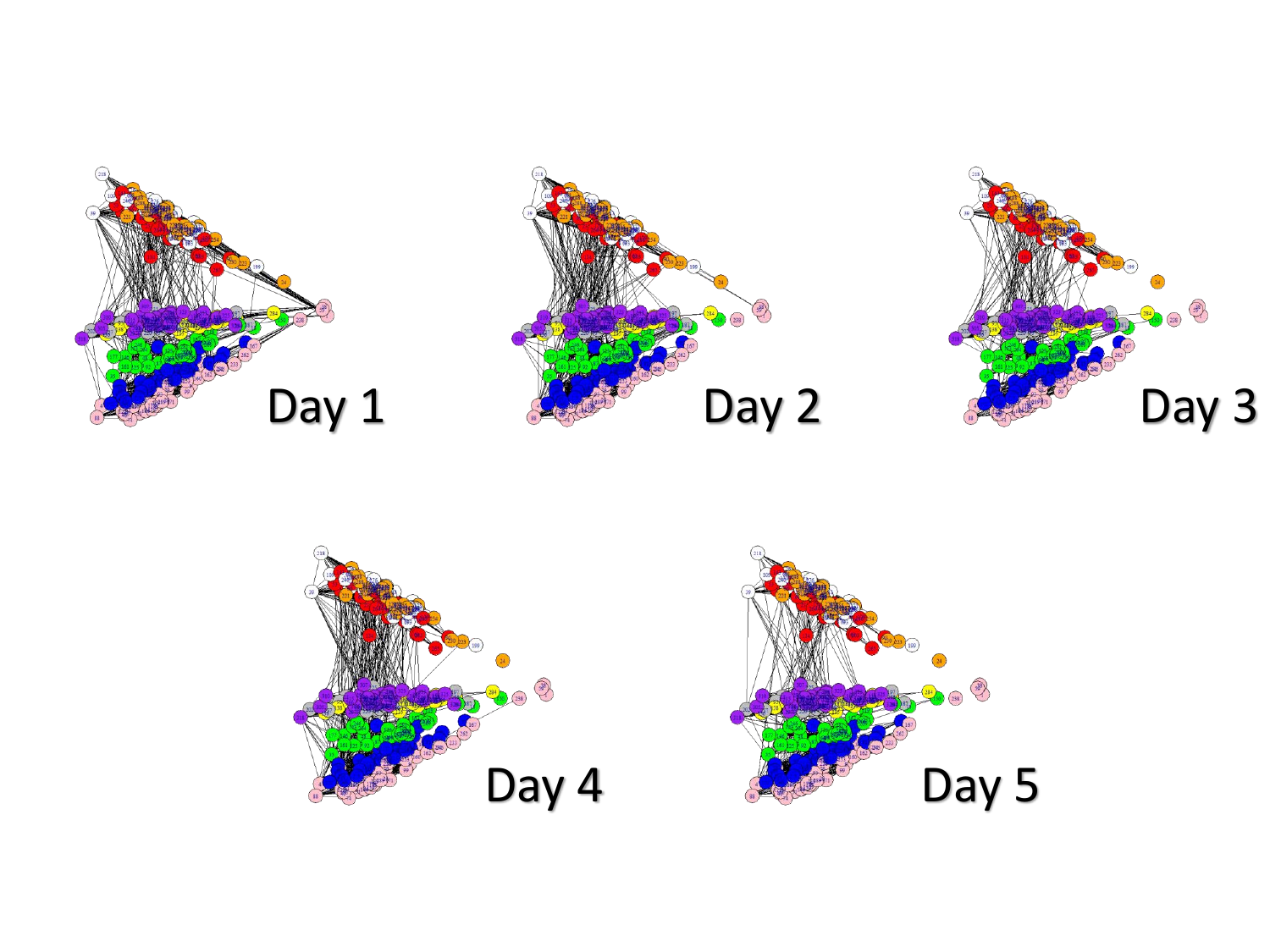}
	\caption{French high school contact networks over 5 days:
		the nodes marked in seven colours represent the $q=9$ clusters determined
		by SCA based on $\wh \bL$ in (\ref{c5}).}
	\label{frenchq7}
\end{sidewaysfigure}

\begin{singlespace}
\begin{table}[htbp]
	{\small
		\caption{Fitting AR(1) stochastic block models with $q=9$
			clusters to the French high school data:
			the estimation parameters and their standard errors (in parentheses).}
		\label{table:french7}
		\begin{center}
			\begin{tabular}{|c|c|c|c|c|c|c|c|c|c|c|}
				\hline
				&	Cluster & 1 & 2 & 3& 4 & 5 & 6 & 7& 8 & 9\\
				\hline
				& 1&   .246& .001& .004&  .006& .001& .009&.003& .024& .003 \\
				& & (.008)& (.001)& (.001)& (.001)& (.001)&  (.001)& (.001)&  (.002)&  (.001)\\
				&2&  & .136& .024& .0018& .001& .007& .001& .001& .027\\
				&&&(.009)& (.002)&(.001)& (.001)&  (.001)& (.000)& (.001)& (.002)\\
				&3& & & .252& .001& .002& .007& .001& .001 & .022\\
				&&&& (.011)&(.001)& (.001)& (.001)& (.001)& (.001)& (.002)\\
				&4& & & & .234&  .020&  .001& .024& .002& .001\\
				&&&&&(.010)& (.002)& (.001)& (.002)& (.001)& (.001)\\
				$\hat{\theta}_{i,j}$&5& & & & & .196& .001& .020& .002& .004\\
				&&&&&&(.008)& (.001)& (.002)& (.000)&(.001)\\
				&6& & & & & & .181& .001& .010& .007\\
				&&&&&&&(.008)& (.001)& (.001)&(.001)\\
				&7& & & & & & & .252& .003& .006\\
				&&&&&&&&(.009)& (.001)&(.001)\\
				&8&&&&&&&&.202& .001\\
				&&&&&&&&&(.006)&(.001)\\
				&9&&&&&&&&&.219\\
				&&&&&&&&&&(.008)\\
				\hline
				\hline
				& 1&  .563& .999& .959& .976& .999& .867&.870& .792& .909\\
				&&(.015)& (.001)& (.036)&(.098)& (.001)& (.054)& (.001)& (.000)&(.051)\\
				&2&  & .472& .761& .888& .999& .866& .999& .999& .866\\
				&&&(.024)& (.036)&(.097)& (.001)& (.054)& (.001)& (.000)&(.026)\\
				&3& & & .453& .999& .928& .864& .999& .999& .772 \\
				&&&&(.016)& (.000)& (.066)&(.048)& (.000)& (.000)&(.031)\\
				&4& & & &  .509& .868& .999& .784& .956& .999\\
				&&&&&(.017)& (.028)& (.000)& (.029)& (.041)&(.000)\\
				$\hat{\eta}_{i,j}$&5& & & & & .544& .999&.929&.842& .935\\
				&&&&&&(.017)& (.001)& (.021)& (.078)&(.041)\\
				&6& & & & & &  .589& .999&.793& .923\\
				&&&&&&&(.019)& (.001)& (.040)&(.036)\\
				&7& & & & & & & .480& .999& .814\\
				&&&&&&&&(.014)&(.000)&(.051)\\
				&8&&&&&&&&.504& .999\\
				&&&&&&&&&(.127)&(.000)\\
				&9&&&&&&&&&.471\\
				&&&&&&&&&&(.014)\\
				\hline			
			\end{tabular}
		\end{center}
	}
\end{table}
 \end{singlespace}

\subsection{Global trade data}

 Our last example concerns the annual international trades among $p=197$ countries between
1950 and 2014 (i.e. $n=65$). We define an edge between two countries to be 1 as long
as there exist trades between the two countries in that year (regardless the direction),
and 0 otherwise. We take this simplistic approach to illustrate our AR(1) stochastic 
block model with a change point. The data used are a subset of
the openly available trade data for 205 countries in 
1870 -- 2014 \citep{barb2009,barb2016}.
We leave out several countries, e.g. Russia and Yugoslavia, which did not
exist for the whole period concerned.

Setting $q=2$, we fit the data with an AR(1) stochastic block model with two clusters.
The $P$-value of the permutation test  
for the residuals resulted from the fitted model is 0, indicating
overwhelmingly that
the stationarity does not hold for the whole period.
Applying the maximum likelihood estimator (\ref{c9}),
the estimated change point is at year 1991. Before this change point, the identified
Cluster I contains 26 countries, including the
most developed industrial countries such as USA, Canada, UK and 
most European countries.
Cluster II contains 171 countries, including all
African and Latin American countries, and most Asian countries.
After 1991, 41 countries switched from Cluster II to Cluster I, including
Argentina, Brazil, Bulgaria, China, Chile, Columbia, Costa Rica, Cyprus, Hungary, Israel,
Japan, New Zealand, Poland, Saudi Arabia, Singapore,
South Korea, Taiwan, and United Arab Emirates. There was no single switch
from Cluster I to II. 
Note that 1990 may be viewed as the beginning of the globalization.
With the collapse of the Soviet Union in
1989, the fall of Berlin Wall and the end of the Cold War in 1991,
the world became more interconnected. The communist bloc countries in
East Europe, which had been isolated from the capitalist West, began to
integrate into the global market economy. Trade and investment increased,
while barriers to migration and to cultural exchange were lowered.

Figure \ref{heat} presents the average adjacency matrix of the 197
countries before and after the
change point, where the cold blue color indicates small value and the
warm red color indicates large value.
Before 1991, there are
only 26 countries in Cluster 1. The intensive red in the small lower left corner
indicates the intensive trades among those 26 countries. After 1991, the densely
connected lower left corner is enlarged as now there are 67 countries in Cluster 1.
Note some members of Cluster 2 also trade with the  members
of Cluster 1, though not all intensively.

\begin{sidewaysfigure}[htbp]
	\centering
	\includegraphics[scale=0.8]{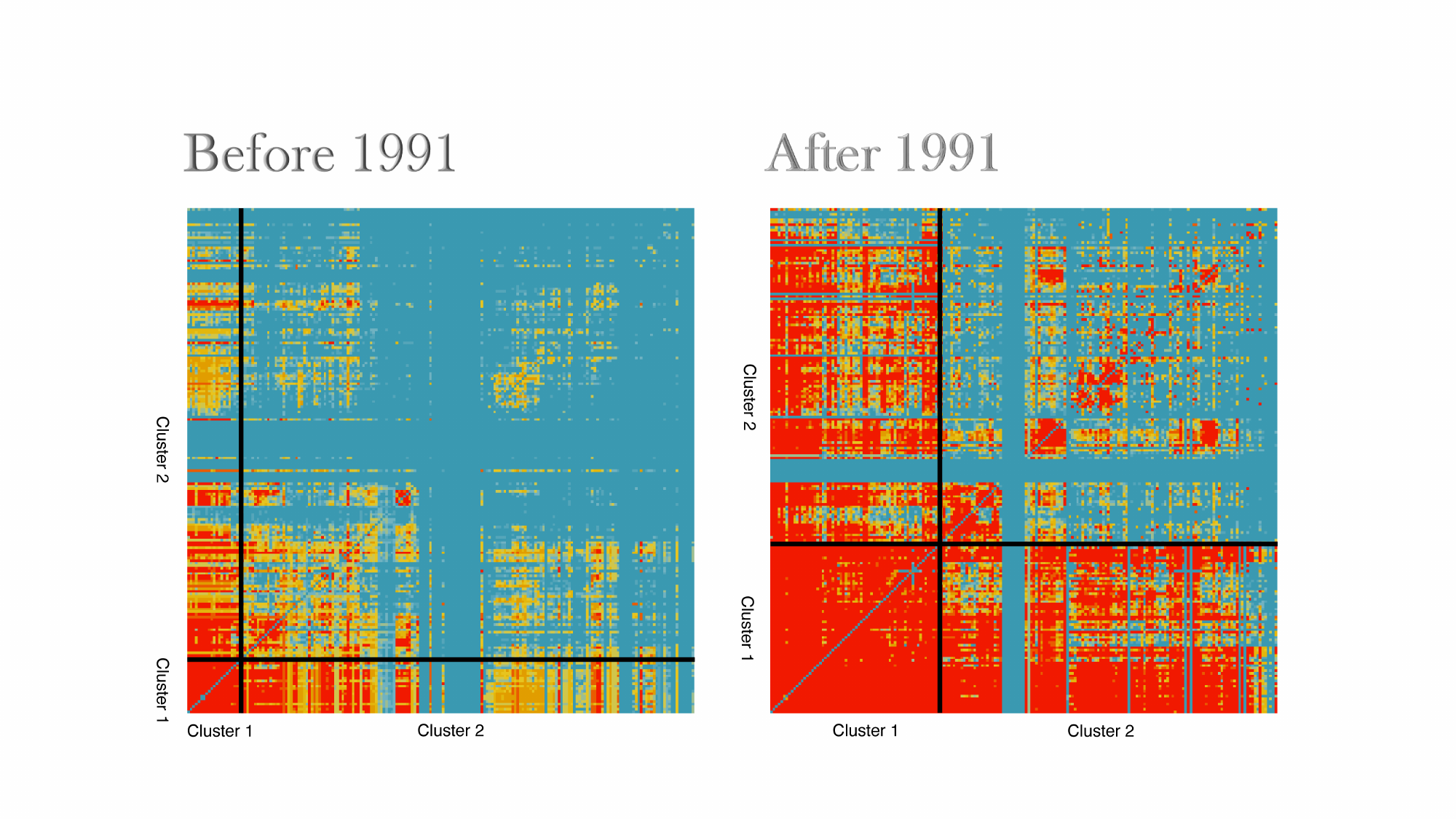}
	\caption{Average adjacency matrix for the trades among
		the 197 countries before and after 1991.  
		The values from 0 to 1 are colour-coded from blue, light blue,
		light red to red. }
	\label{heat}
\end{sidewaysfigure}

The estimated parameters for the fitted AR(1) stochastic block model
with $q=2$ clusters are reported in Table \ref{table:tradech}.
Since estimated values for $\wh\theta_{1,2}, \wh\eta_{1,2}$ before and after
the change point are always small, the trading status between the 
countries across the two clusters are unlikely to change. Nevertheless
$\wh\theta_{1,2}$ is 0.154 after 1991, and 0.053 before 1991;
indicating greater possibility for new trades to happen after 1991.

 	 \begin{singlespace}
\begin{table}[htbp]
	\centering
	\caption{Fitting AR(1) stochastic block model with a change point
		and $q=2$ to the Global trade data: the estimated AR coefficients before and after
		1991.}
	
	\askip
	\label{table:tradech}
	\begin{tabular}{|c|c|c|c|c|}
		\hline
		& \multicolumn{2}{c|}{$t\le 1991$}&\multicolumn{2}{c|}{$t>1991$}\\
		\hline
		Coefficients & Estimates & SE & Estimates & SE\\
		\hline
		$\theta_{1,1}$& .062& .0092 & .046 & .0005\\
		$\theta_{1,2}$& .053& .0008 & .154 & .0013\\
		$\theta_{2,2}$ & .023 & .0002 & .230 & .0109\\
		$\eta_{1,1}$ & .003 & .0005 & .144 & .0016\\
		$\eta_{1,2}$ & .037 & .0008 & .047 & .0007\\
		$\eta_{2,2}$ & .148& .0012 & .006 & .0003\\
		\hline
	\end{tabular}
\end{table}
 	 \end{singlespace}

\bigskip
\noindent
{\bf \large A final remark}. We proposed in this paper a simple AR(1) setting to represent
the dynamic dependence in network data explicitly. It also facilitates easy inference such
as the maximum likelihood estimation and model diagnostic checking. A new class of dynamic
stochastic block models  illustrates the usefulness of the setting in handling more
complex underlying structures including structure breaks due to change points.

It is conceivable to construct AR($p$) or even ARMA network models following the
similar lines. However a more fertile exploration is perhaps to extend the setting for
the networks with dependent edges, incorporating in the model some stylized features
of network data such as {\sl transitivity, homophily}.
The development in this direction will be reported in a follow-up paper.  
On the other hand, dynamic networks with weighted edges may be treated
as matrix time series for which effective modelling procedures have
been developed based on various tensor decompositions \citep{wlc19, chy20}.

\begin{singlespace}
\bibliographystyle{apalike}
\bibliography{reference}
\end{singlespace}

\newpage

\setcounter{page}{1}

\section*{``Autoregressive Networks'' by B. Jiang, J. Li and Q. Yao}

{\bf \Large Appendix: Technical proofs and further real data analysis}

\setcounter{equation}{0}
\renewcommand{\theequation}{A.\arabic{equation}}

\setcounter{subsection}{0}
\renewcommand{\thesubsection}{A.\arabic{subsection}}

\subsection{Proof of Proposition \ref{thm1}} \label{sec51}

Note all $X_{i,j}^t$ take
binary values 0 or 1. Hence
\begin{align*}
& P(X_{i,j}^1 =1)  = P(X_{i,j}^0=1) P(X_{i,j}^1 =1| X_{i,j}^0=1)
+ P(X_{i,j}^0=0) P(X_{i,j}^1 =1| X_{i,j}^0=0) \\
=& \pi_{i,j} (1 - \beta_{i,j}) + (1- \pi_{i,j}) \alpha_{i,j}
= {\alpha_{i,j} \over \alpha_{i,j} + \beta_{i,j}} (1 - \beta_{i,j}) + {\beta_{i,j} \over \alpha_{i,j} + \beta_{i,j}}
\alpha_{i,j} = {\alpha_{i,j} \over \alpha_{i,j} + \beta_{i,j}} = \pi_{i,j}.
\end{align*}
Thus $\calL(X_{i,j}^1) = \calL(X_{i,j}^0)$. Since all $\bX^t$ are
Erd\'os-Renyi, $\calL(\bX^1) = \calL(\bX^0)$.
Condition (\ref{b5}) ensures that $\{ \bX_t \}$ is a homogeneous
Markov chain. 
Hence $\calL(\bX^t) = \calL(\bX^0)$ for any $t\ge 1$.
This implies the required  stationarity.

As $E(X_{i,j}^t) = P(X_{i,j}^t =1)$, and $\var(X_{i,j}^t) = E(X_{i,j}^t) -
\{E(X_{i,j}^t)\}^2$, (\ref{b8}) follows from the stationarity, (\ref{b6})
and (\ref{b7}).

Note that  (\ref{b1}) implies
a Yule-Walker equation
\begin{equation} \label{b14}
\gamma_{i,j}(k) = (1-\alpha_{i,j} - \beta_{i,j}) \gamma_{i,j}(k-1), \quad k=1, 2, \cdots,
\end{equation}
where $\gamma_{i,j}(k) = \cov(X_{i,j}^{t+k}, X_{i,j}^t)$. 

Since the networks are all Erd\"os-Renyi, (\ref{b9}) follows from the
Yule-Walker equation (\ref{b14}) immediately, noting $\rho_{i,j}(k) = \gamma_{i,j}(k)/
\gamma_{i,j}(0)$ and $\rho_{i,j}(0)=1$.
To prove (\ref{b14}), it follows from (\ref{b1}) that for any $k\ge 1$,
\begin{align*}
E(X_{i,j}^{t+k} X_{i,j}^t) &= E(X_{i,j}^{t+k-1} X_{i,j}^t) P(\ve_{i,j}^{t+k}=0)
+ P(\ve_{i,j}^{t+k}=1) EX_{i,j}^t\\
&= (1  - \alpha_{i,j} - \beta_{i,j}) E(X_{i,j}^{t+k-1} X_{i,j}^t) + \alpha_{i,j}^2/(\alpha_{i,j} + \beta_{i,j}). 
\end{align*}
Thus
\begin{align*}
\gamma_{i,j}(k) &= E(X_{i,j}^{t+k} X_{i,j}^t) -(EX_{i,j}^t)^2
= E(X_{i,j}^{t+k} X_{i,j}^t) - {\alpha_{i,j}^2 \over
	(\alpha_{i,j}+\beta_{i,j})^2} \\ &
= (1  - \alpha_{i,j} - \beta_{i,j}) E(X_{i,j}^{t+k-1} X_{i,j}^t) + {\alpha_{i,j}^2\over
	\alpha_{i,j} + \beta_{i,j}} (1 - {1 
	\over \alpha_{i,j} + \beta_{i,j}})\\
& = (1  - \alpha_{i,j} - \beta_{i,j}) \{ E(X_{i,j}^{t+k-1} X_{i,j}^t) -
{\alpha_{i,j}^2 \over
	(\alpha_{i,j}+\beta_{i,j})^2} \}  = (1  - \alpha_{i,j} - \beta_{i,j}) \gamma_{i,j}(k-1).
\end{align*}

This completes the proof.

\subsection{Proof of Proposition \ref{thm2}}

We only prove (\ref{b11}), as (\ref{b10}) follows from (\ref{b11}) 
immediately. To prove (\ref{b11}), we only need to show 
\begin{equation} \label{d1}
d_{i,j}(k) \equiv P(X_{i,j}^t \ne X_{i,j}^{t+k}) = {2 \alpha_{i,j} \beta_{i,j}
	\over (\alpha_{i,j} + \beta_{i,j})^2} \{ 1 - (1 - \alpha_{i,j} -
\beta_{i,j})^k\}, \quad k=1,2, \cdots. 
\end{equation}
We Proceed by induction. It is easy to check that (\ref{d1})
holds for $k=1$. Assuming it also holds for $k\ge 1$, then
\begin{align*}
&d_{i,j}(k+1) =  P(X_{i,j}^t=0, X_{i,j}^{t+k+1}=1)
+ P(X_{i,j}^t=1, X_{i,j}^{t+k+1}=0) \\
=&\,
P(X_{i,j}^t=0, X_{i,j}^{t+k} =1, X_{i,j}^{t+k+1}=1)
+ P(X_{i,j}^t=0, X_{i,j}^{t+k} =0, X_{i,j}^{t+k+1}=1)\\
&+  P(X_{i,j}^t=1, X_{i,j}^{t+k} =0, X_{i,j}^{t+k+1}=0)
+ P(X_{i,j}^t=1, X_{i,j}^{t+k} =1, X_{i,j}^{t+k+1}=0) \\
= &\,
P(X_{i,j}^t=0, X_{i,j}^{t+k} =1) (1-\beta_{i,j})
+ \{ P(X_{i,j}^t=0) - P(X_{i,j}^t=0, X_{i,j}^{t+k} =1) \} \alpha_{i,j} \\
& + P(X_{i,j}^t=1, X_{i,j}^{t+k} =0) (1- \alpha_{i,j}) 
+ \{ P(X_{i,j}^t=1) - P(X_{i,j}^t=1, X_{i,j}^{t+k} =0)\} \beta_{i,j}
\\
= &\,
\{ P(X_{i,j}^t=0, X_{i,j}^{t+k} =1)
+ P(X_{i,j}^t=1, X_{i,j}^{t+k} =0)\}(1- \alpha_{i,j} - \beta_{i,j}) 
+ {2 \alpha_{i,j} \beta_{i,j} \over \alpha_{i,j} + \beta_{i,j}}\\
=& \,
d_{i,j}(k) (1- \alpha_{i,j} - \beta_{i,j}) + {2 \alpha_{i,j} \beta_{i,j} \over
	\alpha_{i,j} + \beta_{i,j}} \, = \, 
{2 \alpha_{i,j} \beta_{i,j}
	\over (\alpha_{i,j} + \beta_{i,j})^2} \{ 1 - (1 - \alpha_{i,j} -
\beta_{i,j})^{k+1}\}.
\end{align*}
Hence (\ref{d1}) also holds for $k+1$. This completes the proof.

\subsection{Proof of Proposition \ref{amixing}}
\begin{proof}
	Note that for any nonempty elements $A\in  {\mathcal F}_0^{k}, B\in   {\mathcal F}_{k+\tau}^\infty$, there exist $A_{0}\in {\mathcal F}_0^{k-1}$ and $B_0\in   {\mathcal F}_{k+\tau+1}^\infty$ such that $A=A_0\times \{0\}, A_0\times \{1\},$ or $A_0\times \{0,1\}$, and $B=B_0\times \{0\}, B_0\times \{1\},$ or $B_0\times \{0,1\}$.  We first consider the case where $B=B_0\times \{x_{k}\}$ and  $A=A_0\times \{x_{k+\tau}\}$ 
	where  $x_k, x_{k+\tau}= 0$ or  $1$.  Note that
	\begin{eqnarray*}
		&&P(A_0, X^k_{i,j}=x_k, B_0, X_{i,j}^{k+\tau}=x_{k+\tau}) \\ 
		&=& P(B_0|X_{i,j}^{k+\tau}=x_{k+\tau}) P(X_{i,j}^{k+\tau}=x_{k+\tau}, A_0, X_{i,j}^k=x_k) \\
		&=& P(B_0, X_{i,j}^{k+ \tau}=x_{k+\tau} )P(A_0, X^k_{i,j}=x_k) \cdot  \frac{P(X_{i,j}^{k+\tau}=x_{k+\tau}| X_{i,j}^k=x_k) }{P(X_{i,j}^{k+\tau}=x_{k+\tau}) }\\
		&=&P(B_0, X_{i,j}^{k+\tau}=x_{k+\tau} )P(A_0, X^k_{i,j}=x_k) \cdot  \frac{P(X_{i,j}^{k+\tau}=x_{k+\tau}, X_{i,j}^k=x_k) }{P(X_{i,j}^{k+\tau}=x_{k+\tau})P( X_{i,j}^k=x_k) }
	\end{eqnarray*}
	On the other hand, note that
	\begin{eqnarray*}
		P(X_{i,j}^{k+\tau}=1,X_{i,j}^k=1)- P(X_{i,j}^{k+\tau}=1)P(X_{i,j}^k=1)=  \rho_{i,j}(\tau);
	\end{eqnarray*}
	\begin{eqnarray*}		
		&& P(X_{i,j}^{k+\tau}=1,X_{i,j}^k=0)- P(X_{i,j}^{k+\tau}=1)P(X_{i,j}^k=0) \\
		&=&	P(X_{i,j}^{k+\tau}=1)-P(X_{i,j}^{k+\tau}=1,X_{i,j}^k=1)- P(X_{i,j}^{k+\tau}=1)[1-P(X_{i,j}^k=1)]  \\
		&=&-\rho_{i,j}(\tau);
	\end{eqnarray*}
	\begin{eqnarray*}		
		&& P(X_{i,j}^{k+\tau}=0,X_{i,j}^k=1)- P(X_{i,j}^{k+\tau}=0)P(X_{i,j}^k=1) \\
		&=&	P(X_{i,j}^{k}=1)-P(X_{i,j}^{k+\tau}=1,X_{i,j}^k=1)- [1-P(X_{i,j}^{k+\tau}=1)]P(X_{i,j}^k=1)  \\
		&=&-\rho_{i,j}(\tau);
	\end{eqnarray*}
	\begin{eqnarray*}		
		&&P(X_{i,j}^{k+\tau}=0,X_{i,j}^k=0)- P(X_{i,j}^{k+\tau}=0)P(X_{i,j}^k=0) \\
		&=&	P(X_{i,j}^{k+\tau}=0)-P(X_{i,j}^{k+\tau}=0,X_{i,j}^k=1)- P(X_{i,j}^{k+\tau}=0)[1-P(X_{i,j}^k=1)]  \\
		&=&\rho_{i,j}(\tau). 
	\end{eqnarray*}
	Consequently, we have
	\begin{eqnarray*}
		&&|P(A_0, X^k_{i,j}=x_k, B_0, X_{i,j}^{k+\tau}=x_{k+\tau})  -P(A_0, X^k_{i,j}=x_k)P(B_0, X_{i,j}^{k+\tau}=x_{k+\tau}) | \\
		&=& \left| P(A_0, X^k_{i,j}=x_k)P(B_0, X_{i,j}^{k+\tau}=x_{k+\tau})\left[\frac{P(X_{i,j}^{k+\tau}=x_{k+\tau}, X_{i,j}^k=x_k) }{P(X_{i,j}^{k+\tau}=x_{k+\tau})P( X_{i,j}^k=x_k) }-1 \right]\right|\\
		&\leq& \rho_{i,j}(\tau). 
	\end{eqnarray*}
	In the case where	$A=A_0\times \{0,1\}$ and/or $B=B_0\times
	\{0,1\}$, since $A$ and $B$ are nonempty, there exist integers $0<k_1<k$
	and/or $k_2>k+1$, and correspondingly $A_1\in  {\mathcal
		F}_0^{k_1-1}\times \{x_{k_1}\}$ and/or $B\in   {\mathcal
		F}_{k_2+\tau+1}^\infty\times \{x_{k_2+\tau}\}$ with $x_{k_1},
	x_{k_2+\tau}=0~ {\rm or}~ 1$, such that $P(A\cap B)-P(A)P(B)=P(A_1\cap
	B_1)-P(A_1)P(B_1)$. Following similar arguments above we have $P(A\cap
	B)-P(A)P(B)\leq \rho_{i,j}(\tau+k_2-k_1)<\rho_{ij}(\tau)$. We thus proved
	that $\alpha^{i,j}(\tau)\leq\rho_{i,j}(\tau)$.
	The conclusion of Proposition \ref{amixing} follows from Proposition \ref{thm1}.
\end{proof}

\subsection{Proof of Proposition \ref{uniformCon}}

We introduce some technical lemmas first. 

\begin{lemma}\label{Ystationary}
	For any $(i,j)\in {\mathcal J}$, denote $Y_{i,j}^{t}:=X_{i,j}^t(1-X_{i,j}^{t-1})$, and let $\Y_t=(Y_{i,j}^{t})_{1\leq i,j\leq p}$ be the $p\times p$ matrix at time t. Under the assumptions of Proposition \ref{thm1}, we have $\{\Y_t, t=1,2\ldots\}$ is  stationary such that for any $(i,j), (l,m)\in {\mathcal J}$, and $t, s\geq 1$, $t\neq s$, 
	\begin{eqnarray*}
		EY_{i,j}^t =
		\frac{\alpha_{i,j} \beta_{i,j} }{\alpha_{i,j}+\beta_{i,j}},
		~~~{\rm Var}(Y_{i,j}^{t})=\frac{\alpha_{i,j}\beta_{i,j}(\alpha_{i,j}+\beta_{i,j}-\alpha_{i,j}\beta_{i,j}) }{(\alpha_{i,j}+\beta_{i,j})^2},
	\end{eqnarray*}
	\begin{eqnarray*}
		\rho_{Y_{i,j}}(|t-s|)\equiv {\rm Corr}(Y_{i,j}^{t}, Y_{lm}^{s}) =\begin{cases}
			-\frac{ \alpha_{i,j}\beta_{i,j}(1-\alpha_{i,j}-\beta_{i,j})^{|t-s|-1}   }{\alpha_{i,j}+\beta_{i,j}-\alpha_{i,j}\beta_{i,j}} & {\rm if}~~(i,j)=(l,m), \\
			0 & {\rm otherwise}.  
		\end{cases}
	\end{eqnarray*}
\end{lemma}
\begin{proof}
	Note that $Y_{i,j}^{t}=X_{i,j}^t(1-X_{i,j}^{t-1})=(1-X_{i,j}^{t-1})I(\ve_{i,j}^t=1)$. 
	We thus have:
	
	$E(Y_{i,j}^t) =P(X_{i,j}^{t-1} =0)\alpha_{i,j}=(1-EX_{i,j}^{t-1})\alpha_{i,j}=			\frac{\alpha_{i,j} \beta_{i,j} }{\alpha_{i,j}+\beta_{i,j}}$.
	
	${\rm Var}(Y_{i,j}^{t})=E(Y_{i,j}^t)[1-E(Y_{i,j}^t)]=\frac{\alpha_{i,j}\beta_{i,j} }{\alpha_{i,j}+\beta_{i,j}}\left(1-\frac{\alpha_{i,j}\beta_{i,j} }{\alpha_{i,j}+\beta_{i,j}}\right)=\frac{\alpha_{i,j}\beta_{i,j}(\alpha_{i,j}+\beta_{i,j}-\alpha_{i,j}\beta_{i,j}) }{(\alpha_{i,j}+\beta_{i,j})^2}$.

	For $k=1$ we have $E(Y_{i,j}^{t}Y_{i,j}^{t+1})=E[ (1-X_{i,j}^{t-1})X_{i,j}^{t}(1-X_{i,j}^{t}) X_{i,j}^{t+1}]=0$.
	For any $k\geq 2$, using the fact that 
	$
	E(X_{ij}^t X_{ij}^{t+k}) = {\alpha_{ij} \over
		(\alpha_{ij} + \beta_{ij})^2} \{ \beta_{ij} (1 - \alpha_{ij}
	- \beta_{ij})^k + \alpha_{ij} \}$,
	we have
	\begin{eqnarray*}
		E(Y_{i,j}^{t}Y_{i,j}^{t+k})&=&E[X_{i,j}^{t}(1-X_{i,j}^{t-1})(1-X_{i,j}^{t+k-1})I(\ve_{i,j}^{t+k}=1)]\\
		&=& \alpha_{i,j}E[X_{i,j}^{t}(1-X_{i,j}^{t-1})(1-X_{i,j}^{t+k-1})]\\
		&=&\alpha_{i,j} P(X_{i,j}^{t+k-1}=0|X_{i,j}^{t}=1) P(X_{i,j}^{t}=1| X_{i,j}^{t-1}=0)P(X_{i,j}^{t-1}=0) \\
		&=&\frac{ \alpha_{i,j}^2\beta_{i,j}  }{\alpha_{i,j}+\beta_{i,j} }[1-P(X_{i,j}^{t+k-1}=1|X_{i,j}^{t}=1)  ] \\
		&=&\frac{ \alpha_{i,j}^2\beta_{i,j}  }{\alpha_{i,j}+\beta_{i,j} }\left[1-\frac{E(X_{i,j}^{t+k-1}X_{i,j}^{t})}{ EX_{i,j}^{t}} \right] \\
		&=&\frac{ \alpha_{i,j}^2\beta_{i,j}  }{\alpha_{i,j}+\beta_{i,j} } \left[1-\frac{\beta_{i,j}(1-\alpha_{i,j}-\beta_{i,j})^{k-1}+\alpha_{i,j}}{\alpha_{i,j}+\beta_{i,j}}\right] \\
		&=&\frac{ \alpha_{i,j}^2\beta_{i,j}^2 [1-   (1- \alpha_{i,j}-\beta_{i,j})^{k-1}] }{(\alpha_{i,j}+\beta_{i,j})^2 }. 
	\end{eqnarray*}
	Therefore we have for any $k\geq 1$, 
	\begin{eqnarray*}
		{\rm Cov}(Y_{i,j}^{t},Y_{i,j}^{t+k})&=&E(Y_{i,j}^{t}Y_{i,j}^{t+k})-EY_{i,j}^tEY_{i,j}^{t+k}\\
		&=&\frac{ \alpha_{i,j}^2\beta_{i,j}^2 [1-    (1- \alpha_{i,j}-\beta_{i,j})^{k-1}] }{(\alpha_{i,j}+\beta_{i,j})^2 }-\frac{\alpha_{i,j}^2\beta_{i,j}^2 }{(\alpha_{i,j}+\beta_{i,j})^2} \\
		&=&- \frac{\alpha_{i,j}^2\beta_{i,j}^2 (1-\alpha_{i,j}-\beta_{i,j})^{k-1}   }{(\alpha_{i,j}+\beta_{i,j})^2}. 
	\end{eqnarray*}
	Consequently,   for any $|t-s|= 1,2, \ldots$, the ACF of the process $\{Y_{i,j}^t, t=1,2\ldots\}$ is given as: 
	\begin{eqnarray*}
		\rho_{Y_{i,j}}(|t-s|)&=&-\frac{\alpha_{i,j}^2\beta_{i,j}^2 (1-\alpha_{i,j}-\beta_{i,j})^{|t-s|-1}   }{(\alpha_{i,j}+\beta_{i,j})^2} \cdot\frac{(\alpha_{i,j}+\beta_{i,j})^2}{\alpha_{i,j}\beta_{i,j}(\alpha_{i,j}+\beta_{i,j}-\alpha_{i,j}\beta_{i,j})} \\
		&=&-\frac{ \alpha_{i,j}\beta_{i,j}(1-\alpha_{i,j}-\beta_{i,j})^{|t-s|-1}   }{\alpha_{i,j}+\beta_{i,j}-\alpha_{i,j}\beta_{i,j}}. 
	\end{eqnarray*}
\end{proof}

Since the mixing property is hereditary, $Y_{i,j}^t$ is also
$\alpha$-mixing. 
From Proposition \ref{amixing} and Theorem 1 of \cite{merlevede2009bernstein}, we obtain the following concentration inequalities:
\begin{lemma}\label{concentration}
	Let conditions (2.5) and C1 hold. There exist positive constants $C_1$ and $C_2$ such that  for all $n\geq 4$ and $\ve<\frac{1}{(\log n)(\log \log n)}$, 
	\begin{eqnarray*}
		P\left( \left|n^{-1}\sum_{t=1}^nX_{i,j}^t-EX_{i,j}^t\right|>\ve\right)\leq \exp\{-C_1n\ve^2\}, \\
		P\left( \left|n^{-1}\sum_{t=1}^nY_{i,j}^t-EY_{i,j}^t\right|>\ve\right)\leq \exp\{-C_2n\ve^2\}.
	\end{eqnarray*}
\end{lemma}	

Now we are ready to prove Proposition \ref{uniformCon}.

\noindent
{\bf Proof of Proposition \ref{uniformCon}:}

\noindent
Let $\ve = C\sqrt{\frac{\log p}{n}}$ with $C^2C_1>2$ and $C^2C_2>2$.  Note that under condition (C2) we have $\ve=o\big(\frac{1}{(\log n)(\log \log n)}\big)$. Consequently by Lemma \ref{concentration}, Proposition \ref{thm1} and Lemma \ref{Ystationary}, we have 
\begin{eqnarray*}\label{Thm:eq1}
	P\left( \left|n^{-1}\sum_{t=1}^nX_{i,j}^{t}-\frac{\alpha_{i,j}}{\alpha_{i,j}+\beta_{i,j}}\right|>C\sqrt{\frac{\log p}{n}}\right)\leq \exp\{-C^2C_1\log p\},
\end{eqnarray*}
\begin{eqnarray*}\label{Thm:eq2}
	P\left( \left|n^{-1}\sum_{t=1}^nY_{i,j}^{t}-\frac{\alpha_{i,j} \beta_{i,j} }{\alpha_{i,j}+\beta_{i,j}}\right|>C\sqrt{\frac{\log p}{n}}\right)\leq \exp\{-C^2C_2\log p\}.
\end{eqnarray*}

Consequently, with probability greater than $1-\exp\{-C^2C_1\log p\}-\exp\{-C^2C_2\log p\}$, 
\begin{eqnarray*}
	\frac{\frac{\alpha_{i,j} \beta_{i,j} }{\alpha_{i,j}+\beta_{i,j}}-C\sqrt{\frac{\log p}{n}}}{\frac{\beta_{i,j}}{\alpha_{i,j}+\beta_{i,j}}+\frac{1}{n}+C\sqrt{\frac{\log p}{n}}}\leq \wh{\alpha}_{i,j}
	\leq 
	\frac{\frac{\alpha_{i,j} \beta_{i,j} }{\alpha_{i,j}+\beta_{i,j}}+C\sqrt{\frac{\log p}{n}}}{\frac{\beta_{i,j}}{\alpha_{i,j}+\beta_{i,j}}-\frac{1}{n}-C\sqrt{\frac{\log p}{n}}}.
\end{eqnarray*}

Note that when $n$ and $\frac{n}{\log p}$ are large enough such that, $\frac{1}{n}\leq C  \sqrt{\frac{\log p}{n}}\leq l/4$, we have 
\[
\alpha_{i,j}-\frac{\frac{\alpha_{i,j} \beta_{i,j} }{\alpha_{i,j}+\beta_{i,j}}-C\sqrt{\frac{\log p}{n}}}{\frac{\beta_{i,j}}{\alpha_{i,j}+\beta_{i,j}}+\frac{1}{n}+C\sqrt{\frac{\log p}{n}}}  \leq \frac{2C\alpha_{i,j} \sqrt{\frac{\log p}{n}}+C\sqrt{\frac{\log p}{n}}}{\frac{\beta_{i,j}}{\alpha_{i,j}+\beta_{i,j}} }\leq 3l^{-1}C \sqrt{\frac{\log p}{n}},
\]
and 
\[
\frac{\frac{\alpha_{i,j} \beta_{i,j} }{\alpha_{i,j}+\beta_{i,j}}+C\sqrt{\frac{\log p}{n}}}{\frac{\beta_{i,j}}{\alpha_{i,j}+\beta_{i,j}}-\frac{1}{n}-C\sqrt{\frac{\log p}{n}}} -\alpha_{i,j}
\leq \frac{2C\alpha_{i,j} \sqrt{\frac{\log p}{n}}+C\sqrt{\frac{\log p}{n}}}{\frac{\beta_{i,j}}{\alpha_{i,j}+\beta_{i,j}} -\frac{l}{2}}\leq 6l^{-1}C \sqrt{\frac{\log p}{n}},
\]
Therefore we conclude that when when $n$ and $\frac{n}{\log p}$ are large enough,
\begin{eqnarray}\label{alphaCon}
P\left(  |\wh{\alpha}_{i,j}-\alpha_{i,j}|\geq 6l^{-1}C\sqrt{\frac{\log p}{n}}\right)\leq   \exp\{-C^2C_1\log p\}+ \exp\{-C^2C_2\log p\}. 
\end{eqnarray}
As a result, we have
\begin{equation*}
P\left( \max_{(i,j)\in {\mathcal J}}|\wh{\alpha}_{i,j}-\alpha_{i,j}|<6l^{-1}C\sqrt{\frac{\log p}{n}}\right)\geq 1-p^2 \exp\{-C^2C_1\log p\}- p^2 \exp\{-C^2C_2\log p\}\rightarrow 1. 
\end{equation*}
Consequently we have $ \max_{(i,j)\in {\mathcal J}}|\wh{\alpha}_{i,j}-\alpha_{i,j}| =O_p\left(\sqrt{\frac{\log p}{n}}\right)$. Convergence of $\wh{\beta}_{i,j} $ can be proved similarly.

\subsection{Proof of Proposition \ref{CLT1}}

Note that the log-likelihood function for $(\alpha_{i,j}, \beta_{i,j})$ is: 
\begin{align*}
l(\alpha_{i,j}, \beta_{i,j})& = \log(\alpha_{i,j}) \sum_{t=1}^n
X_{i,j}^t(1-X_{i,j}^{t-1})
+ \log(1-\alpha_{i,j}) \sum_{t=1}^n (1-X_{i,j}^t)(1-X_{i,j}^{t-1})\\
&
+ \log( \beta_{i,j}) \sum_{t=1}^n (1-X_{i,j}^t) X_{i,j}^{t-1}
+ \log(1-\beta_{i,j}) \sum_{t=1}^n 
X_{i,j}^tX_{i,j}^{t-1}.
\end{align*}
Our first observation is that, owing to the independent edge formation assumption, all the $(\wh{\alpha}_{i,j}, \wh{\beta}_{i,j}), (i,j)\in {\mathcal J}$ pairs are independent. 
For each pair $(\alpha_{i,j}, \beta_{i,j})$, the score equations of the log-likelihood function are: 
\begin{align*}
\frac{\partial  l(\alpha_{i,j}, \beta_{i,j})}{\partial \alpha_{{i,j}}}&=\frac{1}{\alpha_{i,j}}\sum_{t=1}^n
X_{i,j}^t(1-X_{i,j}^{t-1})-\frac{1}{1-\alpha_{i,j}} \sum_{t=1}^n (1-X_{i,j}^t)(1-X_{i,j}^{t-1}),   \\
&= \left( \frac{1}{\alpha_{{i,j}}}+\frac{1}{1-\alpha_{i,j}}\right) \sum_{t=1}^n Y_{i,j}^t-\frac{1}{1-\alpha_{i,j}}\sum_{t=1}^n(1-X_{i,j}^{t})+O(1),\\
\frac{\partial  l(\alpha_{i,j}, \beta_{i,j})}{\partial \beta_{{i,j}}}&=\frac{1}{\beta_{i,j}}\sum_{t=1}^n (1-X_{i,j}^t) X_{i,j}^{t-1}-\frac{1}{1-\beta_{i,j}} \sum_{t=1}^n 
X_{i,j}^tX_{i,j}^{t-1}\\
&=   \frac{1}{\beta_{{i,j}}}\sum_{t=1}^nX_{i,j}^{t-1}+  \left( \frac{1}{\beta_{{i,j}}}+\frac{1}{1-\beta_{i,j}}\right)\sum_{t=1}^n (Y_{i,j}^t-X_{i,j}^t )\\
&= \left( \frac{1}{\beta_{{i,j}}}+\frac{1}{1-\beta_{i,j}}\right)\sum_{t=1}^n Y_{i,j}^t -\frac{1}{1-\beta_{i,j}} \sum_{t=1}^n 
X_{i,j}^t+O(1).
\end{align*}
Clearly, for any $0<  \alpha_{i,j},\beta_{i,j}, \alpha_{i,j}+\beta_{i,j}\leq 1$,  $\left( \frac{1}{\alpha_{{i,j}}}+\frac{1}{1-\alpha_{i,j}}, \frac{1}{1-\alpha_{i,j}}\right) $ and $\left( \frac{1}{\beta_{{i,j}}}+\frac{1}{1-\beta_{i,j}}, \frac{-1}{1-\beta_{i,j}}\right)$ are linearly independent. On the other hand, from Proposition \ref{amixing}, Lemma \ref{concentration} and classical central limit theorems for weakly dependent sequences  \citep{bradley2007introduction, durrett2019probability}, we have $\frac{1}{\sqrt{n}}\sum_{t=1}^n Y_{i,j}^t$ and $\frac{1}{\sqrt{n}}\sum_{t=1}^n X_{i,j}^t$ and any of their nontrivial linear combinations are asymptotically normally distributed. Consequently, 
any nontrivial linear combination of 
$\frac{1}{\sqrt{n}}\frac{\partial  l(\alpha_{i,j}, \beta_{i,j})}{\partial \alpha_{{i,j}}}$, $(i,j)\in J_1$ and $\frac{1}{\sqrt{n}}\frac{\partial  l(\alpha_{i,j}, \beta_{i,j})}{\partial \beta_{{i,j}}}, (i,j)\in J_2$ converges to a normal distribution. By standard arguments for consistency of MLEs, we conclude that $(\sqrt{n}(\wh{\alpha}_{i,j}-\alpha_{i,j}), \sqrt{n}(\wh{\beta}_{i,j}-\beta_{i,j}))'$ converges to the normal distribution with mean ${\bf 0}$ and covariance matrix $I(\alpha_{i,j}, \beta_{i,j})^{-1}$, where 
$I(\alpha_{i,j}, \beta_{i,j})$  is the Fisher information matrix  given as: 
\begin{align*}
I(\alpha_{i,j}, \beta_{i,j})&=  \frac{1}{n}E
\begin{bmatrix}
\frac{\sum_{t=1}^n
	X_{i,j}^t(1-X_{i,j}^{t-1}) }{\alpha^2_{i,j}} +\frac{\sum_{t=1}^n (1-X_{i,j}^t)(1-X_{i,j}^{t-1}) }{(1-\alpha_{i,j})^2}   & 0\\
0 & \frac{\sum_{t=1}^n (1-X_{i,j}^t) X_{i,j}^{t-1}}{\beta_{i,j}^2}    +\frac{\sum_{t=1}^n 
	X_{i,j}^tX_{i,j}^{t-1}}{(1-\beta_{i,j})^2}  
\end{bmatrix} .
\end{align*}
Note that 
\begin{align*}
&\frac{1}{n}E\sum_{t=1}^n
X_{i,j}^t(1-X_{i,j}^{t-1}) =  \frac{1}{n}E\sum_{t=1}^n (1-X_{i,j}^t) X_{i,j}^{t-1} =	\frac{\alpha_{i,j} \beta_{i,j} }{\alpha_{i,j}+\beta_{i,j}},  \\
&\frac{1}{n}E\sum_{t=1}^n(1-X_{i,j}^t)(1-X_{i,j}^{t-1})=\frac{ \beta_{i,j} }{\alpha_{i,j}+\beta_{i,j}}-\frac{\alpha_{i,j} \beta_{i,j} }{\alpha_{i,j}+\beta_{i,j}}=\frac{(1-\alpha_{i,j}) \beta_{i,j} }{\alpha_{i,j}+\beta_{i,j}},\\
&\frac{1}{n}E\sum_{t=1}^nX_{i,j}^tX_{i,j}^{t-1}=\frac{\alpha_{i,j} (1-\beta_{i,j}) }{\alpha_{i,j}+\beta_{i,j}}.
\end{align*}
We thus have 
\begin{align*}
I(\alpha_{i,j}, \beta_{i,j})&=  
\begin{bmatrix}
\frac{\beta_{i,j}}{\alpha_{i,j}(\alpha_{i,j}+\beta_{i,j})}+\frac{\beta_{i,j}}{(\alpha_{i,j}+\beta_{i,j})(1-\alpha_{i,j})} & 0\\
0 & \frac{\alpha_{i,j}}{\beta_{i,j}(\alpha_{i,j}+\beta_{i,j})}  + \frac{\alpha_{i,j}}{(1-\beta_{i,j})(\alpha_{i,j}+\beta_{i,j})}  
\end{bmatrix}  \\
&=\begin{bmatrix}
\frac{\beta_{i,j}}{\alpha_{i,j}(\alpha_{i,j}+\beta_{i,j})(1-\alpha_{i,j})} & 0\\
0 & \frac{\alpha_{i,j}}{\beta_{i,j}(\alpha_{i,j}+\beta_{i,j})(1-\beta_{i,j})}
\end{bmatrix} .
\end{align*}
Consequently, we have 
\begin{align*}
\begin{bmatrix}
\sqrt{n}(\wh{\alpha}_{i,j}-\alpha_{i,j}) \\ 
\sqrt{n}(\wh{\beta}_{i,j}-\beta_{i,j})
\end{bmatrix} 
\rightarrow N\Bigg({\bf 0}, 
\begin{bmatrix}
\frac{\alpha_{i,j}(\alpha_{i,j}+\beta_{i,j})(1-\alpha_{i,j})}{\beta_{i,j}} & 0\\
0 & \frac{\beta_{i,j}(\alpha_{i,j}+\beta_{i,j})(1-\beta_{i,j})}{\alpha_{i,j}}
\end{bmatrix}
\Bigg).
\end{align*}
This together with the independence among the $(\wh{\alpha}_{i,j}, \wh{\beta}_{i,j}), (i,j)\in {\mathcal J}$ pairs proves the proposition.

\subsection{Proof of Proposition \ref{OracleSC}}

Denote $\N= \diag\{\sqrt{s_1},\ldots, \sqrt{s_q}\}$. 
Note that
\begin{eqnarray*}
	\L&=&\D^{-1/2}_1 \bZ\bOmega_1\bZ^\top \D^{-1/2}_1+\D^{-1/2}_2 \bZ\bOmega_2\bZ^\top \D^{-1/2}_2\\
	&=&  \bZ\wt\D^{-1/2}_1\bOmega_1\wt\D^{-1/2}_1\bZ^\top  +  \bZ\wt\D^{-1/2}_2\bOmega_2\wt\D^{-1/2}_2\bZ^\top  \\
	&=& \bZ(\wt\bOmega_1+\wt\bOmega_2) \bZ^\top\\
	&=& (\bZ  \N^{-1}) \N \wt\bOmega\N  (\bZ  \N^{-1})^\top.
\end{eqnarray*}
Note that  the columns of $\bZ\N^{-1}$ are orthonormal, we thus have $rank(\L)=q$. Let $\bQ\bLambda\bQ^\top=\N\wt\bOmega\N$ be the eigen-decomposition of $\N\wt\bOmega\N $,  we immediately have $\L= (\bZ  \N^{-1}) \bQ\bLambda\bQ^\top (\bZ  \N^{-1})^\top$. Again, since the columns of $\bZ\N^{-1}$ are orthonormal, we conclude that $\bGamma_q=\bZ  \N^{-1}\bQ$, and $\bU=\N^{-1}\bQ$. On the other hand, note that $\bU$ is invertible, we conclude that $\bz_{i,\cdot}\bU=\bz_{j,\cdot}\bU$ and $\bz_{i,\cdot}=\bz_{j,\cdot}$ are equivalent.

\subsection{Proof of Theorem \ref{DKthm}}

The key step is to establish an upper bound for the Frobenius norm  $\|\wh\bL\wh\bL-\bL\bL\|_F$, and the theorem can be proved by Weyl's inequality and the Davis-Kahan theorem. 
We first introducing some technical lemmas.

\begin{lemma}\label{cov}
	Under the assumptions of Proposition \ref{thm1}, we have, there exists a constant $C_l>0$ such that   
	\begin{eqnarray*}
		&&Cov\left(\sum_{t=1}^nY_{i,j}^t, \sum_{t=1}^n (1-X_{i,j}^{t-1}) \right) =  -	 Cov\left(\sum_{t=1}^nY_{i,j}^t, \sum_{t=1}^nX_{i,j}^{t-1}\right)  
		\\
		&=&
		\frac{n\alpha_{i,j}\beta_{i,j} [2\alpha_{i,j}(1-\beta_{i,j})+\alpha_{{i,j}}+\beta_{{i,j}}-2\beta_{i,j}^2]}{(\alpha_{i,j}+\beta_{i,j})^3  } 
		+C_{i,j},
	\end{eqnarray*}
	with $|C_{i,j}|\leq C_l$ for any $C_{i,j}, (i,j)\in {\cal J}$.
\end{lemma}
\begin{proof}
	In the following we shall be using the fact that for any $0\leq x<1 $,
	$\sum_{h=1}^{n-1}x^{h-1}=\frac{1-x^n}{1-x}=\frac{1}{1-x}+o(1)$, and
	$\sum_{h=1}^{n-1}hx^{h-1}=\frac{1-x^n-n(1-x)x^{n-1}}{(1-x)^2} =O(1)$. In particular, when $x=1-\alpha_{i,j}-\beta_{i,j}$, under condition C1, we have  $2l\leq 1-x<1$, the $O(1)$ term in will become bounded uniformly for any $(i,j)\in {\cal J}$.
	In what follows, with some abuse of notation, we shall use  $O_l(1)$  to denote a generic constant term with magnitude bounded by a large enough constant $C_l$ that depends on $l$ only.  
	\begin{eqnarray}\label{cov:eq1}
	&&	  Cov\left(\sum_{t=1}^nY_{i,j}^t, \sum_{t=1}^n (1-X_{i,j}^{t-1}) \right)  
	=   -	 Cov\left(\sum_{t=1}^nY_{i,j}^t, \sum_{t=1}^nX_{i,j}^{t-1}\right)   \nonumber\\
	&=& -  \sum_{t=1}^n\sum_{s=1}^n  \left[E(1-X_{i,j}^{t-1})X_{i,j}^tX_{i,j}^{s-1}-
	\frac{\alpha_{i,j} \beta_{i,j}}{\alpha_{i,j}+\beta_{i,j}}\cdot\frac{\alpha_{i,j}}{\alpha_{i,j}+\beta_{i,j}}\right] \nonumber\\
	&=& -  \sum_{t=1}^n\sum_{s=1}^n  \bigg\{   \frac{\alpha_{i,j}}{(\alpha_{i,j}+\beta_{i,j})^2}\Big[\beta_{i,j}(1-\alpha_{i,j}-\beta_{i,j})^{|t-s+1|}+\alpha_{i,j}\Big] 
	-\frac{\alpha_{i,j}^2 \beta_{i,j}}{(\alpha_{i,j}+\beta_{i,j})^2} 
	\bigg\} \nonumber\\
	&&+\sum_{t=1}^n\sum_{s=1}^n  E(X_{i,j}^{t-1}X_{i,j}^tX_{i,j}^{s-1}) \nonumber\\
	&=& -  \sum_{t=1}^n\sum_{s=1}^n    \frac{\alpha_{i,j}\beta_{i,j}(1-\alpha_{i,j}-\beta_{i,j})^{|t-s+1|}   }{(\alpha_{i,j}+\beta_{i,j})^2} -\frac{n^2\alpha_{i,j}^2 (1-\beta_{i,j})}{(\alpha_{i,j}+\beta_{i,j})^2} 
	\nonumber \\
	&&+ (2n-1)E(X_{i,j}^{t-1}X_{i,j}^t )+ \sum_{s<t}  E(X_{i,j}^{t-1}X_{i,j}^tX_{i,j}^{s-1})  
	+\sum_{s>t+1}  E(X_{i,j}^{t-1}X_{i,j}^tX_{i,j}^{s-1}).
	\end{eqnarray}
	For the first three terms on the right hand side of  \eqref{cov:eq1}, we have 
	\begin{eqnarray*}
		&& -  \sum_{t=1}^n\sum_{s=1}^n    \frac{\alpha_{i,j}\beta_{i,j}(1-\alpha_{i,j}-\beta_{i,j})^{|t-s+1|}   }{(\alpha_{i,j}+\beta_{i,j})^2} -\frac{n^2\alpha_{i,j}^2 (1-\beta_{i,j})}{(\alpha_{i,j}+\beta_{i,j})^2} 
		+ (2n-1)E(X_{i,j}^{t-1}X_{i,j}^t )  \\
		&=&-\frac{\alpha_{i,j}\beta_{i,j}}{(\alpha_{i,j}+\beta_{i,j})^2  }
		\left[ n+ \frac{2n(1-\alpha_{i,j}-\beta_{i,j})}{\alpha_{i,j}+\beta_{i,j}}\right] 
		-  \frac{n^2\alpha_{i,j}^2 (1-\beta_{i,j})}{(\alpha_{i,j}+\beta_{i,j})^2}  \\
		&&+ \frac{2n\alpha_{i,j}\Big[\beta_{i,j}(1-\alpha_{i,j}-\beta_{i,j})   +\alpha_{i,j}\Big]}{(\alpha_{i,j}+\beta_{i,j})^2}     +O_l(1) \\
		&=& \frac{3n\alpha_{i,j}\beta_{i,j} }{(\alpha_{i,j}+\beta_{i,j})^2  } -   \frac{2n\alpha_{i,j}\beta_{i,j}   }{(\alpha_{i,j}+\beta_{i,j})^3   }  
		-\frac{2n\alpha_{i,j}\beta_{i,j}  }{\alpha_{i,j}+\beta_{i,j}   }  
		+\frac{2n\alpha_{i,j}^2   }{(\alpha_{i,j}+\beta_{i,j})^2   } 
		-  \frac{n^2\alpha_{i,j}^2 (1-\beta_{i,j})}{(\alpha_{i,j}+\beta_{i,j})^2}  +O_l(1).
	\end{eqnarray*}
	For the last two terms on the right hand side of  \eqref{cov:eq1}, we have 
	\begin{eqnarray*}
		&& \sum_{s<t}  E(X_{i,j}^{t-1}X_{i,j}^tX_{i,j}^{s-1})  
		+\sum_{s>t+1}  E(X_{i,j}^{t-1}X_{i,j}^tX_{i,j}^{s-1})\\
		&=&\sum_{s<t} P(X_{i,j}^t=1|X^{t-1}_{i,j}=1)P(X^{t-1}_{i,j}=1, X_{i,j}^{s-1}=1) \\
		&&+\sum_{s>t+1}  P(X_{i,j}^{s-1}=1|X^{t}_{i,j}=1)P(X^{t}_{i,j}=1, X_{i,j}^{t-1}=1)\\
		&=& 	  (1-\beta_{i,j})\sum_{s<t}   E(X_{i,j}^{t-1}X_{i,j}^{s-1})
		+(1-\beta_{i,j})\sum_{s>t+1}  E(X_{i,j}^{s-1}X_{i,j}^{t}) \\ 
		&=&\frac{(1-\beta_{i,j})\alpha_{i,j}}{(\alpha_{i,j}+\beta_{i,j})^2}\sum_{h=1}^{n-1} (n-h)[\beta_{i,j}(1-\alpha_{i,j}-\beta_{i,j})^h+\alpha_{i,j}] \\
		&&+\frac{(1-\beta_{i,j})\alpha_{i,j}}{(\alpha_{i,j}+\beta_{i,j})^2}\sum_{h=2}^{n-1} (n-h)[\beta_{i,j}(1-\alpha_{i,j}-\beta_{i,j})^{h-1}+\alpha_{i,j}]  \\
		&=&\frac{(n-1)^2\alpha_{i,j}^2 (1-\beta_{i,j})}{(\alpha_{i,j}+\beta_{i,j})^2} + 
		\frac{2n(1-\beta_{i,j})\alpha_{i,j}\beta_{i,j} }{(\alpha_{i,j}+\beta_{i,j})^3}+O_l(1).
	\end{eqnarray*}
	Consequently, we have 
	\begin{eqnarray*} 
		&&	  Cov\left(\sum_{t=1}^nY_{i,j}^t, \sum_{t=1}^n (1-X_{i,j}^{t-1}) \right)  
		=   -	 Cov\left(\sum_{t=1}^nY_{i,j}^t, \sum_{t=1}^nX_{i,j}^{t-1}\right)  \\
		&=&\frac{3n\alpha_{i,j}\beta_{i,j} }{(\alpha_{i,j}+\beta_{i,j})^2  } -   \frac{2n\alpha_{i,j}\beta_{i,j}   }{(\alpha_{i,j}+\beta_{i,j})^3   }  
		-\frac{2n\alpha_{i,j}\beta_{i,j}  }{\alpha_{i,j}+\beta_{i,j}   }  
		+\frac{2n\alpha_{i,j}^2   }{(\alpha_{i,j}+\beta_{i,j})^2   } 
		-  \frac{n^2\alpha_{i,j}^2 (1-\beta_{i,j})}{(\alpha_{i,j}+\beta_{i,j})^2} \\
		&&+\frac{(n-1)^2\alpha_{i,j}^2 (1-\beta_{i,j})}{(\alpha_{i,j}+\beta_{i,j})^2} + 
		\frac{2n(1-\beta_{i,j})\alpha_{i,j}\beta_{i,j} }{(\alpha_{i,j}+\beta_{i,j})^3}+O_l(1) \\
		&=& \frac{3n\alpha_{i,j}\beta_{i,j} }{(\alpha_{i,j}+\beta_{i,j})^2  } -   \frac{2n\alpha_{i,j}\beta_{i,j}   }{(\alpha_{i,j}+\beta_{i,j})^3   }  
		-\frac{2n\alpha_{i,j}\beta_{i,j}  }{\alpha_{i,j}+\beta_{i,j}   }  
		+\frac{2n\alpha_{i,j}^2 \beta_{i,j} }{(\alpha_{i,j}+\beta_{i,j})^2   }  \\
		&& +\frac{2n(1-\beta_{i,j})\alpha_{i,j}\beta_{i,j} }{(\alpha_{i,j}+\beta_{i,j})^3}+O_l(1) \\
		&=&\frac{3n\alpha_{i,j}\beta_{i,j} }{(\alpha_{i,j}+\beta_{i,j})^2  } 
		- \frac{2n\alpha_{i,j}\beta_{i,j}^2 (1+\alpha_{i,j}+\beta_{i,j})}{(\alpha_{i,j}+\beta_{i,j})^3  } +O_l(1).
	\end{eqnarray*}
	This proves the lemma. 
\end{proof}

\begin{lemma}\label{bias}
	{ {\rm{(Bias of $\wh{\alpha}_{i,j}$ and $\wh{\beta}_{i,j}$) } } } Ket conditions C1, C2 and the  assumptions of Proposition \ref{thm1} hold. We have 
	\[
	E\wh{\alpha}_{i,j} -\alpha_{i,j} =-\frac{ \alpha_{i,j}  [2\alpha_{i,j}(1-\beta_{i,j})+\alpha_{{i,j}}+\beta_{{i,j}}-2\beta_{i,j}^2]}{n(\alpha_{i,j}+\beta_{i,j}) \beta_{i,j}  }  +\frac{R_{i,j}^{(1)}}{n}+O(n^{-2}), \]
	\[E\wh{\beta}_{i,j} -\beta_{i,j} =
	\frac{\beta_{i,j}    [2\alpha_{i,j}(1-\beta_{i,j})+\alpha_{{i,j}}+\beta_{{i,j}}-2\beta_{i,j}^2]}{n(\alpha_{i,j}+\beta_{i,j}) \alpha_{i,j}  } +\frac{R_{i,j}^{(2)}}{n}+O(n^{-2}),
	\]
	where $R_{i,j}^{(1)}$ and $R_{i,j}^{(2)}$ are constants such that when $n$ is large enough we have $0\leq R_{i,j}^{(1)}, R_{i,j}^{(2)}\leq R_l$ for some constant $R_l$ and all $(i,j)\in {\cal J}$.
\end{lemma}
\begin{proof}
	From Lemma \ref{concentration} we have, under Condition C2,  the event $\{ |n^{-1}\sum_{t=1}^n X_{i,j}^{t-1}-\pi_{i,j}|\le (1-\pi_{i,j})/2, 1\le i,j\le p\}$ holds with probability larger than $1-O(n^{-2})$.  Denote ${\cal I}:=I(|n^{-1}\sum_{t=1}^n X_{i,j}^{t-1}-\pi_{i,j}|\le (1-\pi_{i,j})/2, 1\le i,j\le p)$.	
	By expanding $\frac{1}{1-n^{-1}\sum_{t=1}^n X_{i,j}^{t-1}}$ around $\frac{1}{1-\pi_{i,j}}$, 
	we have 
	\begin{eqnarray*}
		&&E\wh{\alpha}_{i,j}{\cal I}  \\
		&=& E \frac{ n^{-1}\sum_{t=1}^n X^t_{i,j}(1- X_{i,j}^{t-1})}{
			n^{-1}\sum_{t=1}^n (1- X_{i,j}^{t-1})}  {\cal I}  \\
		&=&  \frac{1}{n}E   \sum_{t=1}^n X^t_{i,j}(1- X_{i,j}^{t-1}) \left[ \frac{1}{1-\pi_{i,j}} 
		+ \frac{( n^{-1}\sum_{t=1}^n X_{i,j}^{t-1}-\pi_{i,j} )}{(1-\pi_{i,j})^2} 
		+ \sum_{k=2}^{\infty}\frac{( n^{-1}\sum_{t=1}^n X_{i,j}^{t-1}-\pi_{i,j} )^k}{(1-\pi_{i,j})^{k+1}}
		\right]  {\cal I} .
	\end{eqnarray*}
	Write $R_{i,j}^{(1)}:=  E\sum_{t=1}^n X^t_{i,j}(1- X_{i,j}^{t-1})\Big(\sum_{k=2}^{\infty}\frac{( n^{-1}\sum_{t=1}^n X_{i,j}^{t-1}-\pi_{i,j} )^k}{(1-\pi_{i,j})^{k+1}}\Big){\cal I}$.  
	By Taylor series with Lagrange remainder we have there exist  random scalars $r^t_{i,j}\in [ n^{-1}\sum_{t=1}^n X_{i,j}^{t-1},\pi_{i,j}]$ such that 
	\[
	R_{i,j}^{(1)} =  E \sum_{t=1}^n X^t_{i,j}(1- X_{i,j}^{t-1})\Bigg(\frac{( n^{-1}\sum_{t=1}^n X_{i,j}^{t-1}-\pi_{i,j} )^2}{(1-r^t_{i,j})^{3}}\Bigg){\cal I}>0.
	\]
	On the other hand,  
	note that  
	\begin{eqnarray*}
		\sum_{k=2}^{\infty}\frac{ |n^{-1}\sum_{t=1}^n X_{i,j}^{t-1}-\pi_{i,j} |^k}{(1-\pi_{i,j})^{k+1}}  {\cal I} &\leq  & 
		\left( n^{-1}\sum_{t=1}^n X_{i,j}^{t-1}-\pi_{i,j} \right)^2\sum_{k=0}^{\infty}\frac{1}{(1-\pi_{i,j})^{3}2^k}
		\\
		&=& \left( n^{-1}\sum_{t=1}^n X_{i,j}^{t-1}-\pi_{i,j} \right)^2\frac{2}{(1-\pi_{i,j})^{3} }.
	\end{eqnarray*}
	Therefore, 
	\begin{eqnarray*}
		R_{i,j}^{(1)}&\leq& E\sum_{t=1}^n  \Big(\sum_{k=2}^{\infty}\frac{|n^{-1}\sum_{t=1}^n X_{i,j}^{t-1}-\pi_{i,j} |^k}{(1-\pi_{i,j})^{k+1}}\Big){\cal I}
		\\
		&\leq&  Var\left(\frac{1}{\sqrt{n}}\sum_{t=1}^n X_{ij}^{t-1}\right)\frac{2}{(1-\pi_{i,j})^{3} } \\
		&=& \frac{2}{ (1-\pi_{i,j})^{3} } Var(X_{ij}^t)\left[1+\frac{2}{n}\sum_{h=1}^{n-1}(n-h)\rho_{ij}(h) \right]\\
		&=&\frac{2}{ (1-\pi_{i,j})^{3} }\cdot \frac{\alpha_{ij}\beta_{ij} }{(\alpha_{ij}+\beta_{ij})^2}
		\Bigg[1 + \frac{2}{n}\sum_{h=1}^{n-1}(n-h)(1-\alpha_{ij}-\beta_{ij})^h 
		\bigg] \\
		&=&\frac{2}{ (1-\pi_{i,j})^{3} }\cdot \frac{\alpha_{ij}\beta_{ij} }{(\alpha_{ij}+\beta_{ij})^2}
		\left[1 + \frac{2(1-\alpha_{ij}-\beta_{ij})}{\alpha_{ij}+\beta_{ij}} +O(n^{-1})\right] \\
		&=&\frac{2}{ (1-\pi_{i,j})^{4}\pi_{i,j}}\cdot  \frac{   2-\alpha_{ij}-\beta_{ij}  }{ \alpha_{ij}+\beta_{ij}  }+O(n^{-1}).
	\end{eqnarray*}
	Again, since $0<l\leq \alpha_{i,j},\beta_{i,j}, \alpha_{i,j}+\beta_{i,j}\leq 1$ holds for all $(i,j)\in {\mathcal J}$, we conclude that there exists a constant $R_l$ such that 
	$ R_{i,j}^{(1)}\leq R_l$.
	Together with Lemma \ref{cov}, we have
	\begin{eqnarray*}
		E\wh{\alpha}_{i,j}&=& E\wh{\alpha}_{i,j}{\cal I} +E\wh{\alpha}_{i,j}(1-{\cal I}) 
		\\ &=&  E   \frac{1}{n}\sum_{t=1}^nX^t_{i,j}(1- X_{i,j}^{t-1}) \left[ \frac{1}{1-\pi_{i,j}} 
		+ \frac{( n^{-1}\sum_{t=1}^n X_{i,j}^{t-1}-\pi_{i,j} )}{(1-\pi_{i,j})^2} 
		\right]{\cal I}+\frac{R_{i,j}^{(1)}}{n} +E\wh{\alpha}_{i,j}(1-{\cal I}) \\
		&=& \alpha_{i,j}+\frac{Cov(\sum_{t=1}^nY_{i,j}^t, \sum_{t=1}^nX_{i,j}^t)}{n^2(1-\pi_{i,j})^2} +\frac{R_{i,j}^{(1)}}{n} +O(n^{-2})\\
		&=& \alpha_{i,j} -\frac{ \alpha_{i,j}  [2\alpha_{i,j}(1-\beta_{i,j})+\alpha_{{i,j}}+\beta_{{i,j}}-2\beta_{i,j}^2]}{n(\alpha_{i,j}+\beta_{i,j}) \beta_{i,j}  }   +\frac{R_{i,j}^{(1)}}{n}+O(n^{-2}). 
	\end{eqnarray*}
	
	Similarly, write ${  R}_{i,j}^{(2)}:=  E\sum_{t=1}^nX^t_{i,j}(1- X_{i,j}^{t-1}) \Big(\sum_{k=2}^\infty \frac{( n^{-1}\sum_{t=1}^n X_{i,j}^{t-1}-\pi_{i,j} )^k}{(-1)^k\pi_{i,j}^{k+1}}\Big){\cal I }'$ where ${\cal I}':=I\{
	|n^{-1}\sum_{t=1}^n (1- X_{i,j}^t)X_{i,j}^{t-1} | \le \pi_{i,j}/2\}$. We have, 
	\begin{eqnarray*}\label{rij2}
		E\wh{\beta}_{i,j} &=& E \frac{  n^{-1}\sum_{t=1}^n (1- X_{i,j}^t)X_{i,j}^{t-1} }{  n^{-1}\sum_{t=1}^n X_{i,j}^{t-1}} \nonumber\\
		&=&  E  \frac{1}{n}\sum_{t=1}^n (1- X_{i,j}^t)X_{i,j}^{t-1} \Big[ \frac{1}{\pi_{i,j}} 
		- \frac{( n^{-1}\sum_{t=1}^n X_{i,j}^{t-1}-\pi_{i,j} )}{\pi_{i,j}^2} \nonumber\\
		&&
		+ \sum_{k=2}^\infty \frac{( n^{-1}\sum_{t=1}^n X_{i,j}^{t-1}-\pi_{i,j} )^k}{(-1)^k\pi_{i,j}^{k+1}}	\Big]{\cal I}' 
		+  E\wh{\beta}_{i,j} (1-{\cal I}')
		\nonumber\\
		&=&\beta_{i,j}-\frac{Cov(\sum_{t=1}^nY_{i,j}^t, \sum_{t=1}^nX_{i,j}^t-X_{i,j}^n+X_{i,j}^0)}{n^2\pi_{i,j}^2}+\frac{  R_{i,j}^{(2)}}{n} +O(n^{-2})\nonumber\\
		&=& \beta_{i,j}+\frac{\beta_{i,j}    [2\alpha_{i,j}(1-\beta_{i,j})+\alpha_{{i,j}}+\beta_{{i,j}}-2\beta_{i,j}^2]}{n(\alpha_{i,j}+\beta_{i,j}) \alpha_{i,j}  } +\frac{  R_{i,j}^{(2)}}{n}+O(n^{-2}).
	\end{eqnarray*}
	Here in the second last   step we have used the fact that
	$En^{-1}(X_{i,j}^0-X_{i,j}^n)(n^{-1}\sum_{t=1}^n X_{i,j}^{t-1}-\pi_{i,j})=O(n^{-2})$, and in the last step we have used the fact that 
	\begin{eqnarray*}
		&&n^{-2}E\sum_{t=1}^nX_{i,j}^t(1-X_{i,j}^{t-1})( X_{i,j}^n-X_{i,j}^0)\\
		&=& n^{-2}E\left[\sum_{t=1}^n X_{i,j}^{t-1}X_{i,j}^{t}X_{i,j}^{0}-\sum_{t=1}^n X_{i,j}^{t-1}X_{i,j}^{t}X_{i,j}^{n} \right] +n^{-2}[E(X_{i,j}^n)^2- E(X_{i,j}^nX_{i,j}^0)]\\
		&=&n^{-2}\Big[\sum_{t=1}^n P(X_{i,j}^t=1|X_{i,j}^{t-1}=1)P(X_{i,j}^{t-1}=1|X_{i,j}^0=1) P(X_{i,j}^0=1)  \\
		&&-\sum_{t=1}^n P(X_{i,j}^n=1|X_{i,j}^{t}=1)P(X_{i,j}^{t}=1|X_{i,j}^{t-1}=1) P(X_{i,j}^{t-1}=1) 
		\Big] +O(n^{-2}) \\
		&=&O(n^{-2})
	\end{eqnarray*}
	On one hand, similar to $R_{i,j}^{(1)}$, we can show that 
	when $n$ is large enough, there exists a $R_l$ such that  $R_{i,j}^{(2)}\leq R_l$ for any $(i,j)\in{\cal J}$.

\end{proof}

Lemma \ref{bias} implies that the bias of the MLEs is of order $O(n^{-1})$. 
The bound $R_l$ here also implies that the $O(n^{-1})$ order of the bias holds uniformly for all $(i,j)\in {\cal J}$.

\begin{lemma}\label{Frob}
	Let conditions (2.5), C1 and C2 hold.  
	For any constant $B>0$, there exists a large enough constant $C>0$ such that 
	\begin{eqnarray}
	\label{lemFrob1}  \\
	P\left\{  \|\wh \L_1\wh \L_1-\L_1\L_1\|_F\geq C\left(\sqrt{\frac{\log (pn)}{np} } +\frac{1}{n}+\frac{1}{p}\right) \right\}\leq 8p\left[ (pn)^{-(1+B)}+  \exp\{
	-B\sqrt{p}\}\right] ,   \nonumber\\
	P\left\{ \|\wh\L_2\wh\L_2-\L_2\L_2\|_F\geq C\left(\sqrt{\frac{\log (pn)}{np} } +\frac{1}{n}+\frac{1}{p}\right) \right\}\leq 8p\left[ (pn)^{-(1+B)}+  \exp\{
	-B\sqrt{p}\}\right], \nonumber\\
	P\left\{ \|\wh\L_1\wh\L_2-\L_1\L_2\|_F\geq C\left(\sqrt{\frac{\log (pn)}{np} } +\frac{1}{n}+\frac{1}{p}\right) \right\}\leq 8p\left[ (pn)^{-(1+B)}+  \exp\{
	-B\sqrt{p}\}\right],\nonumber\\
	P\left\{ \|\wh\L_2\wh\L_1-\L_2\L_1\|_F\geq C\left(\sqrt{\frac{\log (pn)}{np} } +\frac{1}{n}+\frac{1}{p}\right) \right\}\leq 8p\left[ (pn)^{-(1+B)}+  \exp\{
	-B\sqrt{p}\}\right].\nonumber
	\end{eqnarray}
\end{lemma}
\begin{proof}
	We only prove the first inequality in \eqref{lemFrob1} here as the other three inequalities can be proved similarly. 
	Denote 
	\[
	\wt{\L}_1:=\L_1-\diag(\L_1)=\D_1^{-1/2}\left[\W_1-\diag(\W)_1\right]\D_1^{-1/2},
	\]
	and for any $1\leq i, j\leq p$ we denote the $(i,j)$th element of 
	$\wt \L_1\wt \L_1-\L_1\L_1$ as $ \delta_{i,j}$.  Correspondingly, for any $\ell=1, \ldots, p$, we define $\wt d_{\ell,1}:=d_{\ell,1}-\alpha_{\ell,\ell}$. 
	We first evaluate the error introduced by removing the $\diag(\L_1)$ term. With some abuse of notation, let $\wt\alpha_{i,j}=\alpha_{i,j}$ for $1\leq i\neq j\leq p$ and $\wt\alpha_{i,i}=0$ for $i=1,\ldots,p$. We have $\W-\diag(\W)=(\wt\alpha_{i,j})_{1\leq i, j\leq p}$. Therefore, 
	\begin{eqnarray*}
		|\delta_{i,j}|=\left| \sum_{k=1}^p\frac{   \wt{\alpha}_{i,k} \wt{\alpha}_{k,j}}{  d_{k,1}\sqrt{  d_{i,1}  d_{j,1}}}-\sum_{k=1}^p\frac{   {\alpha}_{i,k} {\alpha}_{k,j}}{  d_{k,1}\sqrt{  d_{i,1}  d_{j,1}}} \right| 
		\leq  \frac{   {\alpha}_{i,i} {\alpha}_{i,j}}{  d_{i,1}\sqrt{  d_{i,1}  d_{j,1}}}
		+\frac{   {\alpha}_{i,j} {\alpha}_{j,j}}{  d_{j,1}\sqrt{  d_{i,1}  d_{j,1}}}
		\leq\frac{2}{(p-1)^2l^2}.
	\end{eqnarray*}
	Consequently, we have
	\begin{eqnarray}\label{De_diag}
	\|\wh \L_1\wh \L_1-\L_1\L_1\|_F^2&=& 	\|(\wh \L_1\wh \L_1-\wt\L_1\wt\L_1)+(\wt\L_1\wt\L_1-\L_1\L_1)\|_F^2 \nonumber\\
	&\leq& 2 \left[\|\wh \L_1\wh\L_1-\wt\L_1\wt\L_1\|_F^2+\|\wt\L_1\wt\L_1-\L_1\L_1\|_F^2\right] \nonumber\\
	&=& 2\|\wh \L_1\wh\L_1-\wt\L_1\wt\L_1\|_F^2+2\sum_{1\leq i,j\leq p} \delta_{i,j}^2 \nonumber\\
	&\leq& 2\|\wh \L_1\wh\L_1-\wt\L_1\wt\L_1\|_F^2+\frac{8p^2}{(p-1)^4l^4}.
	\end{eqnarray}
	Next, we derive the asymptotic bound for $\|\wh \L_1\wh\L_1-\wt\L_1\wt\L_1\|_F^2$.

	For any $1\leq i, j\leq p$, we denote the $(i,j)$th element of $\wh\L_1\wh\L_1-\wt\L_1\wt\L_1$ as $\Delta_{i,j}$. 
	By definition we have, 
	\begin{eqnarray*} 
		\Delta_{i,j}= \sum_{\substack{1\leq k\leq p \\ k\neq i,j}}\left( \frac{  \wh{\alpha}_{i,k}\wh{\alpha}_{k,j}}{\wh d_{k,1}\sqrt{\wh d_{i,1}\wh d_{j,1}}}-\frac{  {\alpha}_{i,k} {\alpha}_{k,j}}{d_{k,1}\sqrt{d_{i,1}d_{j,1}}}\right) ,
	\end{eqnarray*}
	where $\wh d_{\ell,1}=\sum_{k=1}^p\wh{\alpha}_{\ell,k}$ and $d_{\ell,1}=\sum_{k=1}^p {\alpha}_{\ell,k}$ for $l=1,\ldots, p$. Note that  $\wh{\alpha}_{i,1},\ldots, \wh{\alpha}_{i,p}$ are independent. Denote $\sigma^2_{i,k}:=Var(\wh{\alpha}_{i,k})$, and $\tau_i^2:=\sum_{k=1}^p \sigma^2_{i,k}$.  Similar to  the proofs of Lemma \ref{cov} we can show that, when $n$ is large enough, their exists a constant $C_\sigma>(2l)^{-1}$ and $c_\sigma:=l(1-l)$ such that $c_\sigma n^{-1}\leq \sigma^2_{i,k}\leq C_\sigma n^{-1}$ for any $(i,j)\in {\cal J}$. Consequently, $\tau_i^2\simeq	 O(n^{-1}p)$.  On the other hand, from Lemma \ref{bias} we know that there exists a large enough constant $C_\alpha>0$ such that $|E\wh\alpha_{i,j}-\alpha_{i,j}|\leq \frac{C_\alpha}{n}$ for all $(i,j)\in {\cal J}$, and consequently, $\frac{|E\wh d_{\ell,1}-  d_{\ell,1}|}{p}\leq \frac{|E\wh d_{\ell,1}-  \wt d_{\ell,1}|}{p} +\frac{1}{p} <\frac{C_\alpha}{n}+\frac{1}{p}$ for any $l=1,\ldots, p$.  We next break our proofs into three steps:
	
	\noindent
	{\bf Step 1.} Concentration of $p^{-1}\wh d_{\ell,1}$.
	
 Note that $|\hat{\alpha}_{\ell,j}|\le 1$. By Bernstein's inequality \citep{bennett1962probability, lin2011probability} we have, for any constant $C_d>0$:
	\begin{eqnarray}\label{Frob1:1}
	&&P\left(  \frac{|\wh d_{\ell,1}- d_{\ell,1}|}{p}\geq  C_d\sqrt{\frac{\log (pn)}{np} } +\frac{C_\alpha }{n} +\frac{1}{p}  \right)  \nonumber \\
	&\leq&P\left(  \frac{|\wh d_{\ell,1}-  E(\wh d_{\ell,1})|}{p}\geq C_d\sqrt{\frac{\log (pn)}{np} } +\frac{C_\alpha }{n}+\frac{1}{p}-\frac{|E(\wh d_{\ell,1})-d_{\ell,1}|}{p} \right)   \nonumber \\
	&\leq&P\left(  \frac{|\wh d_{\ell,1}-  E(\wh d_{\ell,1})|}{p}\geq  {C_d}  \sqrt{\frac{\log (pn)}{np} }   \right)  \nonumber \\
	&\leq& 2\exp\left\{
	-\frac{\sqrt{p}C_d^2n^{-1}\log (pn)}{2 (\sqrt{p}C_\sigma/n+aC_d \sqrt{ \log (pn)/n})} \right\}
	\nonumber \\
	&=& 2\exp\left\{
	-\frac{\sqrt{p}C_d^2n^{-1}\log (pn)}{2 (\sqrt{p}C_\sigma/n+C_d e (6l^{-1}+C_\alpha) \sqrt{ {\log n}/{(C_3n)}}   \sqrt{ \log (pn)/n})} \right\}.
	\end{eqnarray}
	When $\sqrt{p}C_\sigma/n> C_d e (6l^{-1}+C_\alpha) \sqrt{ {\log n}/{(C_3n)}}   \sqrt{ \log (pn)/n})  $, for any constant $B>0$, by choosing $C_d >2\sqrt{(B+1)C_\sigma}$, \eqref{Frob1:1} reduces to
	\begin{eqnarray}\label{Frob1:2} 
	&&P\left(  \frac{|\wh d_{\ell,1}- d_{\ell,1}|}{p}\geq  C_d\sqrt{\frac{\log (pn)}{np} } +\frac{C_\alpha }{n}+\frac{1}{p}  \right)  \nonumber  \\
	&\leq&   2\exp\left\{
	-\frac{\sqrt{p}C_d^2n^{-1}\log (pn)}{4  \sqrt{p}C_\sigma/n } \right\}<2(pn)^{-(B+1)}.
	\end{eqnarray}
	When $\sqrt{p}C_\sigma/n\leq  C_d e (6l^{-1}+C_\alpha) \sqrt{ {\log n}/{(C_3n)}}   \sqrt{ \log (pn)/n} $,  by choosing $C_d=4Be (6l^{-1}+C_\alpha)/\sqrt{C_3}$, \eqref{Frob1:1} reduces to
	\begin{eqnarray}\label{Frob1:3} 
	&&P\left(  \frac{|\wh d_{\ell,1}- d_{\ell,1}|}{p}\geq  C_d\sqrt{\frac{\log (pn)}{np} } +\frac{C_\alpha }{n} +\frac{1}{p} \right) \nonumber  \\
	&\leq&   2\exp\left\{
	-\frac{\sqrt{p}C_d^2n^{-1}\log (pn)}{4 C_d e (6l^{-1}+C_\alpha) \sqrt{ {\log n}/{(C_3n)}}   \sqrt{ \log (pn)/n} } \right\} \nonumber \\
	&\leq&  2\exp\left\{
	-B\sqrt{p}   \right\} .
	\end{eqnarray}
	From \eqref{Frob1:1}, \eqref{Frob1:2} and \eqref{Frob1:3} we conclude that for any $B>0$, by choosing $C_d$ to be large enough, we have, 
	\begin{eqnarray}\label{Frob1} 
	&& P\left(\max_{l=1,\ldots, p}  \frac{|\wh d_{\ell,1}- d_{\ell,1}|}{p}\geq  C_d\sqrt{\frac{\log (pn)}{np} } +\frac{C_\alpha }{n} +\frac{1}{p} \right) \nonumber \\
	&\leq& 2p\left[ (pn)^{-(1+B)}+  \exp\{
	-B\sqrt{p}\}\right].
	\end{eqnarray}
	
	\noindent
	{\bf Step 2.}  Concentration of $\Delta_{i,j}$.
	
	Using the fact that $\wh{\alpha}_{k,k}=0$ for $k=1,\ldots,p$, we have,
	\begin{eqnarray*} 
		\Delta_{i,j}&=& \sum_{k=1}^p\left(\frac{  \wh{\alpha}_{i,k}\wh{\alpha}_{k,j}}{\wh d_{k,1}\sqrt{\wh d_{i,1}\wh d_{j,1}}}
		-\frac{  \wh{\alpha}_{i,k}\wh{\alpha}_{k,j}}{d_{k,1}\sqrt{d_{i,1}d_{j,1}}}\right)
		+\sum_{\substack{1\leq k\leq p \\ k\neq i,j}} \left(\frac{  \wh{\alpha}_{i,k}\wh{\alpha}_{k,j}}{d_{k,1}\sqrt{d_{i,1}d_{j,1}}}
		-\frac{  {\alpha}_{i,k} {\alpha}_{k,j}}{d_{k,1}\sqrt{d_{i,1}d_{j,1}}}\right). 
	\end{eqnarray*}
	We next bound the two terms on the right hand side of the above inequality. 
	For the first term, denote $e_k:=(\wh d_{k,1} -d_{k,1})/p$. From \eqref{Frob1} we have there exists a large enough constant $C_{B}$ such that    
	\begin{eqnarray*}
		P\left\{\max_{k=1,\ldots, p}|e_k|\leq C_B\left(\sqrt{\frac{\log (pn)}{np}}+\frac{1}{n}+\frac{1}{p}\right)\right\}\geq 1-2p\left[ (pn)^{-(1+B)}+  \exp\{
		-B\sqrt{p}\}\right].
	\end{eqnarray*}
	Denote the event $\left\{\max_{k=1,\ldots, p}|e_k|\leq C_B\left(\sqrt{\frac{\log (pn)}{np}}+\frac{1}{n}+\frac{1}{p}\right)\right\}$ as ${\cal E}_B$. Under ${\cal E}_B$, we have, when $n$ and $p$ are large enough, $\sqrt{p^{-1}d_{k,1}+e_k}=\sqrt{p^{-1}d_{k,1}}+e_k/(2\sqrt{p^{-1}d_{k,1}})+O(e_k^2)$, and hence there exists a large enough constant $C_{l,B}>0$ such that for any $1\leq i,j\leq p$,
	\begin{eqnarray*}
		&&\left|\frac{  \wh{\alpha}_{i,k}\wh{\alpha}_{k,j}}{\wh d_{k,1}\sqrt{\wh d_{i,1}\wh d_{j,1}}}
		-\frac{  \wh{\alpha}_{i,k}\wh{\alpha}_{k,j}}{d_{k,1}\sqrt{d_{i,1}d_{j,1}}}\right| \\
		&\leq&  \frac{\left|p^{-1}d_{k,1}\sqrt{p^{-1}d_{i,1}p^{-1}d_{j,1}}-(p^{-1}d_{k,1}+e_k)\sqrt{(p^{-1}d_{i,1}+e_i)(p^{-1}d_{j,1}+e_j)}\right|}{p^2(p^{-1}d_{k,1}+e_k)\sqrt{(p^{-1}d_{i,1}+e_i)(p^{-1}d_{j,1}+e_j)    p^{-1}d_{k,1}\sqrt{p^{-1}d_{i,1}p^{-1}d_{j,1}} }}\\
		&=& O(p^{-2}(|e_i|+|e_j|+|e_k|) ) \\
		&\leq& \frac{C_{l,B}}{p^{2}}\left(\sqrt{\frac{\log (pn)}{np}}+\frac{1}{n}+\frac{1}{p}\right).
	\end{eqnarray*}
	Consequently, we have, under ${\cal E}_B$, 
	\begin{eqnarray}\label{1stterm}
	\left|\sum_{k=1}^p\left(\frac{  \wh{\alpha}_{i,k}\wh{\alpha}_{k,j}}{\wh d_{k,1}\sqrt{\wh d_{i,1}\wh d_{j,1}}}
	-\frac{  \wh{\alpha}_{i,k}\wh{\alpha}_{k,j}}{d_{k,1}\sqrt{d_{i,1}d_{j,1}}}\right)\right| \leq  \frac{C_{l,B}}{p}\left(\sqrt{\frac{\log (pn)}{np}}+\frac{1}{n}+\frac{1}{p}\right).
	\end{eqnarray}
	
	For the second term,  
note that for any $1\le i, j\le p$ and $k\ne i,j$, 
\begin{eqnarray}\label{midstep}
	&&	|E\wh{\alpha}_{i,k}\wh{\alpha}_{k,j}- {\alpha}_{i,k}{\alpha}_{k,j}| \nonumber \\
		&=& 		|E(\wh{\alpha}_{i,k}- {\alpha}_{i,k})(\wh{\alpha}_{k,j}-{\alpha}_{k,j})+ E(\wh{\alpha}_{i,k}-{\alpha}_{i,k}){\alpha}_{k,j}+E{\alpha}_{i,k}(\wh{\alpha}_{k,j}-{\alpha}_{k,j})|\nonumber	\\
		 &\le& |E(\wh{\alpha}_{i,k}- {\alpha}_{i,k})(\wh{\alpha}_{k,j}-{\alpha}_{k,j})|+ \frac{2C_\alpha}{n}. 
			\end{eqnarray}
		When $i\ne j$, by Lemma \ref{bias} and the fact that $\wh{\alpha}_{i,k}$ and $\wh{\alpha}_{k,j}$ are  independent  (since $k\ne i, j$), we have $|E\wh{\alpha}_{i,k}\wh{\alpha}_{k,j}- {\alpha}_{i,k}{\alpha}_{k,j}|\le C_{l,1}n^{-1}$ for some large enough constant $C_{l,1}>0$. 		
 Using the same arguments for obtaining   \eqref{Frob1}, we have, there exists a large enough constant $D_{l,B}>0$ such that when $n$ and $p$ are large enough,
		\begin{eqnarray}
	&& P\left(\max_{1\le i\ne j\le p}  \left|\sum_{\substack{1\leq k\leq p \\ k\neq i,j}} \left(\frac{  \wh{\alpha}_{i,k}\wh{\alpha}_{k,j}}{d_{k,1}\sqrt{d_{i,1}d_{j,1}}}
	-\frac{  {\alpha}_{i,k} {\alpha}_{k,j}}{d_{k,1}\sqrt{d_{i,1}d_{j,1}}}\right)\right| \geq  \frac{D_{l,B}}{p}\Bigg(
	 \sqrt{\frac{\log (pn)}{np} } +\frac{ 1 }{n} +\frac{1}{p} \Bigg)\right) \nonumber \\
	&\leq& 2p\left[ (pn)^{-(1+B)}+  \exp\{
	-B\sqrt{p}\}\right].  \label{2ndterm}
	\end{eqnarray}
	Denote the event $\left\{\underset{{1\le i\ne j\le p}}{\max}  \left|\sum_{\substack{1\leq k\leq p \\ k\neq i,j}} \left(\frac{  \wh{\alpha}_{i,k}\wh{\alpha}_{k,j}}{d_{k,1}\sqrt{d_{i,1}d_{j,1}}}
	-\frac{  {\alpha}_{i,k} {\alpha}_{k,j}}{d_{k,1}\sqrt{d_{i,1}d_{j,1}}}\right)\right| \le \frac{D_{l,B}}{p} \Bigg(
	\sqrt{\frac{\log (pn)}{np} } +\frac{ 1 }{n} +\frac{1}{p} \Bigg) \right\}$
	as ${\cal A}_{B}$.  
	From \eqref{1stterm} and \eqref{2ndterm} we conclude that, when $n$ and $p$ are  large enough, 
	\begin{eqnarray}\label{deltaijCon}
	&&P\left\{ \max_{1\leq i\ne j\leq p}|\Delta_{i,j}| > \frac{C_{l,B}+D_{l,B}}{p}\left(\sqrt{\frac{\log (pn)}{np}}+\frac{1}{n}+\frac{1}{p}\right)\right\} \nonumber \\ 
	&\leq& P({\cal E}_B^c)+P({\cal A}_B^c) \nonumber\\
	&< &4p\left[ (pn)^{-(1+B)}+  \exp\{
	-B\sqrt{p}\}\right].
	\end{eqnarray}
	When $i=j$, 
	by applying Lemma \ref{bias} and \eqref{alphaCon} to \eqref{midstep}, we have, there exists a large enough constant $C_{l,2}>0$, such that 
		\begin{eqnarray*}
	 	|E\wh{\alpha}_{i,k}\wh{\alpha}_{k,i}- {\alpha}_{i,k}{\alpha}_{k,i}|  
	 \le C_{l,2}\left(\frac{\log (pn)}{n}+\frac{1}{n}+\frac{1}{p}\right).
		\end{eqnarray*}
		Consequently, similar to \eqref{deltaijCon}, we have,  there exists a large enough constant $C_{l,3}>0$, such that 
	\begin{eqnarray}\label{deltaijCon2}
	P\left\{ \max_{1\leq i\leq p}|\Delta_{i,i}| > \frac{C_{l,3}}{p}\left(\sqrt{\frac{\log (pn)}{np}}+\frac{\log (pn)}{n}+\frac{1}{p}\right)\right\} 
	< 4p\left[ (pn)^{-(1+B)}+  \exp\{
	-B\sqrt{p}\}\right].
	\end{eqnarray}
	
	{\bf Step 3.} Proof of the first inequality in \eqref{lemFrob1}. 
	
	Note that $\|\wh \L_1\wh \L_1-\wt \L_1\wt \L_1\|_F=\sqrt{\sum_{1\leq i,j\leq p}\Delta_{i,j}^2}\leq  p \max_{1\leq i\neq j\leq p}|\Delta_{i,j}|+\sqrt{p}\max_{1\leq i \leq p}|\Delta_{i,i}|$. 
 From \eqref{De_diag}, \eqref{deltaijCon}, \eqref{deltaijCon2} 
 and the fact that 
 $\frac{1}{\sqrt{p}}\left(\sqrt{\frac{\log (pn)}{np}}+\frac{\log (pn)}{n}+\frac{1}{p}\right)=
 o\left(\sqrt{\frac{\log (pn)}{np} } +\frac{1}{n}+\frac{1}{p}\right)$
 we immediately have that there exists a large enough constant $C>0$ such that when $n$ and $p$ are large enough,      
	\begin{eqnarray*}
		P\left\{  \|\wh \L_1\wh \L_1-\L_1\L_1\|_F\geq C\left(\sqrt{\frac{\log (pn)}{np} } +\frac{1}{n}+\frac{1}{p}\right) \right\}\leq 8p\left[ (pn)^{-(1+B)}+  \exp\{
		-B\sqrt{p}\}\right].
	\end{eqnarray*}
	This proves the first inequality in  \eqref{lemFrob1}. 
	
\end{proof}

\begin{lemma}\label{FrobF}
	Let conditions (2.5), C1 and C2 hold.  
	For any constant $B>0$, there exists a large enough constant $C>0$ such that 
	\begin{eqnarray}\label{FrobF:eq}
	\\
	P\left\{ \|\wh\L\wh\L-\L\L\|_F\geq 4C\left(\sqrt{\frac{\log (pn)}{np} } +\frac{1}{n}+\frac{1}{p}\right) \right\}\leq 16p\left[ (pn)^{-(1+B)}+  \exp\{
	-B\sqrt{p}\}\right]. \nonumber
	\end{eqnarray}
\end{lemma}
\begin{proof}
	
	Note that from the triangle inequality we have
	\begin{eqnarray*} 
		&&\|\wh\L\wh\L-\L\L\|_F  \\
		&=& \|(\wh\L_1+\wh\L_2)(\wh\L_1+\wh\L_2)-(\L_1+\L_2)(\L_1+\L_2)\|_F   \nonumber\\
		&=& \|(\wh\L_1\wh\L_1-\L_1\L_1)+(\wh\L_1\wh\L_2-\L_1\L_2)+(\wh\L_2\wh\L_1-\L_2\L_1)+(\wh\L_2\wh\L_2-\L_2\L_2)\|_F    \\
		&\leq& \|\wh \L_1\wh\L_1-\L_1\L_1\|_F+ \|\wh\L_1\wh\L_2-\L_1\L_2\|_F+ \|\wh\L_2\wh\L_1-\L_2\L_1\|_F+ \|\wh\L_2\wh\L_2-\L_2\L_2\|_F.
	\end{eqnarray*}
	Together with Lemma \ref{Frob}  we immediately conclude that \eqref{FrobF:eq} hold.

\end{proof}
\noindent

\noindent
{\bf Proof of Theorem \ref{DKthm}}

From Weyl’s inequality and Lemma \ref{FrobF}, we have,  
\begin{eqnarray*}
	\max_{i=1,\ldots, p} |\lambda_i^2-\wh\lambda_i^2|\leq  \|\wh\L\wh\L-\L\L\|_2\leq \|\wh\L\wh\L-\L\L\|_F=O_p\left(\sqrt{\frac{\log (pn)}{np} } +\frac{1}{n}+\frac{1}{p}\right) . 
\end{eqnarray*}
\eqref{Con_EVec} is a direct result of the Davis-Kahan theorem \citep{rohe2011spectral, yu2015useful} theorem and Lemma \ref{FrobF}.

\subsection{Proof of Theorem \ref{clusteringCons}}

Recall that $\bGamma_q=\Z\bU$ where $\bU$ is defined as in the proof of Proposition \ref{OracleSC}.  
For any $1\leq i\neq j\leq n$ such that $\z_i\neq \z_j$, we need to show that $\|\z_i\bU\bO_q-\z_j\bU_q\bO_q\|_2=\|\z_i\bU  -\z_j\bU\|_2$ is large enough, so that the perturbed version (i.e. the rows of $\wh\bGamma_q$) is not changing the clustering structure.

Denote the $i$th row of ${\bGamma}_q\bO_q$ and $\wh{\bGamma}_q$ as $\bgamma_i$ and  $\wh\bgamma_i$, respectively, for $i=1, \ldots, p$. 
Notice that from the proof of Proposition \ref{OracleSC}, we have $\bU \bU^\top=  \N^{-1}\bQ \bQ^\top \N^{-1} =\N^{-2}=\diag\{s_1^{-1},\ldots, s_q^{-1}\}$. Consequently, for any $\z_i\neq \z_j$,  we have:
\begin{eqnarray}\label{kmean1}
\|\bgamma_i-\bgamma_j\|_2=\|\z_i\bU\bO_q-\z_j\bU_q\bO_q\|_2=\|\z_i\bU  -\z_j\bU\|_2\geq \sqrt{\frac{2}{s_{\max}}}. 
\end{eqnarray}
%

We first show that  $\z_i\neq \z_j$ implies $\wh\c_i\neq \wh\c_j$. 
Notice that ${\bGamma}_q\bO_q\in  {\cal M}_{p,q}$. Denote $\wh\bC=(\wh \bc_1, \cdots, \wh \bc_p)^\top$. By the definition of $\wh\bC$ we have
\begin{eqnarray}\label{kmean2}
\|{\bGamma}_q\bO_q-\wh\bC\|_F^2  \leq   
\|\wh\bGamma_q -\wh\bC\|_F^2 + \|\wh\bGamma_q -{\bGamma}_q\bO_q\|_F^2  
\leq  
2\|\wh\bGamma_q -{\bGamma}_q\bO_q\|_F^2.
\end{eqnarray}
Suppose there exist $i,j\in \{1,\ldots, p\}$ such that $\z_i\neq \z_j$ but $\wh\c_i=\wh\c_j$. We have
\begin{eqnarray}\label{kmean3}
\|{\bGamma}_q\bO_q-\wh\bC\|_F^2 \geq \|\z_i\bU\bO_q -\wh\bc_i \|_2^2+\|\z_j\bU\bO_q -\wh\bc_j \|_2^2 \geq  \|\z_i\bU\bO_q-\z_j\bU\bO_q\|_2^2.
\end{eqnarray}
Combining \eqref{kmean1},  \eqref{Con_EVec}, \eqref{kmean2} and \eqref{kmean3}, we  have: 
\[
\sqrt{\frac{2}{s_{\max}}}\leq \|{\bGamma}_q\bO_q-\wh\bC\|_F \leq \sqrt{2} \|\wh\bGamma_q -{\bGamma}_q\bO_q\|_F  
\leq 4\sqrt{2} \lambda_{q}^{-2}C\left(\sqrt{\frac{\log (pn)}{np} } +\frac{1}{n}+\frac{1}{p}\right). 
\]
We have reach a contradictory with \eqref{KM_SN}. Therefore we conclude that $\wh\c_i\neq \wh\c_j$.  

Next we show that if $\z_i=\z_j$ we must have $\wh\c_i=\wh\c_j$.  Assume that there exist $1\leq i\neq j\leq p$ such that $\z_i=\z_j$ and $\wh\c_i\neq \wh\c_j$. 
Notice that from the previous conclusion (i.e.,  that different $z_i$ implies different $\wh\c_i$), since there are $q$ distinct rows in $\Z$,   there are correspondingly  $q$ different rows in $\wh\bC$. Consequently for any $\z_i=\z_j$, if $\wh\c_i\neq\wh\c_j$ there must exist a $k\neq i,j$ such that $\z_i=\z_j\neq \z_k$ and $\wh\c_j=\wh\c_k$. Let $\wh\bC^*$ be $\wh\bC$ with the $j$th row   replaced by $ \wh\c_i$. We have 
\begin{eqnarray*}
	&&\|\wh\bGamma_q -\wh\bC^*\|_F^2-\|\wh\bGamma_q -\wh\bC\|_F^2  \\
	&=& \| \wh\bgamma_j-\wh\c_i\|_2^2- \|\wh\bgamma_j -\wh\c_k\|_2^2   \\
	&=& \| \wh\bgamma_j -\bgamma_j+\bgamma_i -\wh\c_i\|_2^2 -\| \wh\bgamma_j -\bgamma_j+\bgamma_i-\bgamma_k+\bgamma_k-\wh\c_k\|_2^2  \\
	&\leq& \| \wh\bgamma_j -\bgamma_j+\bgamma_i -\wh\c_i\|_2^2 +\| \wh\bgamma_j -\bgamma_j+\bgamma_k -\wh\c_k\|_2^2 -\|\bgamma_i-\bgamma_k\|_2^2\\
	&\leq& 2\|\wh\bGamma_q-\bGamma_q\bO_q\|_F^2+\|\bGamma_q\bO_q-\wh\bC\|_F^2  -\frac{2}{s_{\max}} \\
	& \leq& 4\left\{ 4 \lambda_{q}^{-2}C\left(\sqrt{\frac{\log (pn)}{np} } +\frac{1}{n}+\frac{1}{p}\right)   \right\}^2 -\frac{2}{s_{\max}} \\
	&<&0.
\end{eqnarray*}
Again, we reach a contradiction and so we conclude that if $\z_i=\z_j$ we must have $\wh\c_i=\wh\c_j$. 

\subsection{Proof of Theorem \ref{CLT2}}

Note that from Theorem \ref{clusteringCons}, we have the memberships can be recovered with probability tending to 1, i,e, $P(\wh\nu\neq \nu )\rightarrow 0$. On the other hand, given $\wh\nu= \nu$, we have,  the log likelihood
function of $(\theta_{k, \ell}, \eta_{k, \ell})$, $1\le k \le \ell \le q$, is
\begin{eqnarray*}
	l( \{ \theta_{k, \ell}, \eta_{k, \ell}\}; \nu) &=& 
	\sum_{(i,j)\in S_{k,l}} \sum_{t=1}^n
	\Big\{ X_{i,j}^t(1-X_{i,j}^{t-1})  \log \theta_{k,\ell} 	+  (1-X_{i,j}^t)(1-X_{i,j}^{t-1})  \log (1 -\theta_{k,\ell})\\
	&& 
	+  (1-X_{i,j}^t) X_{i,j}^{t-1} \log \eta_{k, \ell} 
	+ X_{i,j}^t X_{i,j}^{t-1} \log(1- \eta_{k, \ell}) \Big\}.
\end{eqnarray*}
Using the same arguments   as in the proof of Proposition \ref{CLT1}, we can conclude that when  $\wh\nu= \nu$,  $\sqrt{n}\N_{J_1,J_2}^{\frac{1}{2}}(\wh{\bPsi}_{{\cal K}_1, {\cal K}_2}-\bPsi_{{\cal K}_1,{\cal K}_2}) \rightarrow N({\bf 0}, \wt\bSigma_{{\cal K}_1, {\cal K}_2})$. 
Let $\bY\sim N({\bf 0}, \wt\bSigma_{{\cal K}_1, {\cal K}_2})$. For any ${\cal Y}\subset {\cal R}^{m_1+m_2}$, let $\bPhi({\cal Y}):=P(\bY \in {\cal Y} )$, we have:
\begin{eqnarray*}
	&&	| P( \sqrt{n} \N_{{\cal K}_1,{\cal K}_2}^{\frac{1}{2}}(\wh{\bPsi}_{{\cal K}_1, {\cal K}_2}-\bPsi_{{\cal K}_1,{\cal K}_2}) \in {\cal Y}) - \bPhi({\cal Y})| \\
	&\leq& P(\wh\nu\neq \nu )+ 	|P(\sqrt{n}  \N_{{\cal K}_1,{\cal K}_2}^{\frac{1}{2}}(\wh{\bPsi}_{{\cal K}_1, {\cal K}_2}-\bPsi_{{\cal K}_1,{\cal K}_2}) \in {\cal Y}| \wh\nu = \nu) - \bPhi({\cal Y})| \\
	&=& o(1). 
\end{eqnarray*}
This proves the theorem. 

\subsection{Proof of Theorem \ref{Thm_CP}}

Without loss of generality, we consider the case where $\tau\in [n_0,\tau_0]$, as the convergence rate  for  $\tau\in [\tau_0,n-n_0]$ can be similarly  derived. 
The idea is to break the time interval $[ n_0,\tau_0]$ into two consecutive parts: $[n_0, \tau_{n,p}]$ and $[\tau_{n,p}, \tau_0]$, where $\tau_{n,p}=\Big\lfloor \tau_0- \kappa n \Delta_F^{-2}   \Big[  {\frac{  \log (np )}{n  } } +\sqrt{\frac{  \log (np )}{np^2 } }   \Big]
\Big\rfloor$ for some large enough $\kappa>0$. Here $\lfloor\cdot \rfloor$ denotes the least integer function.   
We shall show that when $\tau\in[n-n_0, \tau_{n,p}]$, in which $\wh\nu^{\tau+1,n}$ might be inconsistent in estimating $\nu^{\tau_0+1,n}$,  
we have $\sup_{\tau\in[n_0, \tau_{n,p}] }[\mM_n(\tau)-\mM_n(\tau_0)]<0$   in probability.  Hence 
$\argmax_{\tau\in [n_0, \tau_0]} \mM_n(\tau)= \argmax_{\tau\in [\tau_{n,p}, \tau_0]} \mM_n(\tau)$ holds in probability. On the other hand, when $\tau\in [\tau_{n,p}, \tau_0]$,  we shall see that the membership maps can be consistently recovered, and hence the convergence rate can be obtained using classical probabilistic arguments. For simplicity, we consider the case where $\nu^{1,\tau_0}=\nu^{\tau_0+1,n}=\nu$ first, and modification of the proofs for the case where  $\nu^{1,\tau_0}\neq \nu^{\tau_0+1,n}$ will be provided subsequently.

\subsubsection{Change point estimation with $\nu^{1,\tau_0}=\nu^{\tau_0+1,n}=\nu$.}

We first consider the case where the membership structures remain unchanged, while the connectivity matrices before/after the change point are different. Specifically,  we assume that 
$\nu^{1,\tau_0}=\nu^{\tau_0+1,n}=\nu$ for some $\nu$, and $(\theta_{1,k,\ell}, \eta_{1,k,\ell})\neq (\theta_{2,k,\ell}, \eta_{2,k,\ell})$ for some  $1\leq k\leq  l\leq q$.     For brevity, we shall be using the notations $S_{k,l}$, $s_k$, $s_{\min}$ and  $n_{k,\ell}$  defined as in Section \ref{sec3},
and introduce some new notations as follows:

Define
\[
\theta_{2,k,\ell}^\tau= \frac{ \frac{\tau_0-\tau}{n-\tau}\frac{\theta_{1,k,\ell}\eta_{1,k,\ell}}{\theta_{1,k,\ell}+\eta_{1,k,\ell}}+\frac{n-\tau_0 }{n-\tau}\frac{\theta_{2,k,\ell}\eta_{2,k,\ell}}{\theta_{2,k,\ell}+\eta_{2,k,\ell}}  }{
	\frac{\tau_0-\tau}{n-\tau}\frac{ \eta_{1,k,\ell}}{\theta_{1,k,\ell}+\eta_{1,k,\ell}}+\frac{n-\tau_0 }{n-\tau}\frac{ \eta_{2,k,\ell}}{\theta_{2,k,\ell}+\eta_{2,k,\ell}} 	
} 
, \quad \eta_{2,k,\ell}^\tau=\frac{ \frac{\tau_0-\tau}{n-\tau}\frac{\theta_{1,k,\ell}\eta_{1,k,\ell}}{\theta_{1,k,\ell}+\eta_{1,k,\ell}}+\frac{n-\tau_0 }{n-\tau}\frac{\theta_{2,k,\ell}\eta_{2,k,\ell}}{\theta_{2,k,\ell}+\eta_{2,k,\ell}}  }{
	\frac{\tau_0-\tau}{n-\tau}\frac{ \theta_{1,k,\ell}}{\theta_{1,k,\ell}+\eta_{1,k,\ell}}+\frac{n-\tau_0 }{n-\tau}\frac{ \theta_{2,k,\ell}}{\theta_{2,k,\ell}+\eta_{2,k,\ell}} 	
} . 
\]
Clearly when $\tau=\tau_0$  we have $\theta_{2,k,\ell}^{\tau_0} = \theta_{2,k,\ell}$ and $\eta_{2,k,\ell}^{\tau_0} = \eta_{2,k,\ell}$.

Correspondingly, we denote the MLEs  as
\begin{align*}  
\wh \theta_{1, k, \ell}^\tau& ={
	\sum_{(i,j)\in \wh S^\tau_{
			1, k,\ell} }
	\sum_{t=1}^{\tau} X_{i,j}^t(1 - X_{i,j}^{t-1}) \Big/ \hspace{-3mm}
	\sum_{(i,j)\in \wh S^\tau_{1, k,\ell} }
	\sum_{t=1}^\tau (1 - X_{i,j}^{t-1})},\\[1ex]  
\wh \eta^\tau_{1, k, \ell}& =
\sum_{(i,j)\in \wh S^\tau_{1,k,\ell} }
\sum_{t=1}^{\tau}  (1- X_{i,j}^t)X_{i,j}^{t-1} \Big/ \hspace{-3mm}
\sum_{(i,j)\in \wh S^\tau_{1,k,\ell} }
\sum_{t=1}^\tau X_{i,j}^{t-1},  \\  
\wh \theta^\tau_{2, k, \ell}& ={
	\sum_{(i,j)\in \wh S^\tau_{
			2, k,\ell} }
	\sum_{t=\tau+1}^{n} X_{i,j}^t(1 - X_{i,j}^{t-1}) \Big/ \hspace{-3mm}
	\sum_{(i,j)\in \wh S^\tau_{2, k,\ell} }
	\sum_{t=\tau+1}^n (1 - X_{i,j}^{t-1})},\\[1ex]  
\wh \eta^\tau_{2, k, \ell}& =
\sum_{(i,j)\in \wh S^\tau_{2,k,\ell} }
\sum_{t=\tau+1}^{n}  (1- X_{i,j}^t)X_{i,j}^{t-1} \Big/ \hspace{-3mm}
\sum_{(i,j)\in \wh S^\tau_{2,k,\ell} }
\sum_{t=\tau+1}^n X_{i,j}^{t-1},
\end{align*} 
where $\wh S^\tau_{1, k,\ell} $  and $\wh S^\tau_{2, k,\ell} $ are defined in a similar way to $\wh S_{k,\ell}$ (cf. Section \ref{se322}),  based on the estimated memberships $\wh\nu^{1,\tau}$ and $\wh\nu^{\tau+1,n}$, respectively.

Denote 
\begin{align*} 
&\mM_n(\tau):=    l( \{ \wh\theta^\tau_{1, k, \ell}, \wh\eta^\tau_{1, k, \ell}\};\; \wh \nu^{1,\tau})+ l( \{ \wh\theta^\tau_{2, k, \ell}, \wh\eta^\tau_{2, k, \ell}\};\; \wh \nu^{\tau+1,n}), \\ 
& \mM(\tau):= E  l( \{ \theta_{1,k, \ell}, \eta_{1,k, \ell}\};  \nu^{1,\tau})+E  l( \{ \theta_{2,k,\ell}^\tau, \eta_{2,k,\ell}^\tau \};  \nu^{\tau+1,n}).
\end{align*}
We first evaluate several terms in (i)-(v), and all these results will be combined to obtain the error bound in (vi). 
In particular, (vi) states that as a direct result of (v), we can focus on the small neighborhood of $[\tau_{n,p},\tau_0]$ when searching for the estimator $\wh\tau$.  Further, the inequality \eqref{vi_key} transforms the error bound for $\tau_0-\wh\tau$ into the error bounds of the terms that we derived in (i)-(iv).

\noindent
{\bf (i) Evaluating $\mM(\tau)-\mM(\tau_0)$.}  

\noindent
Note that $\tau_0=\arg \max_{n_0 \le \tau \le n - n_0} \mM(\tau)$, and for any $\tau\in [n_0 , \tau_0]$,
\begin{eqnarray*}
	\mM(\tau)-\mM(\tau_0) &= &E  l( \{ \theta_{1,k, \ell}, \eta_{1,k, \ell}\};  \nu^{1,\tau})+E  l( \{ \theta^\tau_{2,k,\ell}, \eta^\tau_{2,k,\ell}\};  \nu^{\tau+1,n})\\
	&&- E  l( \{ \theta_{1,k, \ell}, \eta_{1,k, \ell}\};  \nu^{1,\tau_0})-E  l( \{ \theta_{2,k,\ell}, \eta_{2,k,\ell}\};  \nu^{\tau_0+1,n})    \\
	&=&  E  l( \{ \theta_{2,k,\ell}^\tau, \eta_{2,k,\ell}^\tau \};  \nu^{\tau+1,\tau_0})-E  l( \{ \theta_{1,k,\ell}, \eta_{1,k,\ell}\};  \nu^{\tau+1,\tau_0}) \\
	&&+ E  l( \{ \theta_{2,k,\ell}^\tau, \eta_{2,k,\ell}^\tau \};  \nu^{\tau_0+1,n})-E  l( \{ \theta_{2,k,\ell}, \eta_{2,k,\ell}\};  \nu^{\tau_0+1,n}).
\end{eqnarray*}

Recall that 
\begin{align*}
l( \{ \theta_{k, \ell}, \eta_{k, \ell}\};   \nu) &= 
\sum_{1 \le k \le \ell \le q}
\sum_{(i,j)\in   S_{k,l} } \sum_{t=1}^n
\Big\{ X_{i,j}^t(1-X_{i,j}^{t-1})  \log \theta_{k,\ell} \\
+  (1-X_{i,j}^t)&(1-X_{i,j}^{t-1})  \log (1 -\theta_{k,\ell})
+  (1-X_{i,j}^t) X_{i,j}^{t-1} \log \eta_{k, \ell} 
+ X_{i,j}^t X_{i,j}^{t-1} \log(1- \eta_{k, \ell}) \Big\}.
\end{align*}
By Taylor expansion and the fact that the partial derivative of the expected likelihood evaluated at the true values equals zero we have, there exist $\theta^*_{k,\ell} \in [\theta_{1,k,\ell}, \theta^\tau_{2, k,\ell}],\eta_{k,\ell}^*\in [\eta_{1,k,\ell}, \eta^\tau_{2, k,\ell}], 1\leq k\leq \ell\leq q$, such that
\begin{eqnarray*}
	&&E  l( \{ \theta^\tau_{2,k,\ell}, \eta^\tau_{2,k,\ell}\};  \nu^{\tau+1,\tau_0})-E  l( \{ \theta_{1,k,\ell}, \eta_{1,k,\ell}\};  \nu^{\tau+1,\tau_0})  \\
	&=&-\frac{1}{2} \sum_{1 \le k \le \ell \le q}
	n_{k,\ell}(\tau_0-\tau)  \Bigg\{  
	\frac{\theta_{1,k,\ell}\eta_{1,k,\ell}}{\theta_{1,k,\ell}+\eta_{1,k,\ell}} \Big(\frac{\theta^\tau_{2,k,\ell}-\theta_{1,k,\ell}}{\theta_{k,\ell}^*}\Big)^2
	+ \frac{(1-\theta_{1,k,\ell})\eta_{1,k,\ell}}{\theta_{1,k,\ell}+\eta_{1,k,\ell}} \Big(\frac{\theta^\tau_{2,k,\ell}-\theta_{1,k,\ell}}{1-\theta_{k,\ell}^*}\Big)^2  \\
	&&+  \frac{\theta_{1,k,\ell}\eta_{1,k,\ell}}{\theta_{1,k,\ell}+\eta_{1,k,\ell}} \Big(\frac{\eta^\tau_{2,k,\ell}-\eta_{1,k,\ell}}{\eta_{k,\ell}^*}\Big)^2
	+ \frac{(1-\eta_{1,k,\ell})\theta_{1,k,\ell}}{\theta_{1,k,\ell}+\eta_{1,k,\ell}} \Big(\frac{\eta^\tau_{2,k,\ell}-\eta_{1,k,\ell}}{1-\eta_{k,\ell}^*}\Big)^2 
	\Bigg\} \\
	&\leq& - C_1(\tau_0-\tau)\sum_{1\leq k\leq \ell\leq q } n_{k,\ell} [(\theta_{1,k,\ell}-\theta_{2,k,\ell})^2 +(\eta_{1,k,\ell}-\eta_{2,k,\ell})^2]  \\
	&\leq&-  C_1(\tau_0-\tau)\big[ \| \bW_{1,1}-\bW_{2,1}\|_F^2+\| \bW_{1,2}-\bW_{2,2}\|_F^2\big],
\end{eqnarray*}
for some constant $C_1>0$.
Here in the first step we have used the fact that for any $(i,j)\in S_{k,\ell}$ and $t\leq\tau_0$, $E X_{i,j}^t(1-X_{i,j}^{t-1})=E X_{i,j}^{t-1}(1-X_{i,j}^{t})= \frac{\theta_{1,k,\ell}\eta_{1,k,\ell}}{\theta_{1,k,\ell}+\eta_{1,k,\ell}}$, $E (1-X_{i,j}^t)(1-X_{i,j}^{t-1})= \frac{(1-\theta_{1,k,\ell})\eta_{1,k,\ell}}{\theta_{1,k,\ell}+\eta_{1,k,\ell}}$, and 
$E X_{i,j}^tX_{i,j}^{t-1}=\frac{(1-\eta_{1,k,\ell})\theta_{1,k,\ell}}{\theta_{1,k,\ell}+\eta_{1,k,\ell}}$. Similarly, there exist $\theta^\dagger_{k,\ell} \in [\theta_{2,k,\ell}, \theta^\tau_{2, k,\ell}],\eta_{k,\ell}^\dagger \in [\eta_{2,k,\ell}, \eta^\tau_{2, k,\ell}], 1\leq k\leq \ell\leq q$, such that
\begin{eqnarray*}
	&& E  l( \{ \theta^\tau_{2,k,\ell}, \eta^\tau_{2,k,\ell}\};  \nu^{\tau_0+1,n})-E  l( \{ \theta_{2,k,\ell}, \eta_{2,k,\ell}\};  \nu^{\tau_0+1,n}) \\
	&=&- \frac{1}{2}\sum_{1 \le k \le \ell \le q}
	n_{k,\ell}(n-\tau_0 )  \Bigg\{  
	\frac{\theta_{2,k,\ell}\eta_{2,k,\ell}}{\theta_{2,k,\ell}+\eta_{2,k,\ell}} \Big(\frac{\theta^\tau_{2,k,\ell}-\theta_{2,k,\ell}}{\theta_{k,\ell}^\dagger}\Big)^2
	+ \frac{(1-\theta_{2,k,\ell})\eta_{2,k,\ell}}{\theta_{2,k,\ell}+\eta_{2,k,\ell}} \Big(\frac{\theta^\tau_{2,k,\ell}-\theta_{2,k,\ell}}{1-\theta_{k,\ell}^\dagger}\Big)^2  \\
	&&+  \frac{\theta_{2,k,\ell}\eta_{2,k,\ell}}{\theta_{2,k,\ell}+\eta_{2,k,\ell}} \Big(\frac{\eta^\tau_{2,k,\ell}-\eta_{2,k,\ell}}{\eta_{k,\ell}^\dagger}\Big)^2
	+ \frac{(1-\eta_{2,k,\ell})\theta_{2,k,\ell}}{\theta_{2,k,\ell}+\eta_{2,k,\ell}} \Big(\frac{\eta^\tau_{2,k,\ell}-\eta_{2,k,\ell}}{1-\eta_{k,\ell}^\dagger}\Big)^2 
	\Bigg\} \\
	&\leq& - C_2'(n-\tau_0)\sum_{1\leq k\leq \ell\leq q } \frac{n_{k,\ell}(\tau_0-\tau)^2}{(n-\tau)^2} [(\theta_{1,k,\ell}-\theta_{2,k,\ell})^2 +(\eta_{1,k,\ell}-\eta_{2,k,\ell})^2]  \\
	&\leq& - \frac{ C_2(\tau_0-\tau)^2}{n-\tau} \big[ \| \bW_{1,1}-\bW_{2,1}\|_F^2+\| \bW_{1,2}-\bW_{2,2}\|_F^2\big],
\end{eqnarray*}
for some constants $C_2', C_2>0$.
Consequently, we conclude that there exists a constant $C_3>0$ such that for any $n_0\leq \tau\leq \tau_0$, we have
\begin{eqnarray}\label{Con1}
\mM(\tau)-\mM(\tau_0) \leq -C_3 (\tau_0-\tau)\big[ \| \bW_{1,1}-\bW_{2,1}\|_F^2+\| \bW_{1,2}-\bW_{2,2}\|_F^2\big].
\end{eqnarray}

\noindent
{\bf (ii)  Evaluating  $\sup_{\tau\in [\tau_{n,p}, \tau_0]}P(\wh\nu(\tau)\neq \nu) $. }  

\noindent
Let $\wh\nu(\tau)$ be either $\wh\nu^{1,\tau}$ or $\wh\nu^{\tau+1,n}$.  Note that the membership maps of the networks before/after $\tau$ remain to be $\nu$. 
From Theorems \ref{DKthm} and  \ref{clusteringCons}, we have, under conditions C2-C4,  for any constant $B>0$,  there exists a large enough constant $C_B$ such that
\begin{eqnarray*}
	\sup_{\tau\in [\tau_{n,p}, \tau_0]}P(\wh\nu(\tau)\neq \nu) \leq  C_B(\tau_0-\tau_{n,p}) p[(pn)^{-(B+1)}+\exp\{-B\sqrt{p}\}] .
\end{eqnarray*}
Note that by choosing  $B$ to be large enough, we have  $p(\tau_0-\tau_{n,p})(pn)^{-(B+1)}=o\left( \sqrt{ \frac{(\tau_0-\tau_{n,p})\log (np)}{n^2s_{\min}^2}} \right)$. On the other hand,    the assumption that $\frac{\log (np)}{\sqrt{p}}\rightarrow 0$ in condition C4 implies $pn\sqrt{\frac{(\tau_0-\tau_{n,p})s_{\min}^2}{\log (np)} }=o( \exp\{B\sqrt{p}\})$ for some large enough constant $B$. Consequently, we have $(\tau_0-
\tau_{n,p})p\exp\{-B\sqrt{p}\} =o\left( \sqrt{ \frac{(\tau_0-\tau_{n,p})\log (np)}{n^2s_{\min}^2}} \right)$, and hence we conclude that $\sup_{\tau\in [\tau_{n,p}, \tau_0]}P(\wh\nu(\tau)\neq \nu)= o\left( \sqrt{ \frac{(\tau_0-\tau_{n,p})\log (np)}{n^2s_{\min}^2}} \right)$.

\noindent
{\bf (iii) Evaluating $\sup_{\tau\in [\tau_{n,p}, \tau_0]} [\mM_n(\tau)-\mM(\tau) ]$ when $\wh\nu(\tau)=\nu$.}  

\noindent
From (ii) we have with probability greater than $1-o\left( \sqrt{ \frac{(\tau_0-\tau_{n,p})\log (np)}{n^2s_{\min}^2}} \right)$, $\wh\nu(\tau)=\nu$ for all $\tau\in [\tau_{n,p}, \tau_0]$. 
For simplicity, in this part we assume that  $\wh S^\tau_{1,k,\ell}= \wh S^\tau_{2,k,\ell}=S_{k,l}$  (or equivalently $\wh{\nu}^{1,\tau}=\wh{\nu}^{\tau+1,n}=\nu$) holds for all $1\leq k\leq \ell\leq q$ and $\tau_{n,p}\leq \tau\leq \tau_0$  without indicating that this holds in probability.

Denote
\begin{align*}
g_{1, i,j }(\theta, \eta;\tau)&= 
\sum_{t=1}^\tau
\Big\{ X_{i,j}^t(1-X_{i,j}^{t-1})  \log  \theta  \\
+  (1-X_{i,j}^t)&(1-X_{i,j}^{t-1})  \log (1 -\theta)
+  (1-X_{i,j}^t) X_{i,j}^{t-1} \log \eta 
+ X_{i,j}^t X_{i,j}^{t-1} \log(1- \eta ) \Big\}, 
\end{align*}
and 
\begin{align*}
g_{2, i,j }(\theta, \eta;\tau)&= 
\sum_{t=\tau+1}^n
\Big\{ X_{i,j}^t(1-X_{i,j}^{t-1})  \log  \theta  \\
+  (1-X_{i,j}^t)&(1-X_{i,j}^{t-1})  \log (1 -\theta)
+  (1-X_{i,j}^t) X_{i,j}^{t-1} \log \eta 
+ X_{i,j}^t X_{i,j}^{t-1} \log(1- \eta ) \Big\}.
\end{align*}
When $\wh \nu=\nu$, we have,
\begin{eqnarray}\label{Mn}
\\
&&	\mM_n(\tau)-\mM(\tau)  \nonumber \\
&=&  \sum_{1\leq k\leq \ell\leq q} \sum_{(i,j)\in S_{k,\ell}} g_{1,i,j} (\wh \theta^\tau_{1, k, \ell}, \wh \eta^\tau_{1, k, \ell}; \tau)  +\sum_{1\leq k\leq \ell\leq q} \sum_{(i,j)\in S_{k,\ell}}  g_{2,i,j} (\wh \theta^\tau_{2, k, \ell}, \wh \eta^\tau_{2, k, \ell}; \tau)\nonumber\\
&&-E \sum_{1\leq k\leq \ell\leq q} \sum_{(i,j)\in S_{k,\ell}} g_{1,i,j} (  \theta_{1, k, \ell},   \eta_{1, k, \ell}; \tau)  
-E\sum_{1\leq k\leq \ell\leq q} \sum_{(i,j)\in S_{k,\ell}}  g_{2,i,j} (  \theta^\tau_{2, k, \ell},   \eta^\tau_{2, k, \ell}; \tau)  \nonumber\\
&=&\sum_{1\leq k\leq \ell\leq q} \sum_{(i,j)\in S_{k,\ell}} g_{1,i,j} (\wh \theta^\tau_{1, k, \ell}, \wh \eta^\tau_{1, k, \ell}; \tau)  +\sum_{1\leq k\leq \ell\leq q} \sum_{(i,j)\in S_{k,\ell}}  g_{2,i,j} (\wh \theta^\tau_{2, k, \ell}, \wh \eta^\tau_{2, k, \ell}; \tau)\nonumber\\
&&-  \sum_{1\leq k\leq \ell\leq q} \sum_{(i,j)\in S_{k,\ell}} g_{1,i,j} (  \theta_{1, k, \ell},   \eta_{1, k, \ell}; \tau)  
- \sum_{1\leq k\leq \ell\leq q} \sum_{(i,j)\in S_{k,\ell}}  g_{2,i,j} (  \theta^\tau_{2, k, \ell},   \eta^\tau_{2, k, \ell}; \tau) \nonumber\\
&&+ \sum_{1\leq k\leq \ell\leq q} \sum_{(i,j)\in S_{k,\ell}} g_{1,i,j} (  \theta_{1, k, \ell},   \eta_{1, k, \ell}; \tau)  
+\sum_{1\leq k\leq \ell\leq q} \sum_{(i,j)\in S_{k,\ell}}  g_{2,i,j} (  \theta^\tau_{2, k, \ell},   \eta^\tau_{2, k, \ell}; \tau)    \nonumber\\
&&	-E \sum_{1\leq k\leq \ell\leq q} \sum_{(i,j)\in S_{k,\ell}} g_{1,i,j} (  \theta_{1, k, \ell},   \eta_{1, k, \ell}; \tau)  
-E\sum_{1\leq k\leq \ell\leq q} \sum_{(i,j)\in S_{k,\ell}}  g_{2,i,j} (  \theta^\tau_{2, k, \ell},   \eta^\tau_{2, k, \ell}; \tau)  \nonumber
\end{eqnarray}
Note that $\{\wh \theta^\tau_{1, k,\ell}, \wh \eta^\tau_{1, k,\ell}\}$ is the maximizer of $ \sum_{1\leq k\leq \ell\leq q} \sum_{(i,j)\in S_{k,\ell}} g_{1,i,j} (\theta_{ k, \ell},   \eta_{ k, \ell}; \tau) $. Applying Taylor's expansion we have, there exist random scalars $\theta^-_{k,\ell} \in [\wh\theta^\tau_{1,k,\ell}, \theta_{1,k,\ell}], \eta^-_{k,\ell}\in [\wh\eta^\tau_{1,k,\ell} , \eta_{1,k,\ell}]$ such that 
\begin{eqnarray*} 
	&&
	\sum_{1\leq k\leq \ell\leq q} \sum_{(i,j)\in S_{k,\ell}} g_{1,i,j} (\wh \theta^\tau_{1, k, \ell}, \wh \eta^\tau_{1, k, \ell}; \tau) 
	-  \sum_{1\leq k\leq \ell\leq q} \sum_{(i,j)\in S_{k,\ell}} g_{1,i,j} (  \theta_{1, k, \ell},   \eta_{1, k, \ell}; \tau)  \\
	&\leq& \frac{1}{2} \sum_{1 \le k \le \ell \le q}
	n_{k,\ell} \tau  \Bigg\{  
	\Big(\frac{\theta_{1,k,\ell}-\wh\theta^\tau_{1,k,\ell}}{\theta_{k,\ell}^-}\Big)^2+
	\Big(\frac{\theta_{1,k,\ell}-\wh\theta^\tau_{1,k,\ell}}{1-\theta_{k,\ell}^-}\Big)^2+  
	\Big(\frac{\eta_{1,k,\ell}-\wh\eta^\tau_{1,k,\ell}}{\eta_{k,\ell}^-}\Big)^2
	+  
	\Big(\frac{\eta_{1,k,\ell}-\wh\eta^\tau_{1,k,\ell}}{1-\eta_{k,\ell}^-}\Big)^2
	\Bigg\} .
\end{eqnarray*}
On the other hand, when $\wh \nu=\nu$, similar to Proposition \ref{uniformCon}  and Theorem \ref{uniformCon2}, we can show that for any $B>0$, there exists a large enough constant $C^{-}$ such that $\max_{1\leq k\leq \ell\leq q, \tau\in[\tau_{n,p},\tau_0]} |\wh \theta^\tau_{1, k, \ell} -\theta_{1, k, \ell}|\leq C^-\sqrt{\frac{\log (np )}{ns_{\min}^2}}$, and 
$\max_{1\leq k\leq \ell\leq q, \tau\in[\tau_{n,p},\tau_0]} |\wh \eta^\tau_{1, k, \ell} -\eta_{1, k, \ell}|=C^-\sqrt{\frac{\log (np)}{ns_{\min}^2}}$ hold with probability greater than $1-O((np)^{-B})$. Consequently, we have, when $\wh\nu=\nu$, there exits a large enough constant $C_4>0$ such that
\begin{eqnarray}\label{Mn1} 
&&
\sum_{1\leq k\leq \ell\leq q} \sum_{(i,j)\in S_{k,\ell}} g_{1,i,j} (\wh \theta^\tau_{1, k, \ell}, \wh \eta^\tau_{1, k, \ell}; \tau) 
-  \sum_{1\leq k\leq \ell\leq q} \sum_{(i,j)\in S_{k,\ell}} g_{1,i,j} (  \theta_{1, k, \ell},   \eta_{1, k, \ell}; \tau) \nonumber \\
&&\leq C_4 \tau\sum_{1 \le k \le \ell \le q}
n_{k,\ell} \frac{\log (np )}{ns_{\min}^2}  
\nonumber\\
&&\leq \frac{C_4\tau p^2  \log (np )}{ns_{\min}^2}.
\end{eqnarray}
Similarly, we have there exists a large enough constant $C_5>0$ such that with probability greater than $1-O((np)^{-B})$, 
\begin{eqnarray}\label{Mn2} 
&&
\sum_{1\leq k\leq \ell\leq q} \sum_{(i,j)\in S_{k,\ell}} g_{2,i,j} (\wh \theta^\tau_{2, k, \ell}, \wh \eta^\tau_{2, k, \ell}; \tau) 
-  \sum_{1\leq k\leq \ell\leq q} \sum_{(i,j)\in S_{k,\ell}} g_{2,i,j} (  \theta^\tau_{2, k, \ell},   \eta^\tau_{2, k, \ell}; \tau)  
\nonumber\\
&&\leq  \frac{C_5(n-\tau) p^2  \log (np )}{ns_{\min}^2}.
\end{eqnarray}
On the other hand, similar to Lemma \ref{concentration}, there exists a constant $C_6>0$ such that   with probability greater than $1-O((np)^{-B})$, 
\begin{eqnarray} \label{Mn3}
\\ 
&&
\left| \sum_{1\leq k\leq \ell\leq q} \sum_{(i,j)\in S_{k,\ell}} g_{1,i,j} (  \theta_{1, k, \ell},   \eta_{1, k, \ell}; \tau) 
-E \sum_{1\leq k\leq \ell\leq q} \sum_{(i,j)\in S_{k,\ell}} g_{1,i,j} (  \theta_{1, k, \ell},   \eta_{1, k, \ell}; \tau)   \right|  \nonumber\\
&\leq & C_6\tau p^2   \sqrt{\frac{  \log (np )}{\tau p^2 } }, \nonumber
\end{eqnarray}
and 
\begin{eqnarray} \label{Mn4}
\\ 
&&
\left| \sum_{1\leq k\leq \ell\leq q} \sum_{(i,j)\in S_{k,\ell}} g_{2,i,j} (  \theta^\tau_{2, k, \ell},   \eta^\tau_{2, k, \ell}; \tau) 
-E \sum_{1\leq k\leq \ell\leq q} \sum_{(i,j)\in S_{k,\ell}} g_{2,i,j} (  \theta^\tau_{2, k, \ell},   \eta^\tau_{2, k, \ell}; \tau)   \right|  \nonumber\\
&\leq & C_6(n-\tau)p^2   \sqrt{\frac{  \log (np)}{(n-\tau)p^2 } } .\nonumber
\end{eqnarray}
Combining \eqref{Mn}, \eqref{Mn1}, \eqref{Mn2}, \eqref{Mn3} and \eqref{Mn4} we conclude that when $\wh\nu=\nu$, there exists a large enough constant $C_0>0$ such that with probability greater than $1-O((np)^{-B})$, 
\begin{eqnarray}\label{Mn5} 
\\
\sup_{\tau\in [\tau_{n,p}, \tau_0]}|\mM_n(\tau)-\mM(\tau) |\leq C_0np^2 \left\{    {\frac{  \log (np )}{ns_{\min}^2 } } +\sqrt{\frac{  \log (np)}{np^2 } } \right\} = O\Bigg(np^2\sqrt{\frac{\log(np)}{ns_{\min}^2}}\Bigg). \nonumber
\end{eqnarray}

\noindent
{\bf (iv) Evaluating $E\sup_{\tau\in [\tau_{n,p},\tau_0]} |\mM_n(\tau)-\mM(\tau) -\mM_n(\tau_0)+\mM(\tau_0)| $.}  

\noindent
Notice that when $\wh\nu=\nu$,
\begin{eqnarray*} 
	&&	\mM_n(\tau)-\mM(\tau) -\mM_n(\tau_0)+\mM(\tau_0) \nonumber\\
	&=&  \sum_{1\leq k\leq \ell\leq q} \sum_{(i,j)\in S_{k,\ell}} g_{1,i,j} (\wh \theta^\tau_{1, k, \ell}, \wh \eta^\tau_{1, k, \ell}; \tau)  +\sum_{1\leq k\leq \ell\leq q} \sum_{(i,j)\in S_{k,\ell}}  g_{2,i,j} (\wh \theta^\tau_{2, k, \ell}, \wh \eta^\tau_{1, k, \ell}; \tau)\nonumber\\
	&&-E \sum_{1\leq k\leq \ell\leq q} \sum_{(i,j)\in S_{k,\ell}} g_{1,i,j} (  \theta_{1, k, \ell},   \eta_{1, k, \ell}; \tau)  
	-E\sum_{1\leq k\leq \ell\leq q} \sum_{(i,j)\in S_{k,\ell}}  g_{2,i,j} (  \theta^\tau_{2, k, \ell},   \eta^\tau_{2, k, \ell}; \tau)   \\
	&& -   \sum_{1\leq k\leq \ell\leq q} \sum_{(i,j)\in S_{k,\ell}} g_{1,i,j} (\wh \theta^{\tau_0}_{1, k, \ell}, \wh \eta^{\tau_0}_{1, k, \ell}; \tau_0)  -\sum_{1\leq k\leq \ell\leq q} \sum_{(i,j)\in S_{k,\ell}}  g_{2,i,j} (\wh \theta^{\tau_0}_{2, k, \ell}, \wh \eta^{\tau_0}_{2, k, \ell}; \tau_0)\nonumber\\
	&&+E \sum_{1\leq k\leq \ell\leq q} \sum_{(i,j)\in S_{k,\ell}} g_{1,i,j} (  \theta_{1, k, \ell},   \eta_{1, k, \ell}; \tau_0)  
	+E\sum_{1\leq k\leq \ell\leq q} \sum_{(i,j)\in S_{k,\ell}}  g_{2,i,j} (  \theta_{2, k, \ell},   \eta_{2, k, \ell}; \tau_0)
\end{eqnarray*}
Note that
\begin{eqnarray*}
	&&g_{1,i,j} (\wh \theta^\tau_{1, k, \ell}, \wh \eta^\tau_{1, k, \ell}; \tau) - g_{1,i,j} (\wh \theta^{\tau_0}_{1, k, \ell}, \wh \eta^{\tau_0}_{1, k, \ell}; \tau_0) -E[g_{1,i,j} (  \theta_{1, k, \ell},   \eta_{1, k, \ell}; \tau)  -g_{1,i,j} (  \theta_{1, k, \ell},   \eta_{1, k, \ell}; \tau_0)  ] \\
	&=&  \sum_{t=1}^\tau
	\Bigg\{ X_{i,j}^t(1-X_{i,j}^{t-1})  \log\frac{  \wh \theta^\tau_{1, k, \ell} }{ \wh \theta^{\tau_0}_{1, k, \ell}}  
	+  (1-X_{i,j}^t)(1-X_{i,j}^{t-1})  \log\frac {1 -\wh \theta^\tau_{1, k, \ell}}{1-\wh \theta^{\tau_0}_{1, k, \ell}}\\
	&&+  (1-X_{i,j}^t) X_{i,j}^{t-1}  \log\frac{  \wh \eta^\tau_{1, k, \ell} }{ \wh \eta^{\tau_0}_{1, k, \ell}}  
	+ X_{i,j}^t X_{i,j}^{t-1}\log\frac {1 -\wh \eta^{\tau}_{1, k, \ell}}{1-\wh \eta^{\tau_0}_{1, k, \ell}  }
	\Bigg\}- 
	\sum_{t=\tau+1}^{\tau_0}
	\Bigg\{ X_{i,j}^t(1-X_{i,j}^{t-1})  \log  \wh \theta^{\tau_0}_{1, k, \ell} \\
	&&+  (1-X_{i,j}^t)(1-X_{i,j}^{t-1})  \log  (1 -\wh \theta^{\tau_0}_{1, k, \ell}) 
	+  (1-X_{i,j}^t) X_{i,j}^{t-1}  \log  \wh \eta^{\tau_0}_{1, k, \ell}
	+ X_{i,j}^t X_{i,j}^{t-1}\log (1-\wh \eta^{\tau_0}_{1, k, \ell}  )  \Bigg\} \\
	&&+ E	\sum_{t=\tau+1}^{\tau_0}
	\Big\{ X_{i,j}^t(1-X_{i,j}^{t-1})  \log  \theta_{1,k,\ell}  
	+  (1-X_{i,j}^t)(1-X_{i,j}^{t-1})  \log (1 -\theta_{1,k,\ell}  ) \\
	&&	+  (1-X_{i,j}^t) X_{i,j}^{t-1} \log \eta_{1,k,\ell}   
	+ X_{i,j}^t X_{i,j}^{t-1} \log(1- \eta_{1,k,\ell}   ) \Big\}. 
\end{eqnarray*}
When sum over all   $(i,j)\in S_{k,\ell}$ and $1\leq k\leq\ell\leq q$, the last two terms in the above inequality can be bounded similar to \eqref{Mn1} and \eqref{Mn3}, with $\tau$ replaced by $\tau_0-\tau$. For the first term, with some calculations we have there exists a constant $c_1>0$ such that with probability larger than $1-O(np)^{-B})$, 
\begin{align}\label{appendA}
&\sup_{1\leq k\leq \ell\leq q}\left| \wh \theta^\tau_{1, k, \ell}- \wh \theta^{\tau_0}_{1, k, \ell} \right|
\leq c_1 \sqrt{\frac{  \tau_0-\tau }{\tau_0} }\sqrt{\frac{\log (np)}{ns_{\min}^2}}, \\
&
\sup_{1\leq k\leq \ell\leq q}	\left| \wh \eta^\tau_{1, k, \ell}- \wh \eta^{\tau_0}_{1, k, \ell} \right|
\leq c_1 \sqrt{\frac{ \tau_0-\tau}{\tau_0} }\sqrt{\frac{\log (np)}{ns_{\min}^2}}.\nonumber
\end{align}
Brief derivations of \eqref{appendA}  are provided in Section \ref{append}.  Consequently, similar to \eqref{Mn5}, we have there exists a large enough constant $c_2>0$ such that
\begin{eqnarray}\label{Mn_M1}
&&\Bigg|	 \sum_{1\leq k\leq \ell\leq q}\sum_{(i,j)\in S_{k,\ell}}\Big[g_{1,i,j} (\wh \theta^\tau_{1, k, \ell}, \wh \eta^\tau_{1, k, \ell}; \tau)
- g_{1,i,j} (\wh \theta^{\tau_0}_{1, k, \ell}, \wh \eta^{\tau_0}_{1, k, \ell}; \tau_0)\Big]  \nonumber\\
&&-E\sum_{1\leq k\leq \ell\leq q}\sum_{(i,j)\in S_{k,\ell}}\Big[g_{1,i,j} (  \theta_{1, k, \ell},   \eta_{1, k, \ell}; \tau)  -g_{1,i,j} (  \theta_{1, k, \ell},   \eta_{1, k, \ell}; \tau_0)  \Big]   \Bigg|   \nonumber\\
&\leq &  c_2    p^2 \sqrt{\frac{(\tau_0-\tau)\log (np)}{s_{\min}^2}} .
\end{eqnarray}
Here in the last step we have used the fact that $\tau_0\simeq O(n)$, $\sqrt{\frac{\log (np)}{p^2}}\leq \sqrt{\frac{\log (np)}{s_{\min}^2}} $, and $\frac{(\tau_0-\tau)\log (np)}{ns_{\min}^2}=o\Big( \sqrt{\frac{(\tau_0-\tau)\log (np)}{s_{\min}^2}}  \Big)$.
Similarly, note that,
\begin{eqnarray}\label{g2Decom}
\\
&&g_{2,i,j} (\wh \theta^\tau_{2, k, \ell}, \wh \eta^\tau_{2, k, \ell}; \tau) - g_{2,i,j} (\wh \theta^{\tau_0}_{2, k, \ell}, \wh \eta^{\tau_0}_{2, k, \ell}; \tau_0) -E[g_{2,i,j} (  \theta^\tau_{2, k, \ell},   \eta^\tau_{2, k, \ell}; \tau)  -g_{2,i,j} (  \theta_{2, k, \ell},   \eta_{2, k, \ell}; \tau_0)  ] 
\nonumber	\\
&=&  \sum_{t=\tau_0+1}^n
\Bigg\{ X_{i,j}^t(1-X_{i,j}^{t-1})  \Bigg[ \log\frac{  \wh \theta^\tau_{2, k, \ell} }{ \wh \theta^{\tau_0}_{2, k, \ell}}  
-\log\frac{    \theta^\tau_{2, k, \ell} }{   \theta_{2, k, \ell}} \Bigg]	+  (1-X_{i,j}^t)(1-X_{i,j}^{t-1}) \cdot 
\Bigg[ \log\frac{  1-\wh \theta^\tau_{2, k, \ell} }{1- \wh \theta^{\tau_0}_{2, k, \ell}}  
\nonumber	\\
&&	-\log\frac{ 1-   \theta^\tau_{2, k, \ell} }{ 1-  \theta_{2, k, \ell}} \Bigg]
+   X_{i,j}^t(1-X_{i,j}^{t-1})  \Bigg[ \log\frac{  \wh \eta^\tau_{2, k, \ell} }{ \wh \eta^{\tau_0}_{2, k, \ell}}  
-\log\frac{    \eta^\tau_{2, k, \ell} }{   \eta_{2, k, \ell}} \Bigg]
+ X_{i,j}^t X_{i,j}^{t-1}\Bigg[ \log\frac{  1-\wh \eta^\tau_{2, k, \ell} }{1- \wh \eta^{\tau_0}_{2, k, \ell}}  
\nonumber	\\
&& 
-\log\frac{ 1-   \eta^\tau_{2, k, \ell} }{ 1-  \eta_{2, k, \ell}} \Bigg]
\Bigg\} +[g_{2,i,j} (  \theta^\tau_{2, k, \ell},   \eta^\tau_{2, k, \ell}; \tau_0)  -g_{2,i,j} (  \theta_{2, k, \ell},   \eta_{2, k, \ell}; \tau_0)  ] 
\nonumber	\\
&&-E[g_{2,i,j} (  \theta^\tau_{2, k, \ell},   \eta^\tau_{2, k, \ell}; \tau_0)  -g_{2,i,j} (  \theta_{2, k, \ell},   \eta_{2, k, \ell}; \tau_0)  ] 
+ \sum_{t=\tau+1}^{\tau_0}\Bigg\{ X_{i,j}^t(1-X_{i,j}^{t-1})  \log  \wh \theta^\tau_{2, k, \ell} \nonumber\\
&&+(1-X_{i,j}^t)(1-X_{i,j}^{t-1})  \log  (1 -\wh \theta^\tau_{2, k, \ell}) 
+  (1-X_{i,j}^t) X_{i,j}^{t-1}  \log  \wh \eta^\tau_{2, k, \ell}
+ X_{i,j}^t X_{i,j}^{t-1}\log (1-\wh \eta^\tau_{2, k, \ell}  )  \Bigg\} 
\nonumber\\
&&- E	\sum_{t=\tau+1}^{\tau_0}
\Big\{ X_{i,j}^t(1-X_{i,j}^{t-1})  \log  \theta^\tau_{2,k,\ell}  
+  (1-X_{i,j}^t)(1-X_{i,j}^{t-1})  \log (1 -\theta^\tau_{2,k,\ell}  ) \nonumber\\
&&	+  (1-X_{i,j}^t) X_{i,j}^{t-1} \log \eta^\tau_{2,k,\ell}   
+ X_{i,j}^t X_{i,j}^{t-1} \log(1- \eta^\tau_{2,k,\ell}   ) \Big\} \nonumber\\
&:=&I +II - III +IV -V. \nonumber
\end{eqnarray}
For $II-III$, from Lemma \ref{concentration} and the fact that    $\left|   \theta^\tau_{2, k, \ell}-   \theta_{2, k, \ell} \right|
\leq	 \frac{c_3(\tau_0-\tau)}{n-\tau}  , 
\quad
{\rm and}
\quad
\left|   \eta^\tau_{2, k, \ell}-   \eta_{2, k, \ell} \right|
\leq \frac{c_3(\tau_0-\tau)}{n-\tau}$  for some large enough constant $c_3$, we have there exists a large enough constant $c_4>0$ such that with probability greater than
$1-O((np)^{-B})$,
\begin{eqnarray}\label{II-III}
\\
\left| \sum_{1\leq k\leq \ell\leq q}\sum_{(i,j)\in S_{k,\ell}} (II-III) \right|
\leq c_4 p^2 \frac{\tau_0-\tau}{n-\tau} \sqrt{\frac{\log(np)}{\tau_0p^2}}=
o\Bigg( p^2\sqrt{ \frac{(\tau_0-\tau)\log (np)}{s_{\min}^2}} \Bigg).  \nonumber
\end{eqnarray}
When sum over all   $(i,j)\in S_{k,\ell}$ and $1\leq k\leq\ell\leq q$,  the  $IV-V$ term can be bounded similar to \eqref{Mn1} and \eqref{Mn3}, with $\tau$ replaced by $\tau_0-\tau$, i.e., there exist a constant $c_5>0$ such that with probability greater than $1-O((np)^{-B})$,
\begin{eqnarray}\label{IV-V}
\left| \sum_{1\leq k\leq \ell\leq q}\sum_{(i,j)\in S_{k,\ell}} (IV-V) \right|
&\leq& c_5 p^2\Bigg[\frac{ (\tau_0-\tau)  \log (np )}{ns_{\min}^2}+ \sqrt{\tau_0-\tau }     \sqrt{\frac{  \log (np )}{ p^2 } } \Bigg] \nonumber\\
&=&O\Bigg( p^2\sqrt{ \frac{(\tau_0-\tau)\log (np)}{s_{\min}^2}} \Bigg). 
\end{eqnarray}
Lastly, similar to \eqref{appendA}, we can show that   there exists a constant $c_6>0$ such that  with probability larger than $1-O(np)^{-B})$, 
\begin{eqnarray}\label{appendA2} 
\sup_{1\leq k\leq \ell\leq q}\left| \log \frac{\wh \theta^\tau_{2, k, \ell}}{\theta^\tau_{2, k, \ell}}- \log \frac{ \wh \theta^{\tau_0}_{2, k, \ell} }{\theta_{2, k, \ell} }\right|
\leq	c_6 \sqrt{\frac{ \tau_0-\tau}{n} }\sqrt{\frac{\log (np)}{ns_{\min}^2}},
\\
\sup_{1\leq k\leq \ell\leq q}\left| \log \frac{\wh \eta^{\tau}_{2, k, \ell}}{\eta^\tau_{2, k, \ell}}- \log \frac{ \wh \eta^{\tau_0}_{2, k, \ell} }{\eta_{2, k, \ell} }\right|
\leq c_6 \sqrt{\frac{ \tau_0-\tau}{n} }\sqrt{\frac{\log (np)}{ns_{\min}^2}}. \nonumber
\end{eqnarray}
A brief proof of \eqref{appendA2} is provided in Section \ref{append}. 
Consequently,  we can show that there exists a constant $c_7>0$ such that  with probability larger than $1-O(np)^{-B})$, 
\begin{eqnarray}\label{Mn_M2}
&&\Bigg|	 \sum_{1\leq k\leq \ell\leq q}\sum_{(i,j)\in S_{k,\ell}}\Big[g_{2,i,j} (\wh \theta^\tau_{2, k, \ell}, \wh \eta^\tau_{2, k, \ell}; \tau)
- g_{2,i,j} (\wh \theta^{\tau_0}_{2, k, \ell}, \wh \eta^{\tau_0}_{2, k, \ell}; \tau_0)\Big]  \nonumber\\ 
&&-E\sum_{1\leq k\leq \ell\leq q}\sum_{(i,j)\in S_{k,\ell}}\Big[g_{2,i,j} (  \theta^\tau_{2, k, \ell},   \eta^\tau_{2, k, \ell}; \tau)  -g_{2,i,j} (  \theta_{2, k, \ell},   \eta_{2, k, \ell}; \tau_0)  \Big]   \Bigg|  \nonumber \\
&\leq &  c_7 p^2    \sqrt{ \frac{(\tau_0-\tau)\log (np)}{s_{\min}^2}}  .
\end{eqnarray}

Now combining \eqref{Mn_M1} and \eqref{Mn_M2} and  the  probability for $\wh\nu\neq\nu$  in (ii), we conclude that there exists a constant $C_0>0$ such that 
\begin{eqnarray}\label{EMn_M} 
&& E \sup_{\tau\in [\tau_{n,p}, \tau_0]} |\mM_n(\tau)-\mM(\tau) -\mM_n(\tau_0)+\mM(\tau_0)| \nonumber \\
&\leq& C_0np^2 \left\{   \sqrt{ \frac{(\tau_0-\tau_{n,p})\log (np)}{n^2 s_{\min}^2}}  +o\left( \sqrt{ \frac{(\tau_0-\tau_{n,p})\log (np)}{ n^2s_{\min}^2}} \right)\right\} \nonumber \\
&\leq& 2 C_0p^2 \sqrt{ \frac{(\tau_0-\tau_{n,p})\log (np)}{s_{\min}^2}} . 
\end{eqnarray}

\noindent
{\bf (v)  Evaluating  $\sup_{\tau\in[n_0, \tau_{n,p}] }[\mM_n(\tau)-\mM_n(\tau_0)]$. }  

\noindent
In this part we consider  the case when $\tau\in[n-n_0, \tau_{n,p}]$. 
We shall see    that $\sup_{\tau\in[n_0, \tau_{n,p}] }[\mM_n(\tau)-\mM_n(\tau_0)]<0$    in probability   and hence $\argmax_{\tau\in [n_0, \tau_0]} \mM_n(\tau)= \argmax_{\tau\in [\tau_{n,p}, \tau_0]} \mM_n(\tau)$ holds in probability. 
Note that for any $\tau\in[n-n_0, \tau_{n,p}]$, 
\begin{eqnarray}\label{MnTau0}
\mM_n(\tau)- \mM_n(\tau_0) =\mM_n(\tau) - \mM(\tau) - \mM_n(\tau_0)+ \mM(\tau_0) - [\mM(\tau_0) -\mM(\tau)].
\end{eqnarray}

Given $\wh \nu^{1,\tau}$ and $\wh \nu^{\tau+1,n}$, we define an intermediate term
\begin{align*} 
\mM_n^*(\tau):=      
l( \{  \theta^-_{\tau, k, \ell},  \eta^-_{\tau, k, \ell}\};\;   \wh \nu^{1,\tau})+ l( \{  \theta^*_{\tau, k, \ell},  \eta^*_{\tau, k, \ell}\};\;   \wh \nu^{\tau+1,n}). 
\end{align*}
where
\begin{align*}   
\theta^-_{\tau, k, \ell}={
	\frac{\sum_{(i,j)\in \wh S^\tau_{1, k,\ell} }    \frac{ \theta_{1,\nu(i),\nu(j)}\eta_{1,\nu(i),\nu(j)}}{\theta_{1,\nu(i),\nu(j)}+\eta_{1,\nu(i),\nu(j)}} 
	}
	{	\sum_{(i,j)\in \wh S^\tau_{1, k,\ell} }    \frac{ \eta_{1,\nu(i),\nu(j)}}{\theta_{1,\nu(i),\nu(j)}+\eta_{1,\nu(i),\nu(j)}}   
	}  
},\quad
\eta^-_{\tau, k, \ell} =
\frac{\sum_{(i,j)\in \wh S^\tau_{1, k,\ell} }   \frac{ \theta_{1,\nu(i),\nu(j)}\eta_{1,\nu(i),\nu(j)}}{\theta_{1,\nu(i),\nu(j)}+\eta_{1,\nu(i),\nu(j)}} 
}
{	\sum_{(i,j)\in \wh S^\tau_{1, k,\ell} }    \frac{  \theta_{1,\nu(i),\nu(j)}}{\theta_{1,\nu(i),\nu(j)}+\eta_{1,\nu(i),\nu(j)}}  
}  ,
\end{align*} 
and
\begin{align*}   
\theta^*_{\tau, k, \ell}& ={
	\frac{\sum_{(i,j)\in \wh S^\tau_{2, k,\ell} }  \left[  \frac{(\tau_0-\tau)\theta_{1,\nu(i),\nu(j)}\eta_{1,\nu(i),\nu(j)}}{\theta_{1,\nu(i),\nu(j)}+\eta_{1,\nu(i),\nu(j)}} +\frac{(n-\tau_0)\theta_{2,\nu(i),\nu(j)}\eta_{2,\nu(i),\nu(j)}}{\theta_{2,\nu(i),\nu(j)}+\eta_{2,\nu(i),\nu(j)}}   \right]
	}
	{	\sum_{(i,j)\in \wh S^\tau_{2, k,\ell} }  \left[  \frac{(\tau_0-\tau) \eta_{1,\nu(i),\nu(j)}}{\theta_{1,\nu(i),\nu(j)}+\eta_{1,\nu(i),\nu(j)}} +\frac{(n-\tau_0) \eta_{2,\nu(i),\nu(j)}}{\theta_{2,\nu(i),\nu(j)}+\eta_{2,\nu(i),\nu(j)}}   \right]
	}  
},\\[1ex]  
\eta^*_{\tau, k, \ell}& =
\frac{\sum_{(i,j)\in \wh S^\tau_{2, k,\ell} }  \left[  \frac{(\tau_0-\tau)\theta_{1,\nu(i),\nu(j)}\eta_{1,\nu(i),\nu(j)}}{\theta_{1,\nu(i),\nu(j)}+\eta_{1,\nu(i),\nu(j)}} +\frac{(n-\tau_0)\theta_{2,\nu(i),\nu(j)}\eta_{2,\nu(i),\nu(j)}}{\theta_{2,\nu(i),\nu(j)}+\eta_{2,\nu(i),\nu(j)}}   \right]
}
{	\sum_{(i,j)\in \wh S^\tau_{2, k,\ell}}  \left[  \frac{(\tau_0-\tau) \theta_{1,\nu(i),\nu(j)}}{\theta_{1,\nu(i),\nu(j)}+\eta_{1,\nu(i),\nu(j)}} +\frac{(n-\tau_0) \theta_{2,\nu(i),\nu(j)}}{\theta_{2,\nu(i),\nu(j)}+\eta_{2,\nu(i),\nu(j)}}   \right]
}  .
\end{align*} 
We have
\begin{eqnarray*} 
	\mM_n(\tau) - \mM(\tau) = \mM_n(\tau) 
	-E \mM^*_n(\tau)+E\mM^*_n(\tau)- \mM(\tau)  .
\end{eqnarray*}
Note that the expected log-likelihood $E\sum_{1\leq i\leq j\leq p}g_{1,i,j}(\alpha_{1,i,j}, \beta_{1,i,j},\tau)$ is maximized at  $\alpha_{1,i,j}=\theta_{1,\nu(i),\nu(j)}, \beta_{1,i,j}=\eta_{1,\nu(i),\nu(j)}$, and $E\sum_{1\leq i\leq j\leq p}g_{2,i,j}(\alpha_{2,i,j}, \beta_{2,i,j},\tau)$ is maximized at  $\alpha_{2,i,j}=\theta_{\tau,\nu(i),\nu(j)}$, $\beta_{2,i,j}=\eta_{\tau,\nu(i),\nu(j)}$, we have  
\begin{eqnarray*} 
	E\mM^*_n(\tau)- \mM(\tau)   \leq 0. 
\end{eqnarray*}
On the other hand, notice that given $\wh\nu$,  $\{  \theta^-_{\tau, k, \ell},  \eta^-_{\tau, k, \ell}\}$ is the maximizer of $El( \{  \theta_{ k, \ell},  \eta_{ k, \ell}\};\;   \wh \nu^{1,\tau}) $ and $ \{  \theta^*_{\tau, k, \ell},  \eta^*_{\tau, k, \ell}\}$ is the maximizer of  $El( \{  \theta_{ k, \ell},  \eta_{ k, \ell}\};\;   \wh \nu^{\tau+1,n})$. Similar to \eqref{Mn5}, there exists a large enough constant $C_7>0$ such that with probability greater than $1-O((np)^{-B})$, 
\begin{eqnarray*}  
	\sup_{\tau\in [n_0,\tau_{n,p}]}|\mM_n(\tau) -E \mM^*_n(\tau) | \leq 
	C_7 np^2 \left\{    {\frac{  \log (np )}{n  } } +\sqrt{\frac{  \log (np )}{np^2 } }   \right\}.
\end{eqnarray*}
Consequently we have, with probability greater than $1-O((np)^{-B})$, 
\begin{eqnarray}\label{MnTau1}
\sup_{\tau\in [n_0,\tau_{n,p}]} \big[\mM_n(\tau) - \mM(\tau) \big] \leq C_7 np^2 \left\{    {\frac{  \log (np )}{n  } } +\sqrt{\frac{  \log (np )}{np^2 } }   \right\}.
\end{eqnarray}
We remark that since the membership structure $\wh\nu^{\tau+1, n}$ can be very different from the original $\nu$, the $s_{\min}$ in \eqref{Mn5} is simply replaced by the lower bound $1$, and hence the upper bound in \eqref{MnTau1} is independent of $\wh\nu^{1;\tau}$ and $\wh\nu^{\tau+1, n}$.

Combining \eqref{MnTau0}, \eqref{MnTau1}, \eqref{Con1},   \eqref{Mn5} (with $\tau=\tau_0$), and choosing   $\kappa>0$ to be large enough,  we have 
with probability greater than $1-O((np)^{-B})$, 
\begin{eqnarray*} 
	&& \sup_{\tau\in [n_0,\tau_{n,p}]} \big[ \mM_n(\tau)- \mM_n(\tau_0) \big] \\ &\leq& 
	C_7 np^2 \left\{    {\frac{  \log (np )}{n  } } +\sqrt{\frac{  \log (np )}{np^2 } }   \right\}   
	+C_0np^2 \left\{    {\frac{  \log (np )}{ns_{\min}^2 } } +\sqrt{\frac{  \log (np )}{np^2 } } \right\}  \nonumber \\
	&& -C_3 (\tau_0-\tau_{n,p})\big[ \| \bW_{1,1}-\bW_{2,1}\|_F^2+\| \bW_{1,2}-\bW_{2,2}\|_F^2\big] \nonumber \\
	&<&0.
\end{eqnarray*}
Consequently we have, 
\begin{eqnarray}\label{MnTau}
\\
P\left(  \argmax_{\tau\in [n_0,\tau_0]} \big[ \mM_n(\tau)- \mM_n(\tau_0) \big] = \argmax_{\tau\in [\tau_{n,p},\tau_0]} \big[ \mM_n(\tau)- \mM_n(\tau_0) \big]  \right)\geq 1-O((np)^{-B}).\nonumber
\end{eqnarray}

\noindent
{\bf (vi)  Error bound for $\tau_0-\wh\tau$. }  

\noindent

One of the key steps in the proof of (v) is to compare $M_n(\tau)$, the estimated log-likelihood evaluated under the MLEs at a searching time point $\tau$, with $M(\tau)$, the maximized expected log-likelihood at time $\tau$. The error between $M_n(\tau)$ and $M(\tau)$, which is of order $O\Bigg(np^2\Big( {\frac{  \log (np )}{n  } } +\sqrt{\frac{  \log (np )}{np^2 } } \Big)\Bigg) $  reflects the noise level. On the other hand, the signal is captured by $M(\tau_0)-M(\tau)=O(|\tau_0-\tau|p^2\Delta_F^2)$,  i.e., the difference between the maximized expected log-likelihood evaluated at the true change point $\tau_0$ and the maximized expected log-likelihood evaluated at the searching time point $\tau$.  Consequently, when   $|\tau_0-\tau|p^2\Delta_F^2>\kappa \Bigg[np^2\Big( {\frac{  \log (np )}{n  } } +\sqrt{\frac{  \log (np )}{np^2 } } \Big)\Bigg] $ for some large enough constant $\kappa>0$, we are able to claim that $|\tau_0-\wh\tau|\leq |\tau_0-\tau|=O_p\Big( n \Delta_F^{-2} \Big[  {\frac{  \log (np )}{n  } } +\sqrt{\frac{  \log (np )}{np^2 } }   \Big] \Big)$. By further deriving the estimation errors for any $\tau$ in the  neighborhood of $\tau_0$ with radius $O\Big(  \Delta_F^{-2} \Big[  {\frac{  \log (np )}{n  } } +\sqrt{\frac{  \log (np )}{np^2 } }   \Big] \Big)$, we obtained a better bound based on Markov's inequality (see \eqref{vi_key} below).

From \eqref{MnTau} we have for any $0<\epsilon\leq \tau_0-\tau_{n,p}$, 
\begin{eqnarray*}
	P(\tau_0-\wh\tau> \epsilon) \leq P\bigg(  \sup_{\tau\in [\tau_{n,p},\tau_{0}-\epsilon  ]} \mM_n(\tau)- \mM_n(\tau_0) \geq 0 \bigg) + O((np)^{-B}).
\end{eqnarray*}
Note that from (i) and (iv) we have
\begin{eqnarray}\label{vi_key}
&&P\bigg(  \sup_{\tau\in [\tau_{n,p},\tau_{0}-\epsilon  ]} \mM_n(\tau)- \mM_n(\tau_0) \geq  0 \bigg)   \\
&\leq& P\bigg(  \sup_{\tau\in [\tau_{n,p},\tau_{0}-\epsilon  ]} \big[ ( \mM_n(\tau)- \mM(\tau)  -\mM_n(\tau_0) +\mM(\tau_0)  )  -(\mM(\tau_0)-\mM(\tau))\big] \geq 0
\bigg) \nonumber\\
&\leq& P\bigg(  \sup_{\tau\in [\tau_{n,p},\tau_{0}-\epsilon  ]}  \left| \mM_n(\tau)- \mM(\tau)  -\mM_n(\tau_0) +\mM(\tau_0)   \right|  
\geq  C_3\epsilon p^2\Delta_F^2
\bigg) \nonumber \\
&\leq& \frac{E  \sup_{\tau\in [\tau_{n,p},\tau_{0}-\epsilon  ]}  \big| \mM_n(\tau)- \mM(\tau)  -\mM_n(\tau_0) +\mM(\tau_0)     \big| }{C_3 \epsilon p^2\Delta_F^2} \nonumber\\
&\leq& \frac{ 2C_0p^2\sqrt{\frac{(\tau_0-\tau_{n,p})\log(np)}{s_{\min}^2}} }{C_3 \epsilon p^2\Delta_F^2} . \nonumber
\end{eqnarray}
We thus conclude that $\tau_0-\wh\tau = O_p\left(  \Delta_F^{-2} \sqrt{\frac{(\tau_0-\tau_{n,p})\log(np)}{s_{\min}^2}} \right)$. By the definition of $\tau_{n,p}$ and  
condition C5 we have, 
\begin{eqnarray*} 
	\Delta_F^{-2} \sqrt{\frac{(\tau_0-\tau_{n,p})\log(np)}{s_{\min}^2}}& =&O\Bigg( \frac{\tau_0-\tau_{n,p}}{\Delta_F}  \sqrt{\frac{ \log(np)} { ns_{\min}^2}  } \Bigg[\frac{\log (np)}{n}+\sqrt{\frac{\log (np)}{np^2}}\Bigg]^{-1/2}  \Bigg).
\end{eqnarray*}
Consequently, we conclude that
\[
\tau_0-\wh\tau =O_p\left(  (\tau_0-\tau_{n,p})  \min \Bigg\{1, \frac{\min\Big\{  1,  (n^{-1}p^2\log(np))^{\frac{1}{4}}\Big\} }{\Delta_F s_{\min}} \Bigg\}\right).
\]

\subsubsection{Change point estimation with $\nu^{1,\tau_0}\neq\nu^{\tau_0+1,n}$.}\label{unequal}
We  modify steps (i)-(v) to the case where $\nu^{1,\tau_0}\neq\nu^{\tau_0+1,n}$.  

With some abuse of notations,  
we put $\bW_{1,1}=(\alpha_{1,i,j})_{p\times p}$ with
$\alpha_{1,i,j}=\theta_{1,\nu^{1,\tau_0}(i),\nu^{1,\tau_0}(j)}$,
$\bW_{1,2}=(1-\beta_{1,i,j})_{p\times p}$ with
$\beta_{1,i,j}=\eta_{1,\nu^{1,\tau_0}(i),\nu^{1,\tau_0}(j)}$, 
$\bW_{2,1}=(\alpha_{2,i,j})_{p\times p}$ with
$\alpha_{2,i,j}=\theta_{2,\nu^{\tau_0+1,n}(i),\nu^{\tau_0+1,n}(j)}$, and
$\bW_{2,2}=(1-\beta_{2,i,j})_{p\times p}$ with
$\beta_{2,i,j}=\eta_{2,\nu^{\tau_0+1,n}(i),\nu^{\tau_0+1,n}(j)}$. 
Similar to previous proofs we define 
\begin{align*} 
&\mM_n(\tau):=  \sum_{1\leq i \leq j \leq p} g_{1,i,j}(\wh \alpha^\tau_{1,i,j},\wh\beta^\tau_{1,i,j},\tau)+ \sum_{1\leq i \leq j \leq p} g_{2,i,j}(\wh\alpha^\tau_{2,i,j},\wh\beta^\tau_{2,i,j},\tau) , \\ 
& \mM(\tau):= E  \sum_{1\leq i \leq j \leq p} g_{1,i,j}(\alpha_{1,i,j},\beta_{1,i,j},\tau)+E \sum_{1\leq i \leq j \leq p} g_{2,i,j}(\alpha^\tau_{2,i,j},\beta^\tau_{2,i,j},\tau) ,
\end{align*}
where 
\begin{eqnarray*}
	\alpha^\tau_{2,i,j}=\frac{\frac{\tau_0-\tau}{n-\tau}\frac{\alpha_{1,i,j}\beta_{1,i,j}}{\alpha_{1,i,j}+\beta_{1,i,j}}+\frac{n-\tau_0 }{n-\tau}\frac{\alpha_{2,i,j}\beta_{2,i,j}}{\alpha_{2,i,j}+\beta_{2,i,j}}}{
		\frac{\tau_0-\tau}{n-\tau}\frac{ \beta_{1,i,j}}{\alpha_{1,i,j}+\beta_{1,i,j}}+\frac{n-\tau_0 }{n-\tau}\frac{ \beta_{2,i,j}}{\alpha_{2,i,j}+\beta_{2,i,j}}},
	\\ \beta^\tau_{2,i,j}=\frac{\frac{\tau_0-\tau}{n-\tau}\frac{\alpha_{1,i,j}\beta_{1,i,j}}{\alpha_{1,i,j}+\beta_{1,i,j}}+\frac{n-\tau_0 }{n-\tau}\frac{\alpha_{2,i,j}\beta_{2,i,j}}{\alpha_{2,i,j}+\beta_{2,i,j}}}{
		\frac{\tau_0-\tau}{n-\tau}\frac{ \alpha_{1,i,j}}{\alpha_{1,i,j}+\beta_{1,i,j}}+\frac{n-\tau_0 }{n-\tau}\frac{ \alpha_{2,i,j}}{\alpha_{2,i,j}+\beta_{2,i,j}}}, 
\end{eqnarray*}
and
\begin{eqnarray*}
	&& \wh \alpha^\tau_{1,i,j}= \wh \theta^\tau_{1, \wh\nu^{1,\tau}(i),\wh\nu^{1,\tau}(j) } , \quad
	\wh \beta^\tau_{1,i,j}= \wh \eta^\tau_{1, \wh\nu^{1,\tau}(i),\wh\nu^{1,\tau}(j) }, \\
	&& \wh \alpha^\tau_{2,i,j}= \wh \theta^\tau_{2, \wh\nu^{\tau+1,n}(i),\wh\nu^{\tau+1,n}(j) } , \quad
	\wh \beta^\tau_{2,i,j}= \wh \eta^2_{\tau, \wh\nu^{\tau+1,n}(i),\wh\nu^{\tau+1,n}(j) }.
\end{eqnarray*}

Note that the definition of $M(\tau)$ here is now slightly different from the previous definition in that the $\alpha_{2,i,j}^\tau$ and $\beta_{2,i,j}^\tau$ will generally be different from $\theta_{2,\nu^{\tau_0+1,n}(i),\nu^{\tau_0+1,n}(j)}^\tau$ and $\eta_{2,\nu^{\tau_0+1,n}(i),\nu^{\tau_0+1,n}(j)}^\tau$, unless
$\nu^{1,\tau_0}=\nu^{\tau_0+1,n}$.
We first of all point out the main difference we are facing in the case where $\nu^{1,\tau_0}\neq\nu^{\tau_0+1,n}$.
Consider a detection time $\tau\in [\tau_{n,p},\tau_0]$. 
In the case where $\wh\nu^{1,\tau}=\wh\nu^{\tau+1,n}=\nu$, we have $ \alpha^\tau_{2,i,j} = \theta^\tau_{2, k,\ell}$ for all $(i,j)\in S_{k,\ell}$, and we have $|\wh\theta^\tau_{2, k,\ell}-\theta^\tau_{2, k, \ell}|=O_p\Big( \sqrt{\frac{\log(np)}{ns_{\min}^2} }\Big)$ for all $1\leq k\leq \ell \leq q$, or equivalently, $|\wh\alpha^\tau_{2, i,j}-\theta^\tau_{2, \nu(i), \nu(j)}|=O_p\Big( \sqrt{\frac{\log(np)}{ns_{\min}^2} }\Big)$ for all $1\leq i\leq j \leq p$. However, when $\wh\nu^{1,\tau}=\nu^{1,\tau_0}$ $\wh\nu^{\tau+1,n}= \nu^{\tau_0+1,n}$ but $ \nu^{1,\tau_0} \neq \nu^{\tau_0+1,n}$, the order of the estimation error becomes $O_p\Big( \sqrt{\frac{\log(np)}{ns_{\min}^2} } +\frac{\tau_0-\tau}{n}\Big)$. Here $\frac{\tau_0-\tau}{n}$ is a bias terms brought by the fact that $\wh\nu^{1,\tau} \neq \wh\nu^{\tau+1,n} $. The main issue is that the  the following terms from the definition of $\wh \theta^\tau_{2, k, \ell}$:
\[
{
	\sum_{(i,j)\in \wh S^\tau_{
			2, k,\ell} }
	\sum_{t=\tau+1}^{\tau_0} X_{i,j}^t(1 - X_{i,j}^{t-1}), \quad
	\sum_{(i,j)\in \wh S^\tau_{2, k,\ell} }
	\sum_{t=\tau+1}^{\tau_0} (1 - X_{i,j}^{t-1})},
\]
are no longer unbiased estimators (subject to a normalization) of the following corresponding terms in the definition of  $\theta^\tau_{2, k, \ell}$:
\[
\frac{\theta_{1,k,\ell}\eta_{1,k,\ell}}{\theta_{1,k,\ell}+\eta_{1,k,\ell}},\quad \frac{\theta_{1,k,\ell},  }{\theta_{1,k,\ell}+\eta_{1,k,\ell}}.
\]

The proof of (i) does not involve any parameter estimators and hence can be established similarly.

For (ii), note that $|\wh\alpha^\tau_{2, i, j}-\alpha_{2,i,j}|\leq |\wh\alpha^\tau_{2, i, j}-\alpha^\tau_{2,i,j}|+O\big(\frac{\tau_0-\tau}{n}\big)$  holds for all $1\leq i< j\leq p$, where the $O\big(\frac{\tau_0-\tau}{n}\big)$ is independent of $i,j$. This implies that  when estimating the $\alpha_{2,i,j}$, we have introduced a bias term $O\big(\frac{\tau_0-\tau}{n}\big)$  by including the $\tau_0-\tau$ samples before the change point. From the proofs of  Lemma \ref{Frob}, and condition C4, we conclude that (ii) hold for $\wh\nu^{\tau+1, n}$.

For (iii), replacing the order of the error bound for $\wh\theta^+_{\tau,k,\ell}$  and $\wh\theta^+_{\tau,k,\ell}$ from $\sqrt{\frac{\log(np)}{ns_{\min}^2} }$ to  $\sqrt{\frac{\log(np)}{ns_{\min}^2} }+\frac{\tau_0-\tau}{n}$, 
we have there exists a large enough constant $C_0>0$ such that
\begin{eqnarray*} 
	\sup_{\tau\in [\tau_{n,p}, \tau_0]}|\mM_n(\tau)-\mM(\tau) |&\leq &
	C_0np^2 \left\{    {\frac{  \log (np )}{ns_{\min}^2 } } +\sqrt{\frac{  \log (np)}{np^2 } } +\frac{(\tau_0-\tau_{n,p})^2}{n^2}\right\}\\
	&=&O\Bigg(np^2\left\{\sqrt{\frac{\log(np)}{ns_{\min}^2}}+\frac{(\tau_0-\tau_{n,p})^2}{n^2}\right\}\Bigg). 
\end{eqnarray*}

For (iv), the error bounds related to $g_{1,i,j}(\cdot, \cdot;\cdot)$ remain unchanged. 
Note that the decomposition \eqref{g2Decom} still holds  with $\theta_{2,k,\ell}^\tau, \eta_{2,k,\ell}^\tau$ replaced be $\alpha^\tau_{2,i,j}, \beta_{2,i,j}^\tau$ and $\wh\theta_{2,k,\ell}^\tau, \wh\eta_{2,k,\ell}^\tau$ replaced be $\wh\alpha^\tau_{2,i,j}, \wh\beta_{2,i,j}^\tau$. The bound for \eqref{II-III}  still holds owing to the fact that $|\alpha^\tau_{2,i,j}-\alpha_{2,i,j} |= O\Big(\frac{\tau_0-\tau}{n}\Big)$ and $|\beta^\tau_{2,i,j}-\beta_{2,i,j} |= O\Big(\frac{\tau_0-\tau}{n}\Big)$. The bound for \eqref{IV-V} would become $O\Bigg( p^2\sqrt{ \frac{(\tau_0-\tau)\log (np)}{s_{\min}^2}} +\frac{(\tau_0-\tau_{n,p})^2}{n^2} \Bigg)$. Notice that similar to \eqref{appendA2}, we have with probability larger than $1-O((np)^{-B})$,

\begin{eqnarray*} 
	\sup_{1\leq i\leq j\leq p}\left| \log \frac{\wh \alpha^\tau_{2, i, j}}{\alpha^\tau_{2, i, j}}- \log \frac{ \wh \alpha^{\tau_0}_{2, k, \ell} }{\alpha_{2, k, \ell} }\right|
	=  O\left( \sqrt{\frac{ \tau_0-\tau}{n} }\sqrt{\frac{\log (np)}{ns_{\min}^2}} + \frac{\tau_0-\tau}{n}\right),
	\\
	\sup_{1\leq i\leq j\leq p}\left| \log \frac{\wh \beta^\tau_{2, i, j}}{\beta^\tau_{2, i, j}}- \log \frac{ \wh \beta^{\tau_0}_{2, k, \ell} }{\beta{2, k, \ell} }\right|
	= O\left( \sqrt{\frac{ \tau_0-\tau}{n} }\sqrt{\frac{\log (np)}{ns_{\min}^2}} + \frac{\tau_0-\tau}{n} \right).
\end{eqnarray*}
Consequently,
we have 
\begin{eqnarray*} 
	E \sup_{\tau\in [\tau_{n,p}, \tau_0]} |\mM_n(\tau)-\mM(\tau) -\mM_n(\tau_0)+\mM(\tau_0)|  
	\leq  C_0 p^2 \left\{ \sqrt{ \frac{(\tau_0-\tau_{n,p}) \log (np)}{ s_{\min}^2}} + (\tau_0-\tau_{n,p})\right\}. 
\end{eqnarray*}

By noticing that $\{\alpha_{1,i,j} ,\beta_{1,i,j}, \alpha_{2,i,j}^{\tau}, \beta_{2,i,j}^{\tau}\}$ is the maximizer of $\mM(\tau)$, we conclude that (v) also holds.  
Consequently, for (vi), we have
\begin{eqnarray*}
	P\bigg(  \sup_{\tau\in [\tau_{n,p},\tau_{0}-\epsilon  ]} \mM_n(\tau)- \mM_n(\tau_0) \geq  0 \bigg)   \leq
	\frac{ C_0p^2\sqrt{\frac{(\tau_0-\tau_{n,p})\log(np)}{s_{\min}^2}} +C_0 p^2(\tau_0-\tau_{n,p})}{C_3 \epsilon p^2\Delta_F^2} . 
\end{eqnarray*}

Consequently,   we conclude that 
\[
\tau_0-\wh\tau =O_p\left(  (\tau_0-\tau_{n,p})  \min \Bigg\{1, \frac{\min\Big\{  1,  (n^{-1}p^2\log(np))^{\frac{1}{4}}\Big\} }{\Delta_F s_{\min}}   + \frac{1}{\Delta_F^2}\Bigg\}\right).
\]

\subsubsection{Proofs of \eqref{appendA} and \eqref{appendA2} when $\wh\nu=\nu$ }\label{append} 

For   \eqref{appendA}, note that 
\begin{eqnarray}\label{A_1}
\\
\left| \wh \theta^\tau_{1, k, \ell}- \wh \theta^{\tau_0}_{1, k, \ell} \right|
= \left| 
\frac{	\sum_{(i,j)\in   S_{
			k,\ell} }
	\sum_{t=1}^{\tau} X_{i,j}^t(1 - X_{i,j}^{t-1}) }{
	\sum_{(i,j)\in   S_{k,\ell} }
	\sum_{t=1}^\tau (1 - X_{i,j}^{t-1})} -
\frac{	\sum_{(i,j)\in   S_{
			k,\ell} }
	\sum_{t=1}^{\tau_0} X_{i,j}^t(1 - X_{i,j}^{t-1}) }{
	\sum_{(i,j)\in   S_{ k,\ell} }
	\sum_{t=1}^{\tau_0} (1 - X_{i,j}^{t-1})}   \right|  . \nonumber
\end{eqnarray}
Similar to Lemma \ref{concentration},  we can show that for any constant $B>0$,  there exists a large enough constant $B_1$ such that with probability larger than $1-O((np)^{-(B+2)})$,
\begin{eqnarray*}
	&&  \left|\frac{1}{\tau n_{k,\ell}}\sum_{(i,j)\in   S_{k,\ell} }
	\sum_{t=1}^\tau (1 - X_{i,j}^{t-1}) - \frac{\eta_{1,k,\ell}}{\theta_{1,k,\ell}+\eta_{1,k,\ell}}\right| \leq B_1\sqrt{\frac{\log (np)}{\tau n_{k,\ell}}}, \\
	&& \left|\frac{1}{\tau n_{k,\ell}}\sum_{(i,j)\in   S_{k,\ell} }
	\sum_{t=1}^\tau X_{i,j}^t(1 - X_{i,j}^{t-1}) - \frac{\eta_{1,k,\ell}}{\theta_{1,k,\ell}+\eta_{1,k,\ell}}\right| \leq B_1\sqrt{\frac{\log (np)}{\tau n_{k,\ell}}},
\end{eqnarray*}
and
\begin{eqnarray*}
	&& \frac{1}{\tau (\tau_0-\tau) n^2_{k,\ell}}  \Bigg| 
	\bigg[\sum_{(i,j)\in   S_{k,\ell} }	\sum_{t=1}^\tau X_{i,j}^t(1 - X_{i,j}^{t-1}) \bigg] \bigg[	\sum_{(i,j)\in   S_{k,\ell} }	\sum_{t=\tau+1}^{\tau_0}  (1 - X_{i,j}^{t-1}) \bigg] \\ 
	&&-\bigg[\sum_{(i,j)\in   S_{k,\ell} }	\sum_{t=\tau+1}^{\tau_0} X_{i,j}^t(1 - X_{i,j}^{t-1}) \bigg] \bigg[	\sum_{(i,j)\in   S_{k,\ell} }	\sum_{t= 1}^{\tau }  (1 - X_{i,j}^{t-1}) \bigg] 
	\Bigg| \leq B_1\sqrt{\frac{\log (np)}{(\tau_0-\tau) n_{k,\ell}}}.
\end{eqnarray*}
Plug these into \eqref{A_1} we have with probability larger than $1-O((np)^{-(B+2)})$, 
\begin{eqnarray*} 
	\left| \wh \theta^\tau_{1, k, \ell}- \wh \theta^{\tau_0}_{1, k, \ell} \right|
	\leq\frac{c_0\tau(\tau_0-\tau)n_{k,\ell}^2}{\tau_0\tau n_{k,\ell}^2} \sqrt{\frac{\log (np)}{(\tau_0-\tau)n_{k,\ell}}}
	\leq  \frac{c_0 \sqrt{\tau_0-\tau}  }{\tau_0 } \sqrt{\frac{\log (np)}{ n_{k,\ell}}},
\end{eqnarray*}
for some constant $c_0>0$. Since $\tau_0\simeq O(n)$, and $n_{k,\ell}\geq s_{\min}^2$,  we conclude that there exists a constant $c_1>0$ such that   
with probability larger than $1-O(np)^{-B})$, 
\begin{align*} 
&\sup_{1\leq k\leq \ell\leq q}\left| \wh \theta^\tau_{1, k, \ell}- \wh \theta^{\tau_0}_{1, k, \ell} \right|
\leq c_1 \sqrt{\frac{ \tau_0-\tau}{\tau_0} }\sqrt{\frac{\log (np)}{ns_{\min}^2}}. 
\end{align*}

For \eqref{appendA2}, note that 
\begin{eqnarray*} 
	&&\log \frac{\wh \theta^\tau_{2, k, \ell}}{\theta^\tau_{2, k, \ell}}- \log \frac{ \wh \theta^{\tau_0}_{2, k, \ell} }{\theta_{2, k, \ell} } \\
	& =& \log  \frac{ \frac{1} {n_{k,\ell}(n-\tau)} \sum_{(i,j)\in   S_{ k,\ell} }\sum_{t=\tau+1}^{n} X_{i,j}^t(1 - X_{i,j}^{t-1})   
	}{ 	\frac{\tau_0-\tau}{n-\tau}\frac{\theta_{1,k,\ell} \eta_{1,k,\ell}}{\theta_{1,k,\ell}+\eta_{1,k,\ell}}+\frac{n-\tau_0 }{n-\tau}\frac{ \eta_{2,k,\ell}\eta_{2,k,\ell}}{\theta_{2,k,\ell}+\eta_{2,k,\ell}} 	 }   -    \log  \frac{  \sum_{(i,j)\in   S_{ k,\ell} }\sum_{t=\tau_0+1}^{n} X_{i,j}^t(1 - X_{i,j}^{t-1})   
	}{ 	n_{k,\ell}(n-\tau_0)\cdot  \frac{ \eta_{2,k,\ell}\eta_{2,k,\ell}}{\theta_{2,k,\ell}+\eta_{2,k,\ell}} 	 }  \\
	&&
	-   
	\log  \frac{ \frac{1} {n_{k,\ell}(n-\tau)} \sum_{(i,j)\in   S_{ k,\ell} }\sum_{t=\tau+1}^{n}  (1 - X_{i,j}^{t-1})   
	}{ 	\frac{\tau_0-\tau}{n-\tau}\frac{ \eta_{1,k,\ell}}{\theta_{1,k,\ell}+\eta_{1,k,\ell}}+\frac{n-\tau_0 }{n-\tau}\frac{ \eta_{2,k,\ell}}{\theta_{2,k,\ell}+\eta_{2,k,\ell}} 	 } 
	+
	\log  \frac{  \sum_{(i,j)\in   S_{ k,\ell} }\sum_{t=\tau_0+1}^{n}  (1 - X_{i,j}^{t-1})   
	}{ 	 n_{k,\ell}(n-\tau_0)\cdot \frac{ \eta_{2,k,\ell} }{\theta_{2,k,\ell}+\eta_{2,k,\ell}} 	 } .
\end{eqnarray*}
It suffices to establish a bound for 
\[
\Delta_{\tau_0,\tau}:=  \frac{ \frac{1} {n_{k,\ell}(n-\tau)} \sum_{(i,j)\in   S_{ k,\ell} }\sum_{t=\tau+1}^{n} X_{i,j}^t(1 - X_{i,j}^{t-1})   
}{ 	\frac{\tau_0-\tau}{n-\tau}\frac{\theta_{1,k,\ell} \eta_{1,k,\ell}}{\theta_{1,k,\ell}+\eta_{1,k,\ell}}+\frac{n-\tau_0 }{n-\tau}\frac{ \eta_{2,k,\ell}\eta_{2,k,\ell}}{\theta_{2,k,\ell}+\eta_{2,k,\ell}} 	 }   -      \frac{  \sum_{(i,j)\in   S_{ k,\ell} }\sum_{t=\tau_0+1}^{n} X_{i,j}^t(1 - X_{i,j}^{t-1})   
}{ 	n_{k,\ell}(n-\tau_0)\cdot  \frac{ \eta_{2,k,\ell}\eta_{2,k,\ell}}{\theta_{2,k,\ell}+\eta_{2,k,\ell}} 	 } . 
\]
Note that for any $B>0$, there exists a large enough constant $B_2$ such that with probability greater than $1-O((np)^{-(B+2)})$, 
\begin{eqnarray*}
	|	\Delta_{\tau_0,\tau} | &\leq & \Bigg| \frac{ \frac{1} {n_{k,\ell}(n-\tau)} \sum_{(i,j)\in   S_{ k,\ell} }\sum_{t=\tau+1}^{\tau_0} \Big[ X_{i,j}^t(1 - X_{i,j}^{t-1})   -  \frac{\theta_{1,k,\ell} \eta_{1,k,\ell}}{\theta_{1,k,\ell}+\eta_{1,k,\ell}}\Big]
	}{ 	\frac{\tau_0-\tau}{n-\tau}\frac{\theta_{1,k,\ell} \eta_{1,k,\ell}}{\theta_{1,k,\ell}+\eta_{1,k,\ell}}+\frac{n-\tau_0 }{n-\tau}\frac{ \eta_{2,k,\ell}\eta_{2,k,\ell}}{\theta_{2,k,\ell}+\eta_{2,k,\ell}} 	 }  \Bigg|  \\
	&&+\Bigg| \left( \frac{ \frac{1} {n_{k,\ell}(n-\tau)}  
	}{ 	\frac{\tau_0-\tau}{n-\tau}\frac{\theta_{1,k,\ell} \eta_{1,k,\ell}}{\theta_{1,k,\ell}+\eta_{1,k,\ell}}+\frac{n-\tau_0 }{n-\tau}\frac{ \eta_{2,k,\ell}\eta_{2,k,\ell}}{\theta_{2,k,\ell}+\eta_{2,k,\ell}} 	 }  -  
	\frac{ \frac{1} {n_{k,\ell}(n-\tau_0)}  
	}{ 	 \frac{ \eta_{2,k,\ell}\eta_{2,k,\ell}}{\theta_{2,k,\ell}+\eta_{2,k,\ell}} 	 } 
	\right)
	\sum_{(i,j)\in   S_{ k,\ell}}  \sum_{t=\tau_0+1}^{n} \Big[X_{i,j}^t(1 - X_{i,j}^{t-1})    \\
	&&- \frac{\theta_{2,k,\ell} \eta_{2,k,\ell}}{\theta_{2,k,\ell}+\eta_{2,k,\ell}} \Big]  \Bigg| \\
	&\leq& B_2 \frac{\tau_0-\tau}{n-\tau}\sqrt{\frac{\log (np)}{ (\tau_0-\tau)n_{k,\ell}}} + B_2\frac{\tau_0-\tau}{n-\tau} \sqrt{\frac{\log (np)}{ (n-\tau_0 )n_{k,\ell}}}.
\end{eqnarray*}
\eqref{appendA2} then follows by noticing that $\frac{\tau_0-\tau}{n-\tau} \sqrt{\frac{\log (np)}{ (n-\tau_0 )n_{k,\ell}}}=o\Big(\frac{\tau_0-\tau}{n-\tau}\sqrt{\frac{\log (np)}{ (\tau_0-\tau)n_{k,\ell}}} \Big)$.

\section*{Appendix B: Real data analysis}

\setcounter{equation}{0}
\renewcommand{\theequation}{B.\arabic{equation}}

\setcounter{subsection}{0}
\renewcommand{\thesubsection}{B.\arabic{subsection}}

\subsection{French high school contact data (cont.)}

The more details of the data analysis in Section \ref{sec52} are
presented below.

We now compare the dynamic stochastic block model method in \cite{mm17} which is implemented in an {\tt R} package {\tt dynsbm}. We note that the model in \cite{mm17} allows the membership probabilities and the transition parameters to vary over time.

We use the {\tt dynsbm} package to analyze the same French high school contact data as we reported in the main text. The function {\tt selection.dynsbm} can automatically select the number of clusters by maximizing the so-called integrated classification likelihood (ICL) criterion. Figure \ref{ICL} shows that the optimal cluster number is selected to be $9$ using ICL for this data set. This agrees with our findings reported in the main text when using the BIC selection criterion.
\begin{figure}[htbp]
	\centering
	\includegraphics[scale=0.99]{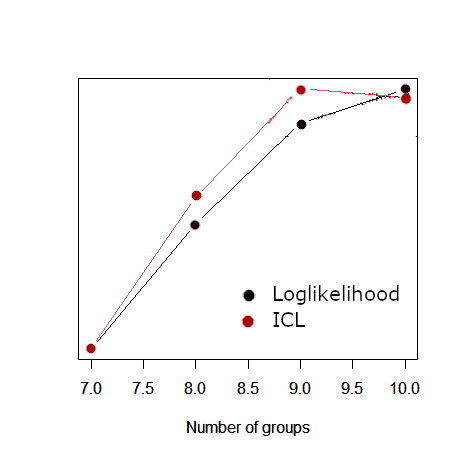}
	\caption{Integrated classification likelihood (ICL) and log-likelihood corresponding to different cluster numbers for French high school data. }
	\label{ICL}
\end{figure}
We then compare the detected clusters from {\tt dynsbm} with the actual class types in Table \ref{table:frenchdyn} at all the five time points.
 	\begin{singlespace}
\begin{table}[htbp]
	{\scriptsize 
		\caption{Clusters for French high school data by using the model in \cite{mm17}.}
		\label{table:frenchdyn}
		\begin{center}
			\begin{tabular}{|ccccccccccc|}
				\hline
				& \multicolumn{10}{c|}{Detected clusters}\\
				\hline
				Class types &  0&  1&  2&  3&  4&  5&  6&  7&  8&  9\\
				\hline
				$t=1$&&&&&&&&&&\\
				BIO1 & 1&  0& 35&  0&  0 & 0 & 0 & 0 & 1&  0\\
				BIO2 & 1&  0 & 0&  0&  0 & 1 & 0 & 0 &31 & 0\\
				BIO3 & 0 & 0&  0&  0&  0&  0& 40&  0&  0&  0\\
				MP1   &  3& 30&  0&  0&  0&  0&  0&  0&  0&  0\\
				MP2 &  3&  0&  0&  0&  0& 26&  0&  0&  0&  0\\
				MP3 &  2&  0&  0& 31&  0&  5&  0&  0&  0&  0\\
				PC1  &  4&  0&  0&  0&  0&  5&  0&  0&  0& 35\\
				PC2  &  0&  0&  0&  0&  0&  1&  0& 38&  0&  0\\
				EGI &  1&  0&  0&  0& 32&  1&  0&  0&  0&  0\\
				\hline
				$t=2$&&&&&&&&&&\\	
				BIO1 & 0&  0& 36&  0&  0&  0&  0&  0&  1&  0\\
				BIO2 & 4&  0&  0&  0&  0&  0&  0 & 0& 29&  0\\
				BIO3 &  1&  0&  0&  0&  0&  0& 39 &  0 & 0&  0\\
				MP1  &   0& 31&  0&  0&  0&  2&  0&  0&  0&  0\\
				MP2&   2&  0&  0&  0&  0& 27&  0&  0&  0&  0\\
				MP3&   1&  0&  0& 34&  0&  3&  0&  0&  0&  0\\
				PC1&     3&  0&  0&  0&  0&  2&  0&  0&  0& 39\\
				PC2&    3&  0&  0&  0&  0&  1&  0& 35&  0&  0\\
				EGI&   3&  0&  0&  0& 31&  0&  0&  0&  0&  0\\
				\hline
				$t=3$&&&&&&&&&&\\	
				BIO1&  5&  0& 31&  0&  0&  0&  0&  0&  1&  0\\
				BIO2&  2&  0&  0&  0&  0&  0&  0&  0& 31&  0\\
				BIO3&  3&  0&  0&  0&  0&  0& 37&  0&  0&  0\\
				MP1&     0& 27&  0&  0&  0&  6&  0&  0&  0&  0\\
				MP2&   1&  0&  0&  0&  0& 28&  0&  0&  0&  0\\
				MP3&   0&  0&  0& 36&  0&  2&  0&  0&  0&  0\\
				PC1&     6&  0&  0&  0&  0&  2&  0&  0&  0& 36\\
				PC2&    4&  0&  0&  0&  0&  0&  0& 35&  0&  0\\
				EGI&   3&  0&  0&  0& 30&  1&  0&  0&  0&  0\\
				\hline
				$t=4$&&&&&&&&&&\\	
				BIO1&  2&  0& 35&  0&  0&  0&  0&  0&  0&  0\\
				BIO2&  3&  0&  0&  0&  0&  0&  0&  0& 30&  0\\
				BIO3&  7&  0&  0&  0&  0&  0& 33&  0&  0&  0\\
				MP1&     1& 28&  0&  0&  0&  4&  0&  0&  0&  0\\
				MP2&   1&  1&  0&  0&  0& 27&  0&  0&  0&  0\\
				MP3&   4&  0&  0& 34&  0&  0&  0&  0&  0&  0\\
				PC1 &     4&   0&  0&  0&  0&  3&  0&  0&  0& 37\\
				PC2&    4&  0&  0&  0&  0&  0&  0& 35&  0&  0\\
				EGI&   6&  0&  0&  0& 26&  2&  0&  0&  0&  0\\
				\hline
				$t=5$&&&&&&&&&&\\	
				BIO1&  3&  0& 33&  0&  0&  0&  0&  0&  1&  0\\
				BIO2&  2&  0&  0&  0&  0&  0&  0&  0& 31&  0\\
				BIO3&  5&  0&  0&  0&  0&  0& 35&  0&  0&  0\\
				MP1&     1& 26&  0&  0&  0&  6&  0&  0&  0&  0\\
				MP2&   4&  0&  0&  0&  0& 25&  0&  0&  0&  0\\
				MP3&   3&  0&  0& 33&  0&  2&  0&  0&  0&  0\\
				PC1&     4&  0&  0&  0&  0&  2&  0&  0&  0& 38\\
				PC2&    3&  0&  0&  0&  0&  0&  0& 36&  0&  0\\
				EGI&   3&  0&  0&  0& 30&  1&  0&  0&  0&  0\\
				\hline	
			\end{tabular}
		\end{center}
	}
\end{table}
 	\end{singlespace}

We may notice that {\tt dynsbm} method reserves one group as ``0" for subjects with no edges (the absence nodes). Our algorithm, in comparison, is not affected by those subjects.

 	\begin{singlespace}
\begin{table}[htbp]
	{\scriptsize 
		\caption{Detected clusters for the French high school data by using our method.}
		\label{table:french22}
		\begin{center}
			\begin{tabular}{|cccccccccc|}
				\hline
				& \multicolumn{9}{c|}{Detected clusters}\\
				\hline
				Class types &    1&  2&  3&  4&  5&  6&  7&  8&  9\\
				\hline
				BIO1 & 0&  0&  1&  0&  0&  0&  0&  0& 36\\
				BIO2&  0&  1& 32&  0&  0&  0&  0&  0&  0\\
				BIO3&  0&  1&  0&  0& 39&  0&  0&  0&  0\\
				MP1&    33&  0&  0&  0&  0&  0&  0&  0&  0\\
				MP2 &  0&  1&  0& 28&  0&  0&  0&  0&  0\\
				MP3 &  0&  0&  0&  0&  0& 38&  0&  0&  0\\
				PC1 &    0& 44&  0&  0&  0&  0&  0&  0&  0\\
				PC2&    0&  0&  0&  0&  0&  0&  0& 39&  0\\
				EGI &  0&  0&  0&  0&  0&  0& 34&  0&  0\\
				
				\hline
			\end{tabular}
		\end{center}
	}
\end{table}		
 	\end{singlespace}		
The grouping results from {\tt dynsbm} method is quite stable over the five time points. Such results lend support to our method which assumes the constant cluster structure over time. Furthermore, our clustering results, shown in Table \ref{table:french22}, appear to be more accurate and agree more closely to the true grouping (class types) for this data analysis.

Another practical advantage of our method is its relatively short computing time. Using a computer with Intel(R) Core(TM) i7-10875H CPU and 32.0 GB RAM, we need to spend 0.36 and 252.39 seconds to obtain the community detection results with our method and {\tt dynsbm} method, respectively.
\subsection{Enron email data}
The Enron email dataset contains approximately 500,000 emails generated by employees of the Enron Corporation. It was obtained by the Federal Energy Regulatory Commission during its investigation of Enron's collapse. The data file was published at \url{https://www.cs.cmu.edu/~./enron/}.

\cite{rastelli2017} developed a latent stochastic block model for directed dynamic network data. They applied their methods for the Enron email data. To compare with their results, we used data for all the emails exchanged from January 2000 to March 2002 ($n=27$) between the Enron members. The number of nodes in our analysis is $184$ (which is different from \cite{rastelli2017} where they have kept $148$ subjects).

Using our spectral clustering algorithm for the whole network data and using the BIC, we obtain the best number of cluster is $13$. \cite{rastelli2017} used the exact ICL criterion and found $17$ clusters but $4$ groups seem to be extremely small or just contain inactive nodes. Their algorithm is for directed graph and also assume time-varying membership. It seems our results are still quite close to \cite{rastelli2017}. Similar to their analysis, we can see that there is a high degree of heterogeneity with $13$ different clusters for this data set. When applying {\tt dynsbm} package on this data the optimal number of cluster is only $6$. 

Next we considered the change point analysis. Using our binary segmentation methods, we detect one change point at October 2001 which is exactly the month corresponding to the disclosure of Enron bankruptcy (see the log-likelihood functions plotted in Figure \ref{ENRON}). This again agrees with the empirical findings in \cite{rastelli2017}.

\begin{figure}[htbp]
	\centering
	\includegraphics[scale=0.99]{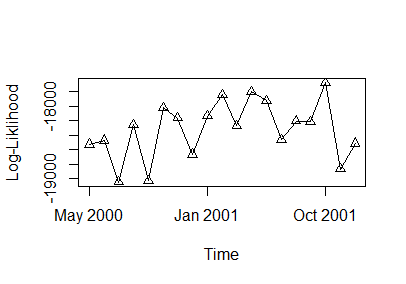}
	\caption{Log-likelihood functions corresponding to different change points in time for Enron email data. }
	\label{ENRON}
\end{figure}


\end{document}